\documentclass[12pt,a4paper,notitlepage, reqno]{amsart}
\usepackage[utf8]{inputenc}
\usepackage[english]{babel}

\usepackage{bm,bbm}
\usepackage{amsbsy,amsmath,amsthm, amssymb,amsfonts}
\usepackage{graphicx}
\usepackage{subfigure}
\usepackage{verbatim}
\usepackage{accents}
\usepackage{stmaryrd}

\usepackage[foot]{amsaddr}

\usepackage{ifthen}

\usepackage{adjustbox}
\usepackage{float}
\usepackage{varwidth}

\usepackage[	pdfauthor={},%
			pdftitle={},%
			pdftex]{hyperref}
\hypersetup{	colorlinks,%
			citecolor=black,%
			filecolor=black,%
			linkcolor=black,%
			urlcolor=black%
			}

        {%
            \ifthenelse{\isundefined{\showunderconstruct}}%
                    {\expandafter\comment}%
                    {\clearpage\noindent \framebox{$\downarrow$UNDER CONSTRUCTION$\downarrow$}}%
                    }%
         {%
            \ifthenelse{\isundefined{\showunderconstruct}}%
                    {\expandafter\endcomment}%
                    {\noindent \framebox{$\uparrow$UNDER CONSTRUCTION$\uparrow$}\clearpage}%
          }

        {%
            \ifthenelse{\isundefined{\showsectioncomments}}%
                    {\expandafter\comment}%
                    {\noindent$\longrightarrow$\itshape{}}%
                    }%
         {%
            \ifthenelse{\isundefined{\showsectioncomments}}%
                    {\expandafter\endcomment}%
                    {\par}%
          }
          
\newcommand{\assumpa}{\hypertarget{href: assumpa}{$(\ast)$}}          
\newcommand{\assumparef}{\hyperlink{href: assumpa}{$(\ast)$}}     

\newcommand{\assumpb}{\hypertarget{href: assumpb}{$(\ast')$}}          
\newcommand{\assumpbref}{\hyperlink{href: assumpb}{$(\ast')$}}     

\newcommand{\assumpc}{\hypertarget{href: assumpc}{$(\ast\ast)$}}          
\newcommand{\assumpcref}{\hyperlink{href: assumpc}{$(\ast\ast)$}}

\newcommand{\ddiamond}{\spadesuit}

\newcommand{\Z}{\mathbb{Z}}
\newcommand{\N}{\mathbb{N}}
\newcommand{\C}{\mathbb{C}}
\newcommand{\R}{\mathbb{R}}

\newcommand{\half}{\mathbb{H}}
\newcommand{\disc}{\mathbb{D}}

\newcommand{\ii}{\mathrm{i}}

\renewcommand{\P}{\mathbb{P}}

\newcommand{\E}{\mathbb{E}}

\newcommand{\ind}{\mathbbm{1}}

\newcommand{\F}{\mathcal{F}}

\newcommand{\OO}{\mathcal{O}}

\newcommand{\de}{\mathrm{d}}

\DeclareMathOperator{\dist}{dist}

\newcommand{\ffrac}[2]{\left(\frac{#1}{#2}\right)}

\newcommand{\eps}{\varepsilon}
\DeclareMathOperator{\imag}{Im}
\DeclareMathOperator{\real}{Re}

\newcommand{\dd}{\mathrel{\mathop:}}

\newcommand{\cd}{\partial}
\newcommand{\cdbar}{\bar{\partial}}

\DeclareMathOperator{\Proj}{Proj}

\newcommand{\varggr}{\Gamma}

\newcommand{\outeredge}{\extarc{\medge_o}{\medge_i}}
\newcommand{\inneredge}{\extarc{\varmedge_o}{\varmedge_i}}

\newcommand{\Gouter}{G_{\outeredge}}
\newcommand{\Ginner}{G_{\inneredge}}

\newcommand{\coeffa}{\alpha}
\newcommand{\coeffb}{\beta}

\newcommand{\argu}{\upsilon}
\newcommand{\argv}{\phi}

\newcommand{\arguu}{\Upsilon}
\newcommand{\argvv}{\Phi}

\newcommand{\auxvara}{\rho}
\newcommand{\auxvarb}{\mu}

\newcommand{\confmap}{\Psi}  
\newcommand{\varconfmap}{\psi}

\newcommand{\medge}{e}  
\newcommand{\varmedge}{f}

\newcommand{\vroot}{v_\text{root}}
\newcommand{\tree}{\mathcal{T}}

\newcommand{\lattice}{\mathbb{L}}
\newcommand{\sqlattice}{\lattice^\bullet}
\newcommand{\dsqlattice}{\lattice^\circ}
\newcommand{\dmdlattice}{\lattice^\diamond}
\newcommand{\ddmdlattice}{\lattice^\ddiamond}

\newcommand{\orient}{\rightarrow}

\newcommand{\odmdlattice}{\dmdlattice_\orient}
\newcommand{\oddmdlattice}{\ddmdlattice_\orient}

\newcommand{\LoopEnseble}{\mathcal{L}}
\newcommand{\Loop}{L}
\newcommand{\RandomLoopEnseble}{\Theta}
\newcommand{\RandomLoop}{\theta}

\newcommand{\lpe}{\LoopEnseble}
\newcommand{\lp}{\Loop}
\newcommand{\rndlpe}{\RandomLoopEnseble}
\newcommand{\rndlp}{\RandomLoop}

\newcommand{\Tree}{\mathcal{T}}
\newcommand{\Branch}{T}
\newcommand{\RandomTree}{\Tree}
\newcommand{\RandomBranch}{\Branch}

\newcommand{\ttr}{\Tree}
\newcommand{\bran}{\Branch}
\newcommand{\rndttr}{\RandomTree}
\newcommand{\rndbran}{\RandomBranch}

\newcommand{\therndttr}{\RandomTree}
\newcommand{\therndbran}{\RandomBranch}

\newcommand{\delp}{\de_{\textnormal{loop}}}
\newcommand{\delpe}{\de_{\textnormal{LE}}}
\newcommand{\dettr}{\de_{\textnormal{tree}}}
\newcommand{\decc}{\de_{\textnormal{curve}}}

\newcommand{\domain}{\Omega}
\newcommand{\vardomain}{\domain}

\newcommand{\unf}[2]{#1^\textnormal{u}_#2}

\newcommand{\extarc}[2]{#1 \smile #2}

\newcommand{\intarcp}[4]{(#1 \frown #2, #3 \frown #4)}

\newcommand{\todisc}[1]{#1_\disc}

\newcommand{\sle}[1]{SLE$(#1)${}}

\newcommand{\cle}[1]{CLE$(#1)${}}

\newcommand{\slek}{\sle{\kappa}}
\newcommand{\clek}{\cle{\kappa}}

\newcommand{\thekr}{\sle{{\frac{16}{3},-\frac{2}{3}}}}

\newcommand{\Znneg}{\Z_{\geq 0}}

\newcommand{\Ztime}{\Znneg}

\newcommand{\Rtime}{[0,\infty)}

\newcommand{\Vtarget}{V_\textnormal{mid}}

\newcommand{\inz}[1]{{#1 \in \Ztime}}
\newcommand{\tinz}{\inz{t}}

\theoremstyle{plain}
\newtheorem{theorem}{Theorem}[section]
\newtheorem{lemma}[theorem]{Lemma}
\newtheorem{proposition}[theorem]{Proposition}
\newtheorem{corollary}[theorem]{Corollary}

\theoremstyle{definition}
\newtheorem{definition}[theorem]{Definition}

\newtheorem*{condition*}{Condition~G}

\theoremstyle{remark}
\newtheorem{remark}[theorem]{Remark}

\usepackage[hmargin=3cm,vmargin={2.5cm,2cm},includefoot]{geometry}

\newcommand{\showsectioncomments}{show them}  

\newcommand{\refcond}{Condition~G}

\newcommand{\LemmaReverseLoewner}{Lemma~5.6 of \cite{Kemppainen:tb}}
\newcommand{\LemmaReverseLoewnerLip}{Lemma~6.1 of \cite{Kemppainen:tb}}
\newcommand{\PropositionLoewnerLip}{Proposition~6.1 of \cite{Kemppainen:tb}}

\begin{document}


\title[Conformal invariance in random cluster models. II.]{%
  Conformal invariance in random cluster models. II. Full scaling limit as a branching SLE.}


\author{Antti Kemppainen$^1$}
\address{$^1$Department of Mathematics and Statistics\\
         Faculty of Science\\
         P.O. Box 68\\
         FIN-00014 University of Helsinki\\
         Finland}
\email{Antti.H.Kemppainen@helsinki.fi}
\author{Stanislav Smirnov$^2$}
\address{$^2$Section de math\'ematiques,
         Universit\'e de Gen\`eve,
         2-4, rue du Li\`evre, c.p. 64,
         1211 Gen\`eve 4, Switzerland, \textnormal{and} 
Skolkovo Institute of Science and Technology, Russia, \textnormal{and}
Chebyshev Laboratory, St.~Petersburg State University, Russia  
}
\email{Stanislav.Smirnov@unige.ch}

%
\date{}
\begin{abstract}
In the second article of this series, we establish the convergence of the loop ensemble 
of interfaces in the random cluster Ising model to a conformal loop ensemble (CLE)
--- thus completely describing the scaling limit of the model in terms of the random geometry
of interfaces. 
The central tool of the present article 
is the convergence of an exploration tree of the discrete loop ensemble
to a family of branching \thekr{} curves.
Such branching version of the Schramm's SLE not only enjoys the locality property, but
also arises logically from the Ising model observables.
\end{abstract}
\maketitle



\section{Introduction}

Starting with the introduction of the Lenz-Ising model of ferromagnetism,
lattice models of natural phenomena played important part in modern mathematics and physics.
While overly simplified 
--- e.g. restricted to a regular lattice of spins with nearest-neighbor interaction ---
they often give a very accurate qualitative description of what we observe in nature.
In particular, they exhibit phase transitions when temperature passes through 
the critical point, 
and the critical system
is expected to enjoy (in the scaling limit) universality and (at least in 2D) conformal invariance.

While this is well understood on the physical and computational level, mathematical proofs (and understanding) are often lacking.
In the first paper of this series \cite{Smirnov:2010ie}, one of us established conformal invariance of some observables in the FK Ising model at criticality, from which description of 
the scaling limit of a single interface 
as a universal, conformally invariant, fractal curve was -- the so-called Schramm's \sle{16/3} -- was deduced 
\cite{Chelkak:2014gs}.
The mathematical theory of such curves
was started by Oded Schramm in his seminal paper \cite{Schramm:2000th}.
The SLE curves are obtained by running a Loewner evolution with 
a speed $\kappa$ Brownian driving term,
and form a one-parameter family of fractals, interesting in themselves
\cite{Rohde:2005fy,Lawler:2008wf}. 
Schramm has shown, that all scaling limits of interfaces or domain walls, if they exist and are conformally invariant, are always described by SLEs; for the exact formulation of the principle, see
\cite{Schramm:2000th,Smirnov:2007tp,Kemppainen:2010vu,Kemppainen:tb}.
A generalization of SLE is the conformal loop ensemble (CLE), which describes the joint
law of all the interfaces in a model.

So far, convergence of a single discrete interface  to \slek's has been established for but 
a few models corresponding to special values of the parameter:
 $\kappa=2$ and $\kappa=8$ \cite{Lawler:2004un},
$\kappa=3$ and $\kappa=\frac{16}{3}$ \cite{Chelkak:2014gs},
$\kappa=4$ \cite{Schramm:2005wh,Schramm:2006uc} and
$\kappa=6$ \cite{Smirnov:2001hw,Smirnov:2009vj}.
However, the framework for the full scaling limit, including all interfaces, is less developed: 
$\kappa=3$ \cite{Benoist:2016uy}, $\kappa=\frac{16}{3}$ \cite{Kemppainen:2015vu}
and $\kappa=6$ \cite{Camia:2006wv}.

The present article extends the convergence showed in \cite{Kemppainen:2015vu}
to include all the interfaces, not just those (infinitely many in the limit)
that touch the boundary.
Effectively, we give a geometric description of the full scaling limit of the FK Ising model, 
which is universal, conformally invariant, and can be obtained
by a canonical coupling of branching SLE curves.

\subsection{The main theorem}

As we indicated above, we consider in this article the Fortuin--Kasteleyn model 
with the cluster parameter $q=2$ (and hence call it the FK Ising model) 
and when the percolation parameter $p$ equals to its critical value
$\sqrt{2}/(1+\sqrt{2})$
(and we say that the model is at criticality).
We study, more specifically, the loop representation of the model. 
A loop configuration of 
the FK Ising model
is illustrated in Figure~\ref{sfig: loop config and exploration a}. 
The underlying Fortuin--Kasteleyn random cluster model configuration
lives on the edges connecting neighboring gray octagons
(the edges form a regular square lattice), while the loop configuration
is a dense collection of non-intersecting loops on the square--octagon lattice.
The probability of the loop configuration of the critical FK Ising
model is proportional to $\sqrt{2}^{(\# \textnormal{ of loops})}$.

\begin{figure}[tbh]
\centering
\subfigure[A random cluster configuration (with wired boundary)
shown here as the red subcollection of the edges of the square lattice
and its loop configuration shown here as
the blue collection of closed loops which are dense,
in the sense that each horizontal or vertical edge of the square-octagon
lattice is traversed by a loop (exactly one to be more accurate).] 
{
	\label{sfig: loop config and exploration a}
	\includegraphics[scale=.375]
{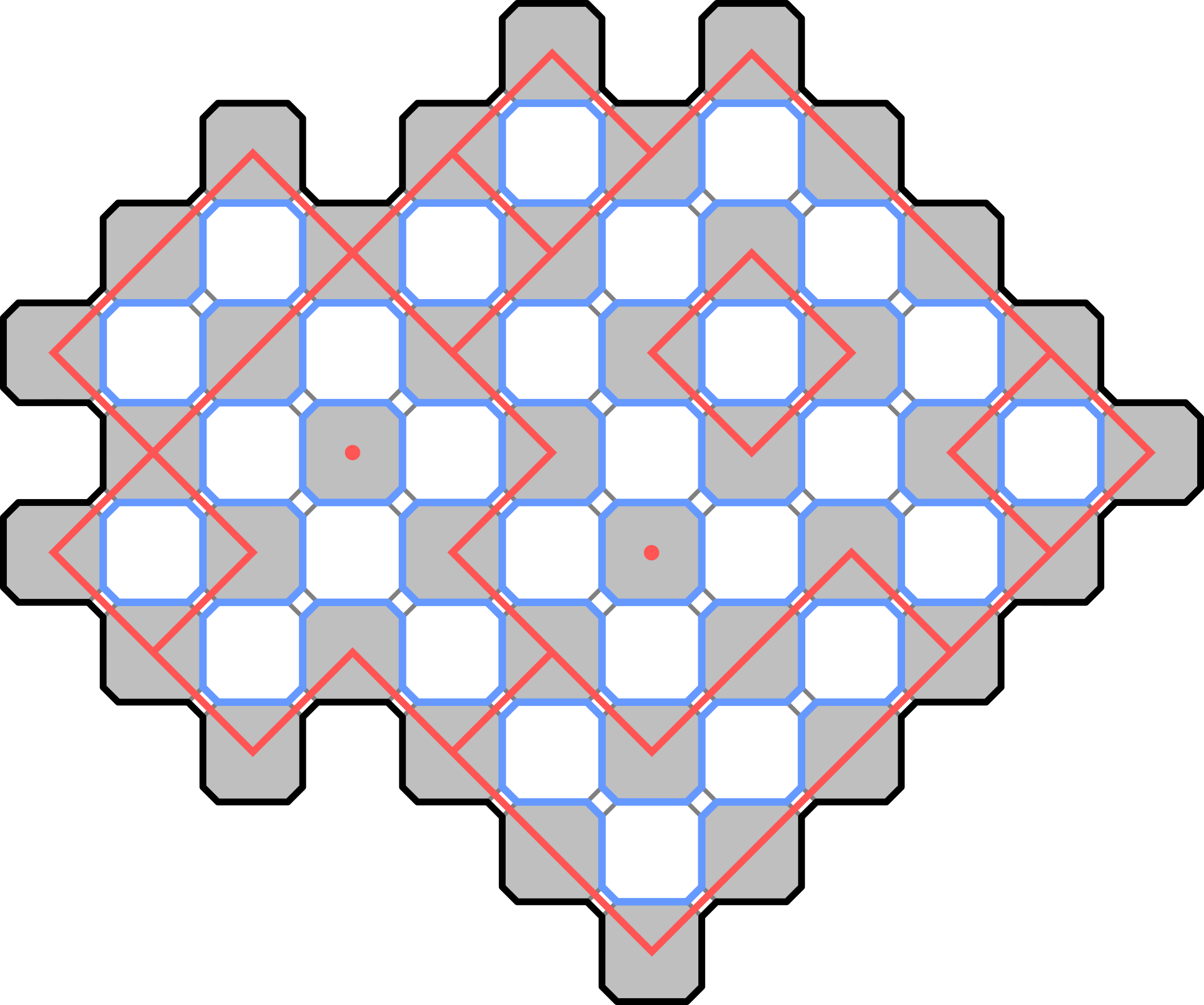}
} 
\hspace{0.5cm}
\subfigure[First steps of exploration: the exploration process
 starts to follow
the dark blue path at the root point
and it branches to green, yellow and orange paths, 
in that order, upon each disconnection event of the exploration region in to two 
disconnected halves.]
{
	\label{sfig: loop config and exploration b}
	\includegraphics[scale=.375]
{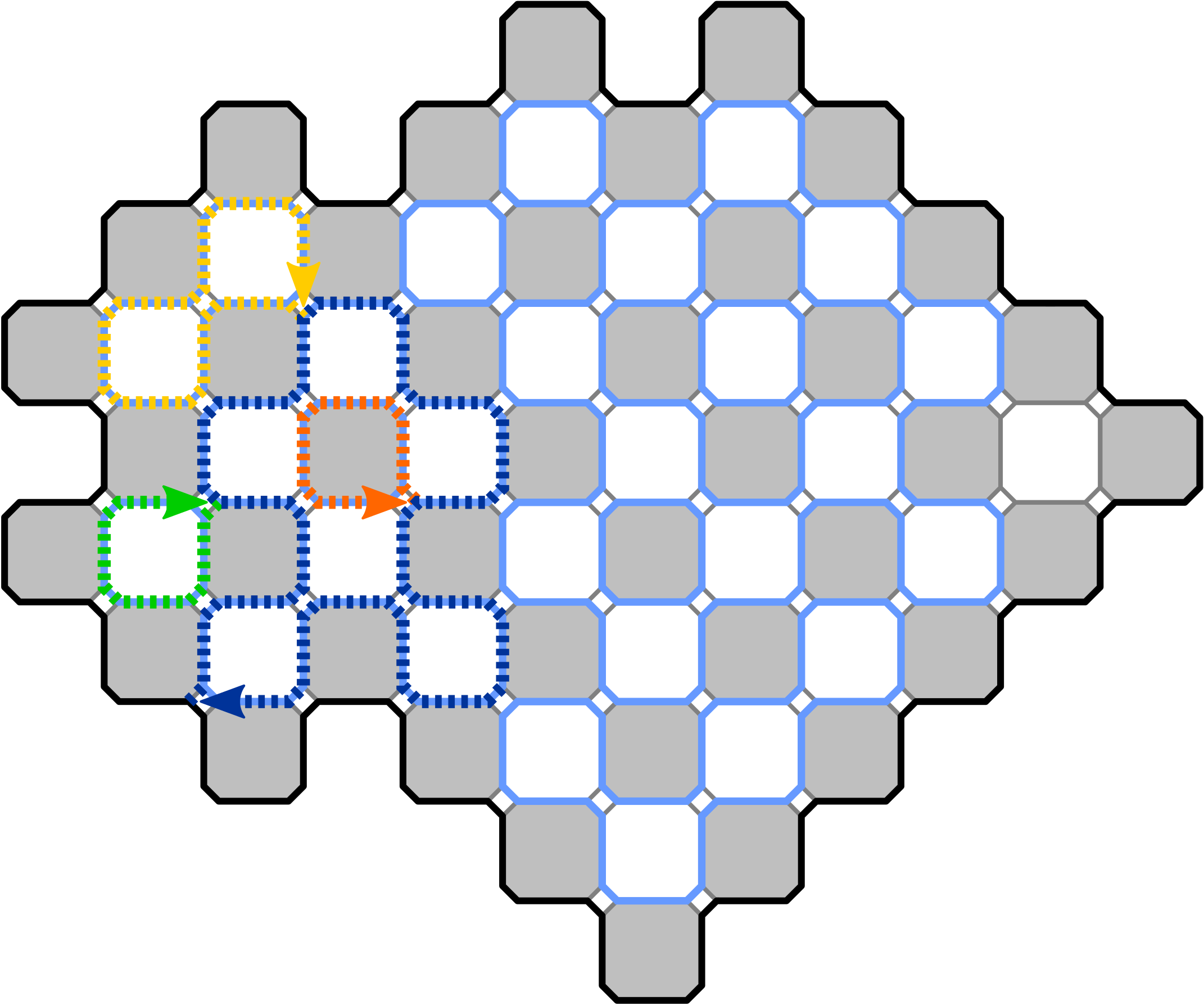}
}
\caption{The loop configuration and the exploration process
of a random cluster model model}
\label{fig: loop config and exploration}
\end{figure}

Figure~\ref{sfig: loop config and exploration b}
 illustrates the first couple of steps of the \emph{exploration process} which is defined
 to be a collection of paths, each starting from a fixed root and then exploring loops in clockwise direction and splitting to two paths whenever a disconnection of an area in to two 
lattice-disconnected halves occurs. The exploration paths form a tree and mathematically the tree
consists of paths from the root to any other lattice point on the domain
(the other end point of the branch is called the target point).

The following theorem is the main theorem of this article
establishing the convergence of 
the FK Ising loop ensemble to \cle{16/3}. We consider a sequence
of discrete domains $\domain_\delta$, where $\delta>0$ is the lattice mesh, 
converging to a domain $\domain$, and a sequence of root points $a_\delta$
converging to a boundary point of $\domain$. The topology of such convergence
is discussed below.

\begin{theorem}\label{thm: main thm}
The joint law of the FK Ising loop ensemble in a discrete domain $\domain_\delta$
and its exploration tree 
(rooted at $a_\delta$)
converges in distribution 
to the joint law of CLE$(\kappa)$ and its SLE$(\kappa,\kappa-6)$ exploration tree
with $\kappa=16/3$ in the topology described below.
\end{theorem}

Below we deepen the definitions required to understand more exactly the main theorem and 
the main tools to give its proof.

\subsection{Fortuin--Kasteleyn representation of the Ising model}

For general background on the Ising model, the random cluster model and other models of statistical physics, 
see the books \cite{Baxter:1989ug,Grimmett:2006jx,Grimmett:2010tt,mccoy2014two}.
See also the first article \cite{Smirnov:2010ie}, Section~2.

\subsubsection{Notation and definitions for graphs}

In this article, the lattice $\sqlattice$ is the \emph{square lattice} $\Z^2$ rotated by $\pi/4$
and scaled by $\sqrt{2}$, 
$\dsqlattice$ is its dual lattice, which itself is also a square lattice, and $\dmdlattice$ is their (common) \emph{medial lattice}. 
More specifically, we define three lattices $G=(V(G),E(G))$, where $G= \sqlattice,\dsqlattice,\dmdlattice$, as
\begin{gather*}
V( \sqlattice ) = \left\{ (i,j) \in \Z^2 \,:\, i+j \text{ even} \right\}, \quad
   E( \sqlattice ) = \left\{ \{v,w\} \subset V( \sqlattice ) \,:\, |v-w| = \sqrt{2} \right\},  \\
V( \dsqlattice ) = \left\{ (i,j) \in \Z^2 \,:\, i+j \text{ odd} \right\}, \quad
   E( \dsqlattice ) = \left\{ \{v,w\} \subset V( \dsqlattice ) \,:\, |v-w| = \sqrt{2} \right\},  \\
V( \dmdlattice ) = (1/2 + \Z)^2 , \quad
   E( \dmdlattice ) = \left\{ \{v,w\} \subset V( \dmdlattice ) \,:\, |v-w| = 1\right\} 
\end{gather*}
Here $V$ stands for vertices (or sites) and $E$ for edges.
Notice that sites of $\dmdlattice$ are the midpoints of the edges of $\sqlattice$ and $\dsqlattice$.

Denote the set of midpoints of the edges of $\dmdlattice$ as
\begin{equation}\label{eq: def midpoints}
V_\textnormal{mid} = \left\{ (i,j) \in \left(\frac{1}{2} \Z\right)^2 \,:\, i+j \in \Z + \frac{1}{2} 
   \right\} .
\end{equation}
It is natural to identify midpoints $V_\textnormal{mid}$ with their corresponding edges $E(\dmdlattice )$.

We call the vertices and edges of $V( \sqlattice )$ \emph{black} and
the vertices and edges of $V( \dsqlattice )$ \emph{white}.
Correspondingly the faces of $\dmdlattice$ are colored black and white depending whether the
center of that face belongs to $V( \sqlattice )$ or $V( \dsqlattice )$.

The directed version $\odmdlattice$ is defined by setting $V(\odmdlattice)=V(\dmdlattice)$
and orienting the edges around any black face in the counter-clockwise direction.

\begin{figure}[tbh]
\centering
\subfigure[Lattices $\sqlattice, \dsqlattice, \dmdlattice$.] 
{
	\label{sfig: graphs a}
	\includegraphics[scale=.6]
{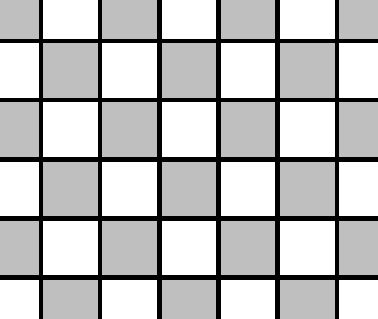}
} 
\hspace{0.5cm}
\subfigure[Modification and $\ddmdlattice$.]
{
	\label{sfig: graphs b}
	\includegraphics[scale=.6]
{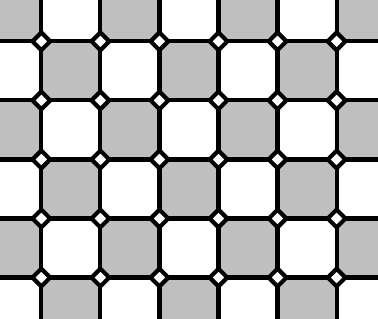}
}
\caption{The square lattices we are considering are $\sqlattice$ formed by the centers of the black squares,
$\dsqlattice$ formed by the centers of the white squares and $\dmdlattice$ formed by the corners of the black and white squares.
We will also consider the square--octagon lattice $\ddmdlattice$ which we see as a modification of $\dmdlattice$.}
\label{fig: graphs}
\end{figure}

The \emph{modified medial lattice} $\ddmdlattice$, which is a square--octagon lattice, 
is obtained from $\dmdlattice$ by replacing each site by a small square.
See Figure~\ref{fig: graphs}.
The faces of $\ddmdlattice$ are refered to 
as \emph{octagons (black or white)} and \emph{small squares}.
The oriented lattice $\oddmdlattice$
is obtained from $\ddmdlattice$ by orienting the edges around black and white octagonal faces
in counter-clockwise and clockwise directions, respectively.

\begin{definition}
A simply connected, non-empty, bounded domain $\domain$ is said to be a \emph{wired $\oddmdlattice$-domain} (or \emph{admissible domain})
if $\partial \domain$ oriented in counter-clockwise 
direction is a path on the directed lattice $\oddmdlattice$.
\end{definition}

\begin{figure}[tbh]
\centering
\subfigure[Oriented lattice $\oddmdlattice$.] 
{
	\label{sfig: orientation and domain a}
	\includegraphics[scale=.2]
{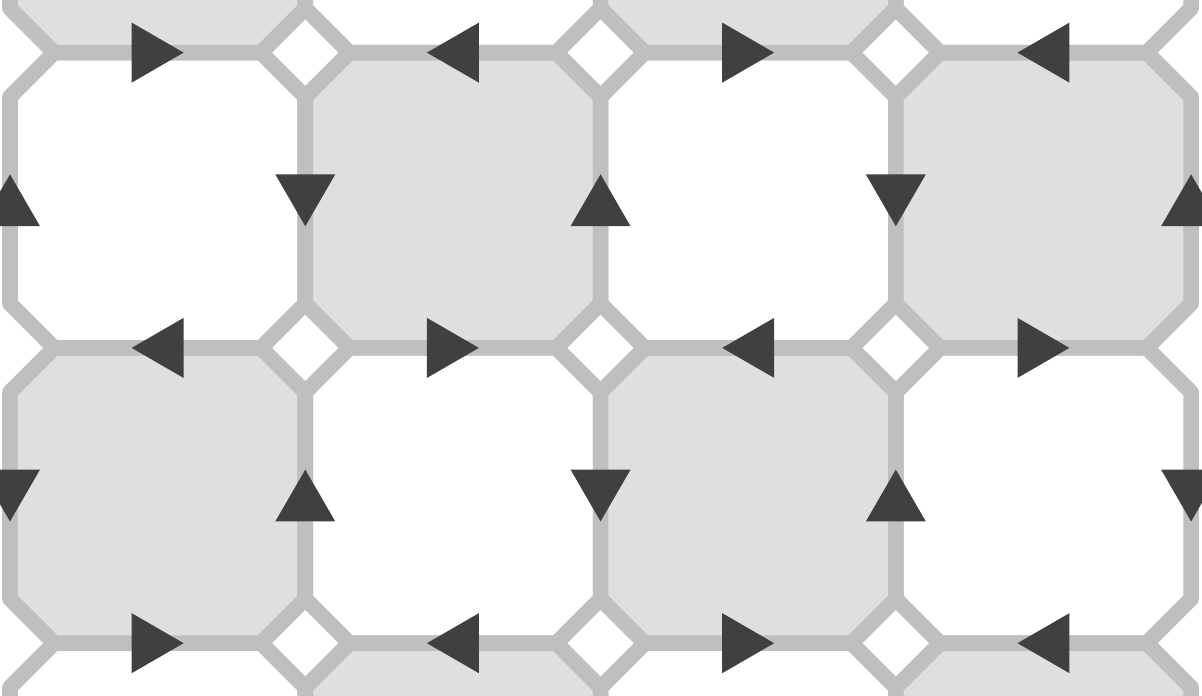}
} 
\hspace{0.5cm}
\subfigure[Wired $\oddmdlattice$-domain.]
{
	\label{sfig: orientation and domain b}
	\includegraphics[scale=.3]
{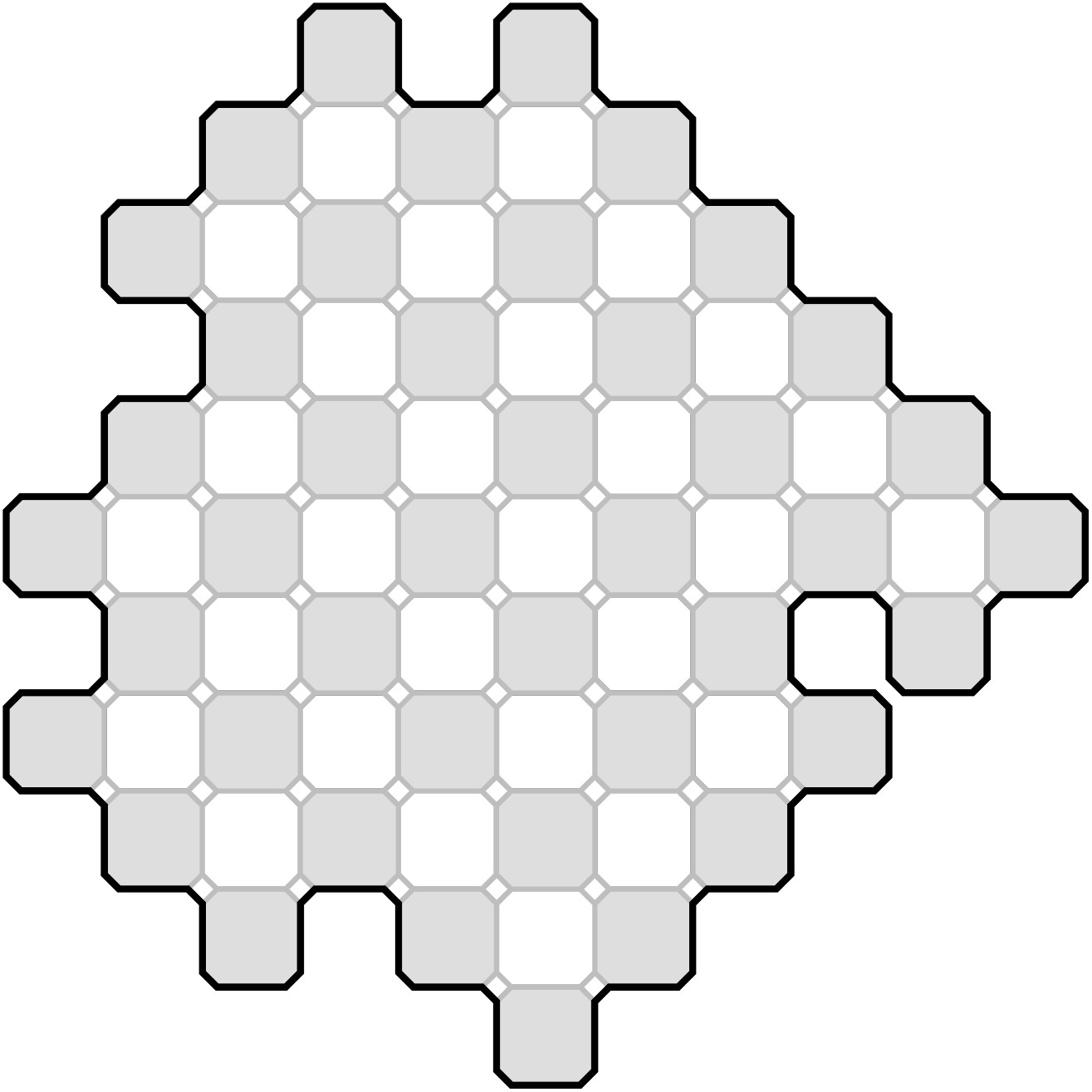}
}
\subfigure[Loop representation of the random cluster configurations by drawing two quarters of octagons
next to open and dual-open edges.]
{
	\label{sfig: orientation and domain c}
	\includegraphics[scale=.2]
{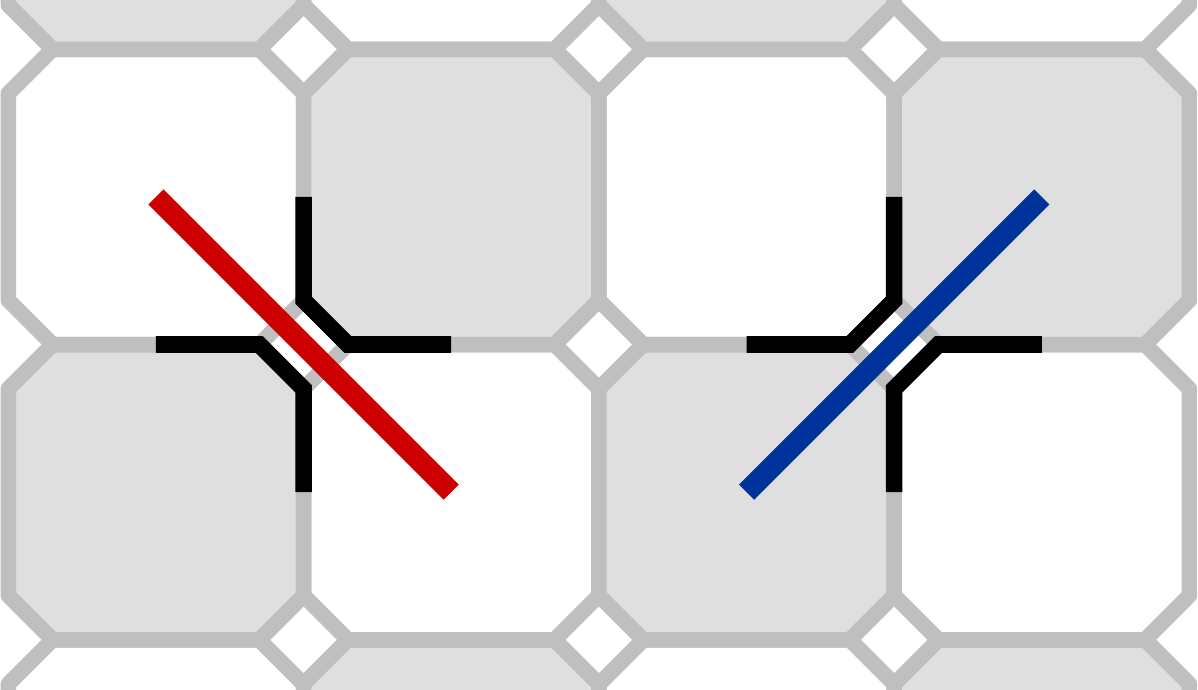}
}
\caption{The oriented lattice and a discrete admissible domain on it.}
\label{fig: orientation and domain}
\end{figure}

See Figure~\ref{fig: orientation and domain} for an example of such a domain.
The wired $\oddmdlattice$-domains are in one to one correspondence with non-empty finite subgraphs of $\sqlattice$
which are simply connected, i.e., they are graphs who have an unique unbounded face 
and the rest of the faces are squares.

\subsubsection{FK Ising model}

We consider the Fortuin--Kastelleyn model, also known as the random cluster model. 
On a finite graph $G$, the state of the model is a collection of open edges (the complementary set of edges are said to be closed) and
the probability of a state
is proportional to $q^{\# \text{ of clusters}} \, (p/(1-p))^{\# \text{ of open}}$, 
that is, 
the  probability given by the edge percolation model,
which gives weight $p$ to open edges 
and $1-p$ to closed edges, is weighted by $q$ per each open connected cluster on the graph.
Denote this probability measure by $\mu_{p,q}^1$.

Let $G$ be a simply connected subgraph of the square lattice $\sqlattice$
corresponding to a wired $\oddmdlattice$-domain. 
Consider the random cluster measure
$\mu=\mu_{p,q}^1$ of $G$ with all \emph{wired boundary conditions}
(the measure is conditioned assuming that all the edges of $G$ next to the boundary are open) 
in the special case of the critical FK Ising model,
that is, when $q=2$ and $p = \sqrt{2}/(1 + \sqrt{2})$. 
Also the planar dual of the random cluster configuration is distributed according to  
a critical FK Ising model. This dual model is defined on the dual graph $G^\circ$ of $G$
which is the (simply connected) subgraph of $\dsqlattice$ containing all vertices and edges
fully contained the domain, and 
satisfies free boundary conditions.
They have common loop representation on the corresponding
subgraph $G^\ddiamond$ of the modified medial
lattice $\ddmdlattice$. The loops are surrounding each open cluster both from its outside and from inside any of its holes.
Equivalently a quarter of an octagon is drawn on both sides of any open or dual-open edge,
see Figure~\ref{fig: orientation and domain}.

We call a collection of loops $\mathcal{L}=(L_j)_{j =1 ,\ldots  N}$ 
on $G^\ddiamond$ \emph{dense collection of non-intersecting loops} (DCNIL)
if 
\begin{itemize}
\item each $L_j \subset G^\ddiamond$ is a simple loop
\item $L_j$ and $L_k$ are vertex-disjoint when $j \neq k$
\item for every edge $e \in E^\diamond$ 
there is a loop $L_j$ that visits $e$. Here we use the fact that 
$E^\diamond$ is naturally a subset of $E^\ddiamond$.
\end{itemize}
We consider loop collections only modulo permutations.
Let the collection of all the loops in the loop representation be $\rndlpe=(\rndlp_j)_{j =1 ,\ldots  N}$. 
Then the set of all DCNIL is exactly the support of $\rndlpe$ and
for any DCNIL collection $\lpe$ of loops 
\begin{equation}\label{eq: def fk ising loop rep}
\mu(\rndlpe = \lpe) = \frac{1}{Z} (\sqrt{2})^{\text{\# of loops in }\lpe}
\end{equation}
where $Z$ is a normalizing constant.
In this paper, we mostly consider the loop representation \eqref{eq: def fk ising loop rep}
which could have been taken as the main definition of the FK Ising model.
Occasionally we refer the underlying random cluster configurations and use the property
that the loops separate open and dual-open clusters.

\subsection{Exploration tree of a loop ensemble}\label{ssec: intro exploration tree}

Suppose that we are given a wired $\oddmdlattice$-domain $\domain$ and
a DCNIL loop collection $\lpe=(\lp_j)_{j =1 ,\ldots  N}$ on $\domain$. We wish
to define a spanning tree which corresponds to $\lpe$ in a one-to-one manner
with an easy rule to recover $\lpe$ from the spanning tree.
We follow here the ideas of \cite{Sheffield:2009}.

\begin{figure}[tbh]
\centering
 	{\includegraphics[scale=.3]
{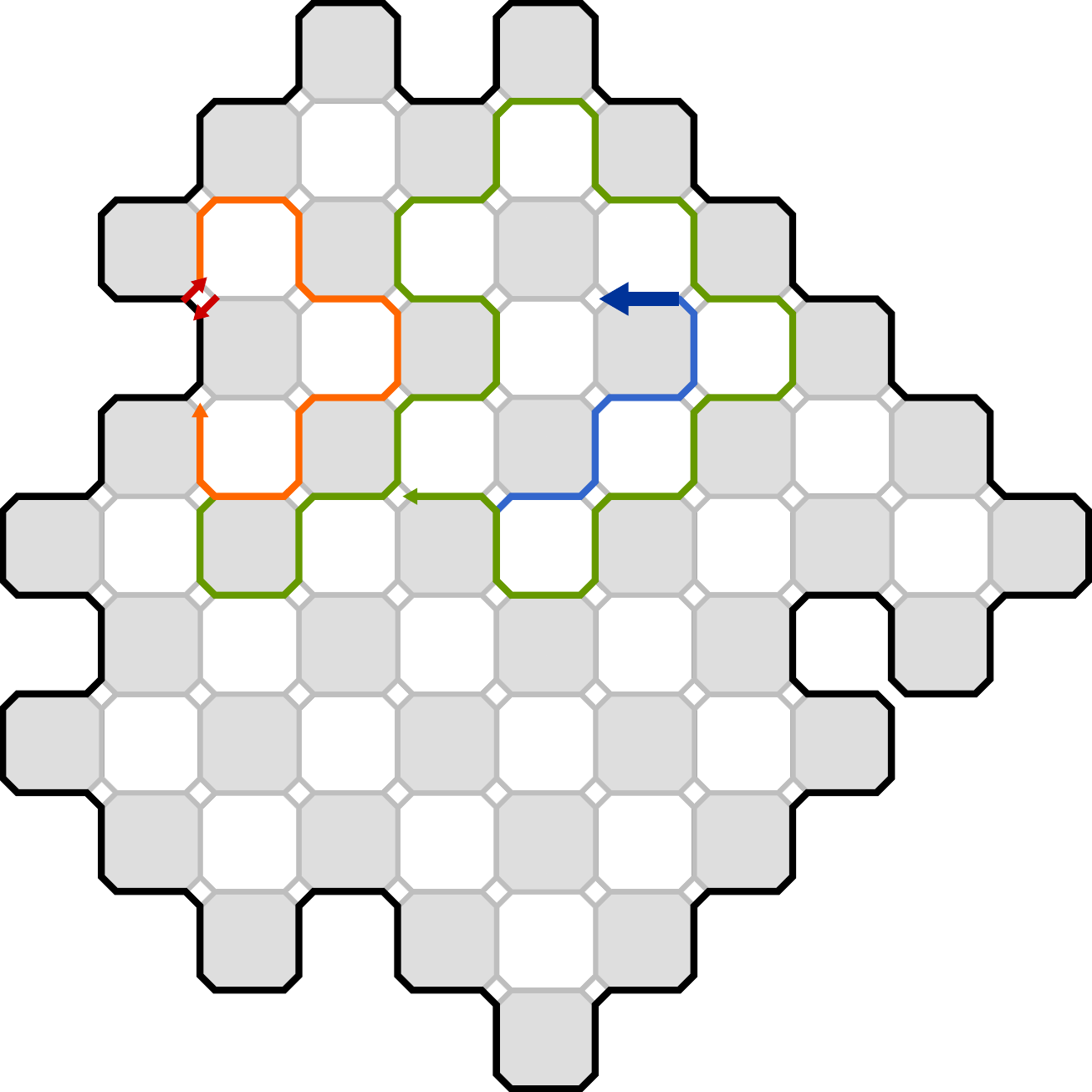}}
\caption{The branch of an exploration tree from $a$ to $e$. Here the square $S_1$ is 
between the red arrows $a$ and $b$ (pointing inwards and outwards, respectively)
and $e$ is the blue arrow.}
\label{fig: domain and a branch}
\end{figure}

\begin{figure}[tbh]
\centering
	{\includegraphics[scale=.3]
{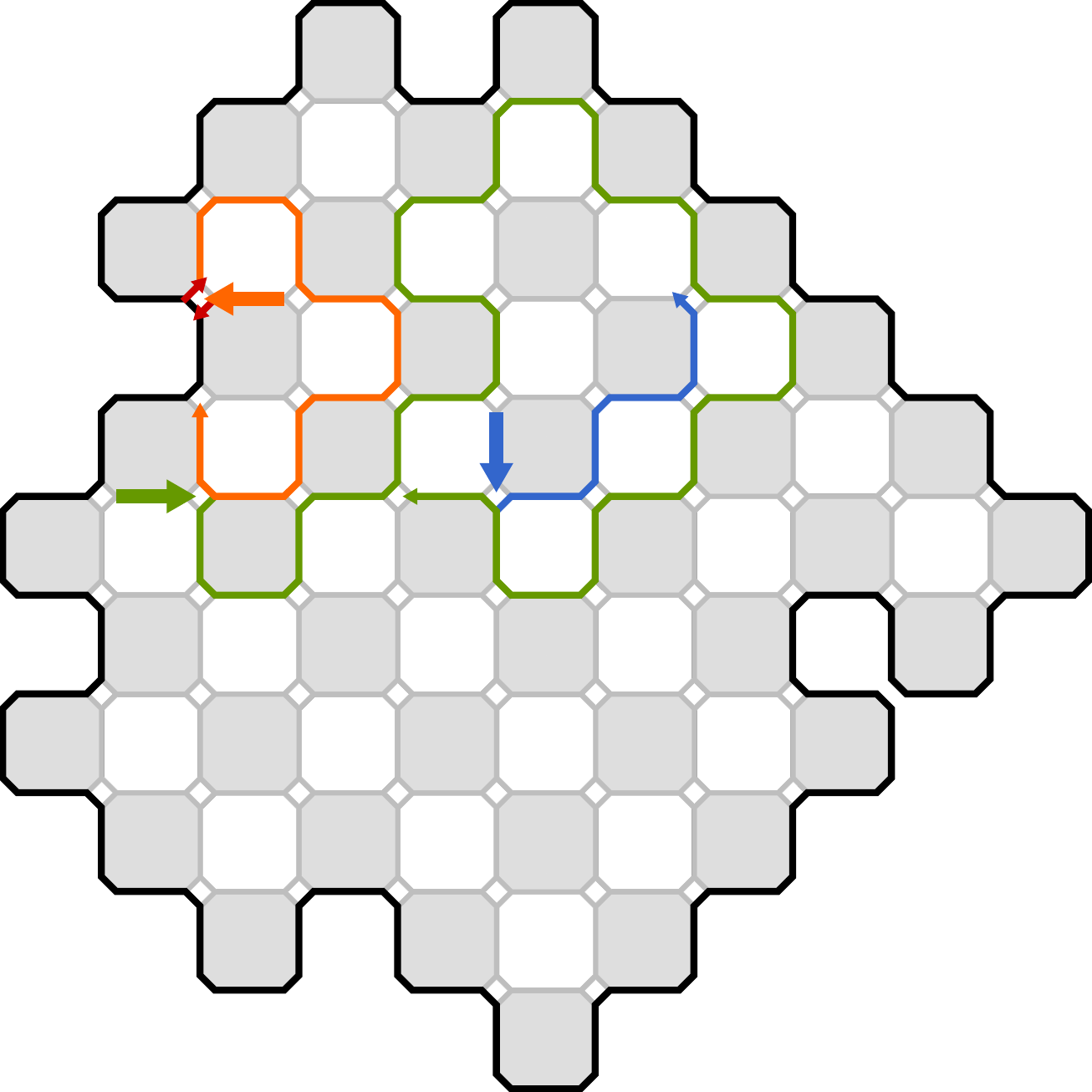}}
\caption{The target edges, when loops are recovered from the tree, are the thick colored arrows
in the picture.}
\label{fig: loops from tree}
\end{figure}

Select a small square $S_1$ next to the boundary. We can assume that it shares exactly one edge with
the boundary (if it shares two, it is a ``bottle neck'' --- a case we exclude
and which doesn't play any role in the continuum limit).
Let the edge in $S_1$ incoming to the domain be $a$ and the outgoing edge be $b$, see
also Figure~\ref{fig: domain and a branch} or Figure~\ref{fig: loops from tree}.
Let $e \in E(\odmdlattice)$ (that is, as an edge in $E(\oddmdlattice)$ it lies between
white and black octagon). Then define in the following way the
branch $\bran_e$ from the root $a$ to the target $e$.
\begin{itemize}
\item Cut open the loop $\lp_{j_1}$ that goes through the edge passing from the tail of $b$
to the head of $a$ by removing that edge. Follow from $a$ the $\lp_{j_1}$ until
the disconnection of $e$ and $b$ on the lattice. Suppose that it happens
on the small square $S_2$.
\item Let $n\geq 1$ and suppose that we have constructed the branch following the
loops $\lp_{j_k}$, $k=1,2,\ldots,n$ until we are at the square $S_{n+1}$ 
and on the loop $\lp_{j_n}$ and that the next step on the loop $\lp_{j_n}$ would
disconnect $e$ and $b$ on the lattice.
Instead of using the next edge on the loop $\lp_{j_n}$,
we use the other possible edge on $S_{n+1}$ (which is not
on any loop) and we arrive to a new (unexplored) loop $\lp_{j_{n+1}}$. Then we
follow that loop until disconnection of $b$ and $e$. Suppose that it happens at the
small square $S_{n+2}$. We continue this construction recursively.
\item The process ends when we reach $e$. Suppose that it happens on a loop $\lp_{j_{N'}}$.
Rename the loops in the sequence as $\lp_{j'_{e,k}}$, $k=1,2,\ldots, N'=N'(e)$,
and the small square sequences as $S'_{e,k}$.
\end{itemize}
This defines the simple lattice path $\bran_e$ from the root $a$ to the target $e$
which we call the \emph{branch of the exploration tree}. The collection
$\tree=(\bran_e)_e$ where $e$ runs over all edges $e \in E(\odmdlattice)$,
is called \emph{the exploration tree} of the loop collection $\lpe=(\lp_j)_{j =1 ,\ldots  N}$.
This construction is illustrated in Figure~\ref{fig: domain and a branch}.

\begin{remark}
It is equally natural to consider a setting 
the marked ``boundary points'' $a'$ and $b'$ are points in $V_\textnormal{mid}$.
Namely, for the incoming edge $a$ and outgoing edge $b$ of $S_1$, take $a'$
to be next and previously point in $V_\textnormal{mid}$ from $a$ and $b$, respectively,
by following the directed graph (these points are unique). Similarly $e$ could be replaced by 
its midpoint $w \in V_\textnormal{mid}$.
\end{remark}

When we consider $\tree=(\bran_e)_e$ as a collection of edges of $\oddmdlattice$
it forms a rooted spanning tree of the graph with
vertices $V(\oddmdlattice) \cap (\{a\} \cup \Omega)$ and all edges connecting
pairs of them.

We say that the branches $\bran_e$ (or rather their coupling)
are \emph{target independent} or local, in the sense that
the initial segments of $\bran_e$ and $\bran_{e'}$ are equal until they
disconnect $e$ and $e'$ on the graph.
Even the sequences $\lp_{j'_{e,k}}$ and $\lp_{j'_{e',k}}$
and on the other hand the sequences $S'_{e,k}$ and $S'_{e',k}$
agree until the disconnection.

The ``tree-to-loops''
construction is illustrated in Figure~\ref{fig: loops from tree}
and it is the inverse of the above ``loops-to-tree'' construction.
Each loop corresponds to exactly one small square where branching of the tree occurs.
Suppose that $e_1 \in E(\odmdlattice)$ is the incoming edge 
used by the branch to arrive
to the small square for the first time and $e_2 \in E(\odmdlattice)$ is the other
incoming edge (opposite to $e_1$ in the square). Then the loop is reconstructed
when we follow the branch to $e_2$ and keep the part after the first exit
from the small square and then closing the loop by adding the side of the
small square that goes from the head of $e_2$ to that exit point.

Finally let us emphasize the geometric characteristic of the branching point.
As it is illustrated in Figure~\ref{fig: loop tree correspondence},
any typical branching point $S_n$ in the scaling limit 
is uniquely characterized as been a ``$5$-arm point''
of a branch in the tree. That is, in the figure, the branch goes through or close
to the square $S_n$ so that the branch forms a ``$5$-arm figure'' ---
two blue, one gray and two green arms.

\begin{figure}[tbh]
\centering
\subfigure[A schematic illustration of the loops--tree correspondence.
The arcs with solid lines form a exploration path towards a branching point
(a corner of $S_n$ in the figure (b)). Thus
the last green part is one of the loops in the loop ensemble which can be identified
as the second generation loop around the point marked with cross.] 
{
   \adjustbox{trim={0\width} {0\height} {0\width} {0\height},clip}%
	{\includegraphics[scale=.5]
{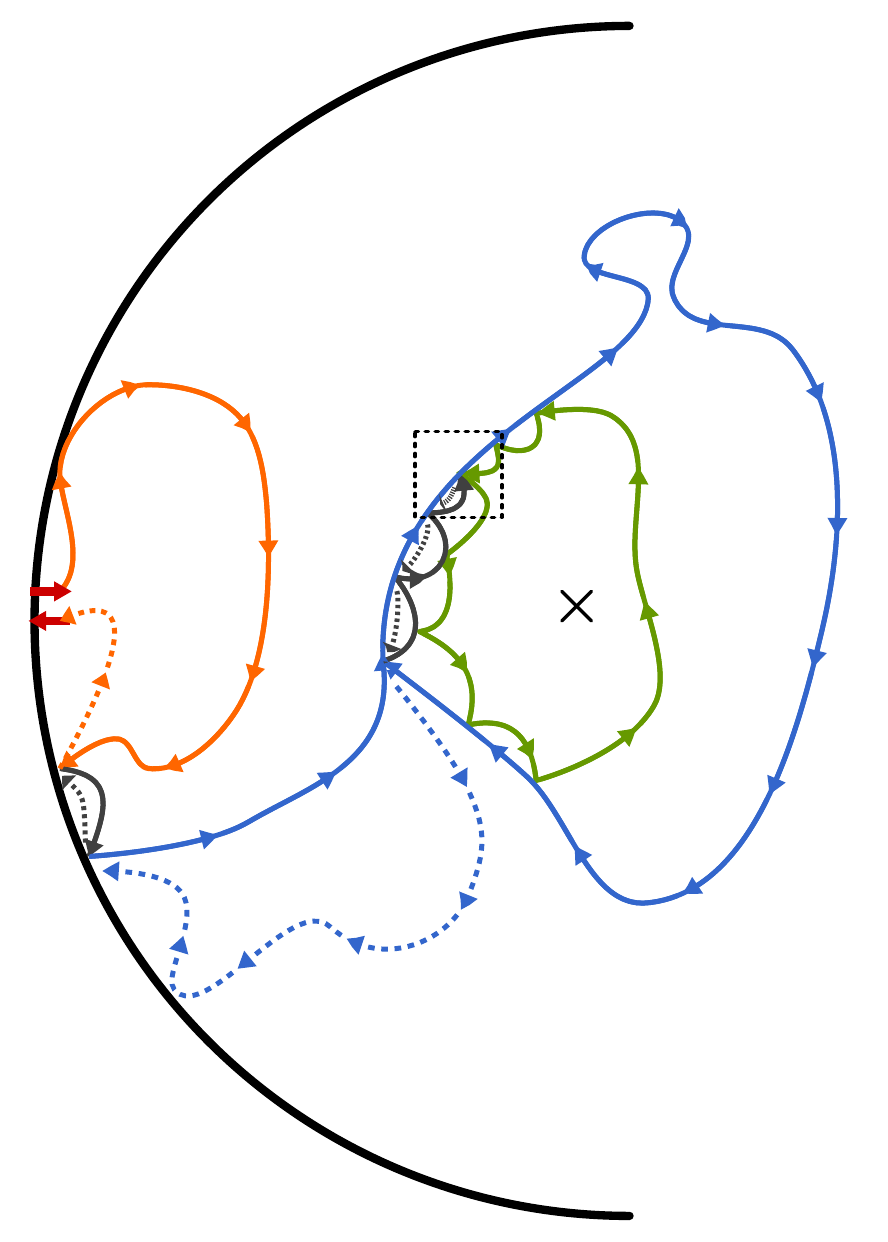}}
} 
\hspace{0.2cm}
\subfigure[Enlargement of the neighborhood of $S_n$ (the dashed box in the figure (a)).
The $5$-arms formed by $2$ blue, $1$ gray and $2$ green
arms emanating from the corners of the square $O_n$.]
{
   \adjustbox{trim={0\width} {0\height} {0\width} {0\height},clip}%
	{\includegraphics[scale=.3,angle=0]
{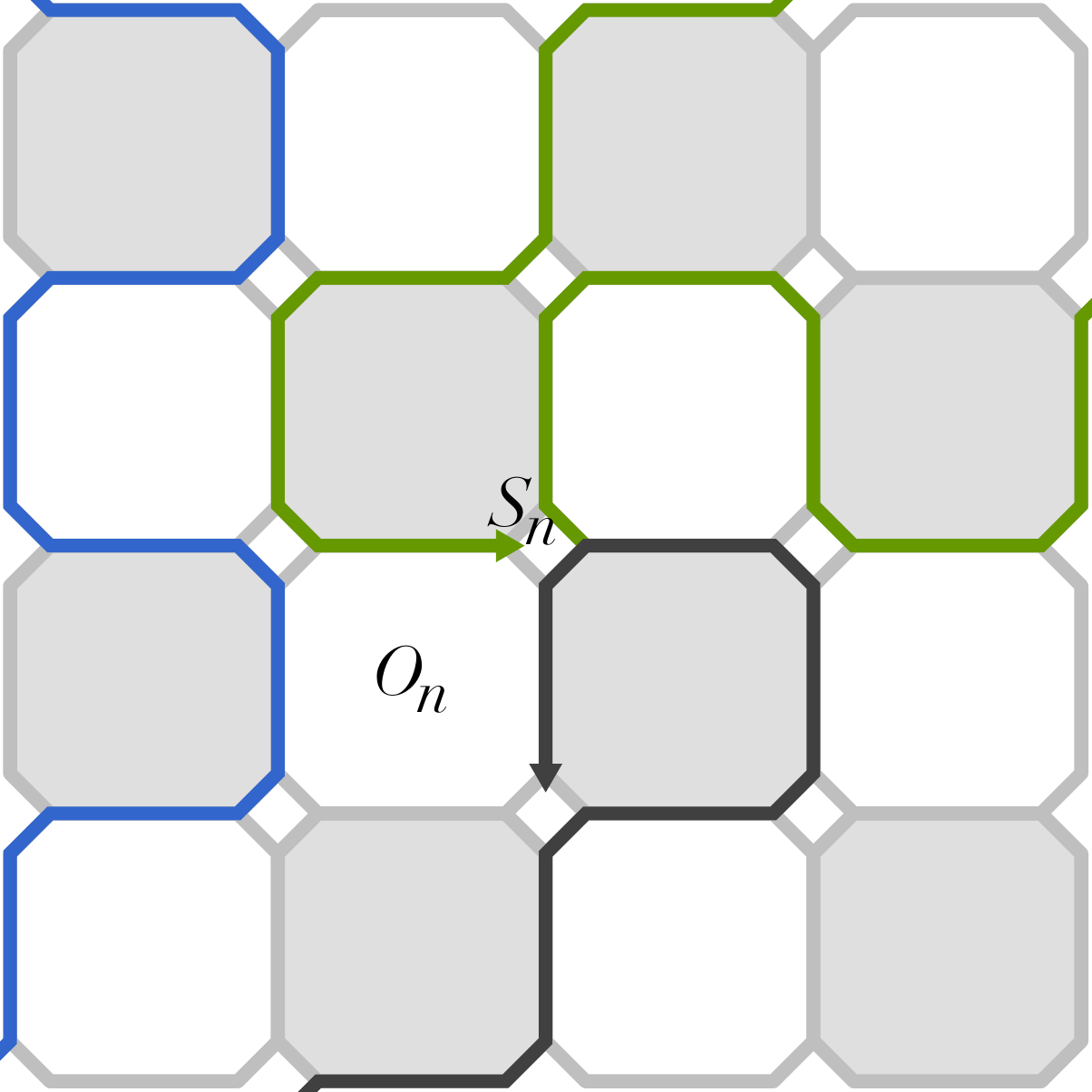}}
}
\caption{The correspondence between a loop ensemble and a tree and
the geometric ``$5$-arm property'' of a branching point.}
\label{fig: loop tree correspondence}
\end{figure}

\subsubsection{Notation for the scaling limit}

Let $\delta >0$.
Suppose that $\Omega_\delta$ is a wired $\delta \oddmdlattice$-domain
and $0 \in \Omega_\delta$. Take a small square that share exactly one edge with the boundary.
One of the edges of the square start from the boundary and one ends at the boundary.
Call them $a_\delta$ and $b_\delta$, respectively. 

We shall consider the random loop collection (loop ensemble) $\rndlpe_\delta$
on each $\Omega_\delta$, $\delta >0$, which are distributed
as the loop representation of FK Ising model (on the corresponding graph).
Define also $\rndttr_\delta$ to be the exploration tree of $\rndlpe_\delta$
with the root $a_\delta$.

When $\Omega_\delta$ is a sequence converging in the Carath\'eodory
sense with respect to $0$ as $\delta \to 0$,
the \emph{scaling limit} is the limit $\lim_{\delta \to 0} (\rndlpe_\delta,\rndttr_\delta)$
with respect to a suitable topology. See Section~\ref{sssec: metrices} below on the discussion
on the topology.

\subsection{SLEs, CLEs and conformal invariant scaling limits}

SLEs are a family of random curves with the conformal Markov property, which
states that conditionally on the initial segment of the random curve,
its continuation is distributed as a conformal
image of an independent sample of the same random curve.
Due to this property they are the natural candidates for the continuum scaling limits
of random curves in critical statistical physics models, which are expected to satisfy
conformal invariance.
Where SLEs are continuum models for
one single curve in statistical physics, their generalizations, CLEs, describe the full continuum limit
of the collection of all random curves in a statistical physics model.

From the perspective of conformal invariance and 
the Riemann mapping theorem, it is natural to
uniformize any simply connected domain by mapping it to a fixed domain. In this article,
we mostly work with the unit disc $\disc=\{ z\in\C \,:\, |z|<1\}$.

\subsubsection{Schramm--Loewner evolution: radial \slek{} and radial \sle{\kappa, \kappa-6}}
\label{sssec: sle}

The radial \slek{} and radial  \sle{\kappa, \kappa-6} are random curves in a simply connected domain
connecting a boundary point to an interior point. We give their definitions
in their natural reference domain $\disc$. The definitions
are extended to other domains as conformal images.

Let ${\gamma:[0,\infty)} \to {\overline{\disc} \setminus \{0\}}$ 
be a curve such that $\gamma(0) \in \partial \disc$.
Denote the connected component containing $0$ in $\half \setminus \gamma( (0,t])$
by $D_t$. Then $D_t$ is simply connected and there exists a unique
conformal and onto map $g_t : D_t \to \half$ such that 
$g_t(0)=0$ and $g_t'(0) > 0$, by the Riemann mapping theorem.
By moving to so-called capacity parametrization, we may assume
that $\gamma$ is parametrized such that $g_t(z) = e^t z + \OO(|z|^2)$
as $z \to 0$.\footnote{%
In fact, we make an assumption here that the capacity increases on any time interval
and that the capacity tends to $\infty$. The latter statement is equivalent to
$\liminf_{t \to \infty} |\gamma(t)| =0$.}
This map satisfies the Loewner equation in $\disc$
\begin{equation}\label{eq: le in disc}
\partial_t g_t(z) = - g_t(z) \frac{g_t(z) + U_t}{g_t(z) - U_t} , \qquad g_0(z)=z 
\end{equation}
for each $t \in [0,\infty)$ and $z \in \disc$. Such process is called the Loewner chain
generated by a curve $\gamma$.

More generally we can also start from \eqref{eq: le in disc} and solve it for a continuous driving term $t \mapsto U_t$.
Denote the solution by $g_t$. For each fixed $z$, the solution ceases to exists at time $t$ which is
the smallest $t$ such that $\liminf_{s \nearrow t} |g_s(z)-U_s|=0$. Then $D_t$ is defined naturally
as those $z \in \disc$ such that the solution $g_s(z)$ exists on a semi-open interval $[0,u)$ with $u>t$
($u$ can depend on $z$). It turns out that $D_t$ is an open set containing $0$ for all $t$ and
$g_t: D_t \to \disc$ is a conformal map such that $g_t(z) = e^t z + \OO(|z|^2)$ as $z \to 0$.
Such a solution is called the Loewner chain (of the unit disc Loewner equation~\eqref{eq: le in disc}) 
driven by the driving term $U_t$.
The trace of the Loewner chain is the limit $\gamma(t) = \lim_{\eps \searrow 0} g_t^{-1} ( U_t( 1 -\eps) )$
if it exists. If the trace exists for all $t$ and defines a continuous function $t  \mapsto \gamma(t)$,
then the Loewner chain is generated by the curve $\gamma$ also in the above sense.
It is standard to set $K_t = \overline{\disc \setminus D_t}$. The family of ``hulls'' $(K_t)_t$
increases continuously as a function of $t$ and is the closest generalization for the curve $\gamma$.

\begin{definition}
A random Loewner chain $(g_t,K_t)$ is a \emph{radial \slek}, 
if $U_t = \exp(\ii\sqrt{\kappa} B_t)$
for some Brownian motion $(B_t)_{t \in [0,\infty)}$.
\end{definition}

\begin{definition}
A random Loewner chain $(g_t,K_t)$ is a \emph{radial \sle{\kappa, \kappa-6}}
(we assume that $\kappa \in (4,8)$), 
if $U_t$ is the first coordinate of an adapted, continuous semimartingale $(U_t,V_t)$ 
such that
$U_0=V_0$,
\begin{equation}
V_t = V_0  - \int_0^t V_s \frac{V_s + U_s}{V_s - U_s} \de s
\end{equation}
for all $t$,
\begin{align}
\de U_t &= \ii \sqrt{\kappa} U_t \;\de B_t + \left( -\frac{\kappa}{2} U_t
  - \frac{\kappa-6}{2} U_t \frac{U_t + V_t}{U_t - V_t}  \right) \de t \end{align}
for all $t$ such that $U_t \neq V_t$
(for some Brownian motion $(B_t)_{t \in [0,\infty)}$)
and $\arg (V_t/U_t) \in [0,2\pi]$ is instantaneously 
reflecting at $0$ and $2 \pi$, meaning in particular that
$$\P\left[ \int_0^\infty \ind_{U_t=V_t} \de t = 0 \right]  = 1.$$ 
\end{definition}

It turns out that 
radial \slek{} and radial \sle{\kappa, \kappa-6}  have continuous traces and are thus 
generated by curves. However, such statement is not needed below, since a side product
of our proofs is that the convergence is strong enough to show that
the limiting Loewner chains are generated by curves. More specifically, if
the curve $\gamma_{\delta}$ in $\disc$ is defined as the conformal image of
the discrete curve in $\domain_\delta$ under a transformation to $\disc$
and parametrized by the d-capacity (i.e. $\log g_t'(0) = t$),
then the sequence of curves $\gamma_{\delta}$ in distribution in the topology of
continuous functions to a random curve $\gamma$. It turns out
that this convergence implies that the convergence takes place simultaneously
as Loewner chains, see \cite{Lawler:2008wf,Kemppainen:tb}.

The chordal and radial \sle{\kappa, \kappa-6} processes only differ by 
the fact that the target point for the former process is on the boundary while
the one for the latter process is in the bulk. Their laws until the disconnection
of the two alternative target points are the same, and after the disconnection
the processes turn towards their own target points.
This follows since the sum of ``$\rho$'s'' is equal to
 $\kappa-6$ and hence there is no force applied by the marked points $\infty$ or $0$
(in $\half$ and $\disc$, respectively). This is the target independence 
or locality property of
SLE$(\kappa, \kappa-6)$.
This property is not valid for the chordal and radial SLE$(\kappa)$ 
whose laws are different (though absolutely continuous with respect 
to each other on appropriately chosen time intervals);
see \cite{Schramm:2005ty} for the transformation rule between the upper half-plane and
the unit disc.

\subsubsection{Conformal loop ensembles}

CLEs are random collections of loops in simply connected domains. We give below
first their axiomatic definition. In the main bulk of the article, we rely on their construction
via branching SLE processes rather than the axiomatic definition.

Suppose that we are given a family of probability measures $(\mu^\domain)_\domain$ where
$\Omega$ runs over simply connected domains and
$\mu^\domain$ is the law of a random loop collection on $\domain$. If $\rndlpe = ( \rndlp_j )_j$
is distributed according to $\mu^\domain$, we suppose that almost surely
each loop $\rndlp_j$ is simple, $\rndlp_i \cap \rndlp_j = \emptyset$ when $i \neq j$
and they satisfy the following properties:
\begin{itemize}
\item{} {(\it Conformal invariance (CI))} 
  If $\varconfmap:\Omega \to \C$ is conformal and $\varconfmap^*$ is its pushforward map, then
  $\varconfmap^* \mu^\Omega = \mu^{\varconfmap(\Omega)}$.
\item{} {(\it Domain Markov property (DMP))} If $\Omega' \subset \Omega$ is
a simply connected domain, $J'$ is the set indices $j$ 
such that $\rndlp_j \cap (\Omega \setminus \Omega') \neq \emptyset$
and $\tilde \Omega$ is equal to $\Omega' \setminus \bigcup_{j \in J'} 
  \overline{\operatorname{int}(\gamma_j)}$,
then the law of $(\gamma_j)_{j \notin J'}$ is equal to $\mu^{\tilde \Omega}$.   
\end{itemize}
If the collection $\rndlpe = ( \rndlp_j )_j$ satisfy these properties, we
call it \emph{conformal loop ensemble} (CLE).

It turns out that loops in CLE's are SLE-type curves;  
see Section~1 of \cite{Sheffield:2012kw} for several formulations of this kind of a result. 
A given CLE corresponds to \slek{} with a unique $\kappa \in (8/3,4]$.
We use the notation \clek{} for the CLE that corresponds to \slek{}.
See also \cite{Sheffield:2012kw} for uniqueness statement on CLE's.

A third view that we adopt to CLE is the branching \sle{\kappa,\kappa-6}
construction of \clek, $\kappa \in (8/3,8)$, which allows the extension of the definition
to values $\kappa \in (4,8)$ which is highly relevant for this article.
This process is a collection of curves $\gamma_z$ from the root $a \in \partial \disc$
to the target $z$, where $z$ runs over all points in $\disc$.

\begin{definition}
The random collection of curves $(\gamma_z)_z$ is a \emph{branching
\sle{\kappa,\kappa-6}}, if
the law of $\gamma_z$ is the radial \sle{\kappa,\kappa-6} from $a$ to $z$ and
moreover the curves are coupled so that for each $z \neq z'$ it holds
that $\gamma_z$ and $\gamma_{z'}$ are equal until the disconnection
of $z$ and $z'$ by $\gamma_z$ (or $\gamma_{z'}$). 
\end{definition}

A tree in graph theory is a connected graph without any cycles, or equivalently
a graph such that any pair of points is connected by a unique simple path.
In the same spirit, it is natural to say that the branching
\sle{\kappa,\kappa-6} forms a tree: from the root $a$ to any (generic) point $z$
there is a unique path $\gamma_z$ and between any (generic) points $z \neq z'$
the unique path follows the reversal of $\gamma_z$ to the \emph{branching point}
of $\gamma_z$ and $\gamma_{z'}$ and then $\gamma_{z'}$ from that
point to $z'$.

\subsubsection{Convergence of curves and curve collections}\label{sssec: metrices}

In this subsection,
we present first the topology for the convergence for 
branches and trees and then for loops and loop ensembles.

{\it (Metrics for branches and trees)}
Consider a triplet $(\domain,\confmap; \bran)$
where
\begin{itemize}
\item $\Omega$ is a simply connected domain
\item $\confmap:\Omega \to \disc$ is a conformal and onto map
\item $\bran : [0,1] \to \overline{\Omega}$ is a curve 
such that there exists a curve $\todisc \bran$ in $\overline{\disc}$
parametrized by the d-capacity such that $\confmap \circ \bran$
and $\todisc \bran$ are equal up to a non-decreasing 
reparametrization.
\end{itemize}

Define using the supremum norm a metric for the d-capacity parametrized curves 
\begin{equation}
\decc ( \bran_1,\bran_2) =
\decc \big( (\domain_1,\confmap_1; \bran_1), (\domain_2,\confmap_2; \bran_2) \big)
\dd = \Vert \todisc {(\bran_1)} - \todisc {(\bran_2)} 
        \Vert_{\infty,[0,\infty)} .
\end{equation}

\begin{definition}
A \emph{rooted tree} $\ttr = (x_0;(\bran_x)_{x \in S})$ is pair such that $x_0$ is a point
called \emph{root} and $(\bran_x)_{x \in S}$ a set of curves starting at $x_0$ indexed by 
a set of points $S$ so that $x \in S$ is the other endpoint of $\bran_x$.
\end{definition}

Define a metric for trees as 
\begin{align}
\dettr (\ttr_1,\ttr_2) &=
\dettr \big( (\domain_1,\confmap_1; \ttr_1), 
    (\domain_2,\confmap_2; \ttr_2) \big) \nonumber \\
   &\dd= \max\left\{ \sup_{\bran_1} \inf_{\bran_2} \decc ( \bran_1,\bran_2)  ,
      \sup_{\bran_2} \inf_{\bran_1} \decc ( \bran_1,\bran_2) \right\}.
\end{align}
where $\bran_k$ runs over all the branches of $\ttr_k$, for $k=1,2$.
This is the familiar Hausdorff metric for bounded closed sets.

{\it (Metrics for loops and loop ensembles)}
Similarly we define metrics for loops and loop ensembles.
The difference is that there are no marked points for loops and
thus there is no natural starting or ending point and we cannot describe
it in a canonical way with Loewner evolutions. Thus it makes sense to
define in the following way. Let
\begin{equation}
\delp( \lp_1, \lp_2 ) 
  = \delp \big( (\domain_1,\confmap_1; \lp_1), 
    (\domain_2,\confmap_2; \lp_2) \big)
  = \inf_{f_1,f_2} \Vert f_1-f_2 \Vert_{\infty}
\end{equation}
where $f_k$ runs over all parametrizations of $\confmap_k \circ \lp_k$.
The metric $\delpe$ for loop ensembles is 
defined to be the Hausdorff metric for bounded, closed sets of loops.

\subsection{On the structure of the proof of Theorem~\ref{thm: main thm}}

The general approach which we take in the proof of Theorem~\ref{thm: main thm},
is to use of compactness and uniqueness. We show that the sequence of probability
measures on the metric space of trees is precompact with respect to
the weak convergence of probability measures. This implies that
there are weakly convergent subsequences of the probability measures.
Then we will show that the limit of the subsequence is unique. Thus the entire sequence
converges.

The compactness part of the proof is based on estimates implying H\"older regularity
of the path collections. The estimates are based on probability bounds on annulus crossings
of the same type as in \cite{Kemppainen:2012vma}.
Due to the fact that we need to work with mixed boundary conditions (of the FK Ising model)
we need the strong version of the probability estimates for crossing events,
which was established in \cite{Chelkak:2016jw}.

The uniqueness part is based on the holomorphic observables of FK Ising model developed in \cite{Smirnov:2010ie}
and extended here to be adapted to the exploration tree setting.
We don't use here any a priori knowledge on the processes at hand, such as
the fact that the FK Ising interfaces in the chordal setting converge to the chordal \sle{\frac{16}{3}}
(see \cite{Chelkak:2014gs})
or that the radial \sle{\kappa,\kappa-6} satisfies the target independence property \cite{Schramm:2005ty,Sheffield:2009}.
The fact that a single branch of the exploration tree converges to the radial
\sle{\frac{16}{3},-\frac{2}{3}} is merely a product of a calculation based on the observables.


\section{The discrete holomorphic observable and its scaling limit}

\subsection{The discrete observable}\label{ssec: def int observable}

Let us consider FK Ising model on a square lattice with a lattice mesh parameter $\delta>0$.
In that setup,
suppose that we are given a Dobrushin domain 
$\domain = \domain_\delta$ with an incoming edge $\medge_i$ and an outgoing edge $\medge_o$,
see Figure~\ref{fig: dobrushin domain} for the definition.
Denote the set of directed edges of the medial lattice by $E(\oddmdlattice)$. 
As usual, a directed edge 
$\medge \in E(\oddmdlattice)$ is given by an ordered pair $(\medge_-,\medge_+) \in (V(\oddmdlattice))^2$.

\begin{figure}[tbh]
\centering
	{\includegraphics[scale=.3]
{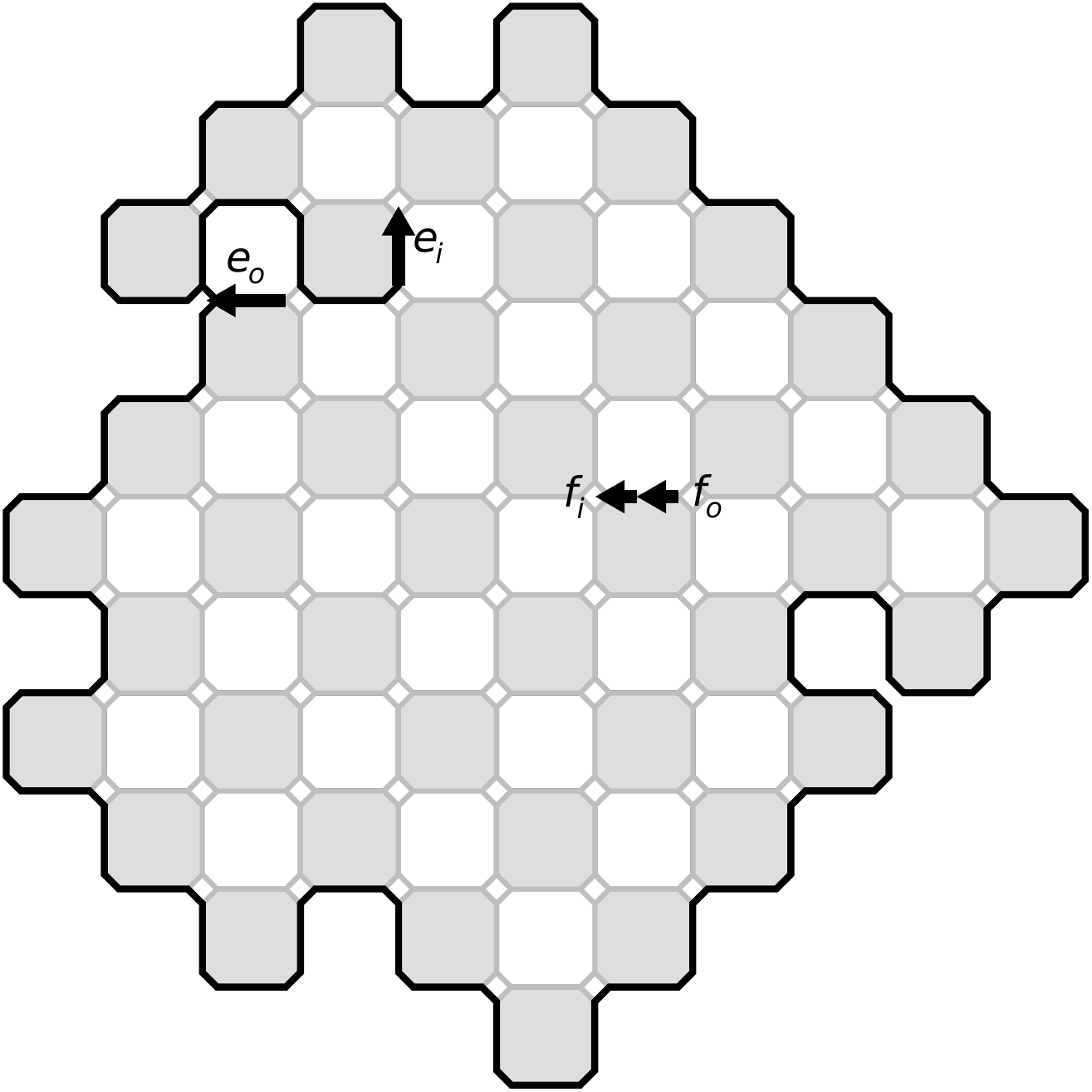}}
\caption{A Dobrushin domain has two boundary arcs, one with wired (black) boundary and the other
with free (white) boundary. We call the edges where the boundary conditions change
$\medge_i$ and $\medge_o$; later $a$ and $b$, respectively. We call a fixed target edge $\varmedge$
and its halves $\varmedge_o$ and $\varmedge_i$; later we will also notation $w$ for the target.}
\label{fig: dobrushin domain}
\end{figure}

We fix an interior edge $\varmedge \in E(\oddmdlattice)$ which we split into 
two halves $\varmedge_o$ and $\varmedge_i$ which are outgoing and incoming edges, respectively,
in the new graph $G$. At first, $(\varmedge_o)_+$ and $(\varmedge_i)_-$ are not connected by an edge.

\begin{definition}
We define two enhanced graphs 
$\Gouter$ and $\Ginner$
by adding an edge between 
$(\medge_o)_+$ and $(\medge_i)_-$ or between $(\varmedge_o)_+$ and $(\varmedge_i)_-$, respectively.
\end{definition}

Here $\extarc{\medge_1}{\medge_2}$ denotes an ``external arc'', that is, an edge outside
of the graph $G$ that is added, which affects the boundary conditions of the model
in the sense that if that external edge closes a loop in the loop configuration
than an extra weight $\sqrt{2}$ is given to the configuration.
In contrast,  
${\intarcp{\medge_1}{\medge_2}{\medge_3}{\medge_4}}$ would be an internal arc configuration, which is 
a connection pattern in the loop configuration (of the FK loop representation) and which can be interpreted
as an event.

Next we define the central tool of this article, namely, the discrete spin holomorphic observables; 
in fact, two of them,
the one which was introduced in the first paper \cite{Smirnov:2010ie} and its
variant for the branch of the radial exploration tree.
For presenting the definition, 
let $\hat e = (e_+ - e_-)/|e_+ - e_-|$ and
$\vartheta(e) = (\hat e)^{-1/2}$
for any oriented edge $e = (e_-,e_+)$.
Set
\begin{equation}
\lambda = e^{-\ii \frac{\pi}{4} } .
\end{equation} 
We choose the branch of the square root so that 
$\vartheta(e) \in \{1,\lambda,\bar{\lambda},\ii\}$
for all $e \in \odmdlattice$.
For $e \in \odmdlattice$, $\vartheta(e)$ are the complex factors
associated to edges which is used for instance in 
\cite{Smirnov:2010ie}.

On the graph $\Gouter$, the loop representation configuration consists of loops and a single open curve $\gamma$
which runs from $\varmedge_i$ to $\varmedge_o$. Notice that either a loop
or the curve uses the external arc $\outeredge$.
Define a function $F = F_{\outeredge}: E(\oddmdlattice) \to \C$ by
\begin{equation}\label{eq: def int observable}
F (\medge) = 
   (- \vartheta(\varmedge)) \; \E_{\outeredge} 
   \left( \ind_{\medge \in \gamma} \; e^{- \frac{\ii}{2} W(\varmedge_o,\medge) }  \right)
\end{equation}
where $W(\varmedge_0,\medge)$ is the winding along $\gamma^{\leftarrow}$ 
from $\varmedge_0$ to $\medge$. Here 
the expected value is taken with respect to the critical FK Ising
loop measure on the planar graph $\Gouter$.
The measure is supported
on loop configurations with an path $\gamma$ from $\varmedge_i$ to $\varmedge_0$ 
and in the formula~\eqref{eq: def int observable},
$\gamma^\leftarrow$ denotes the reversal of $\gamma$. 
There are two natural ways (namely, by following one of the boundary arcs) 
to define the winding along the arc from $(\medge_o)_+$ to $(\medge_i)_-$ 
but both choices lead to the same value for $F$: namely,
the difference in $W(\varmedge_0,\medge)$ is $\pm 4 \pi$ hence $e^{- \frac{\ii}{2} W(\varmedge_o,\medge) }$ is well-defined.
The constant in front of the expectation value is to ensure that $F(\varmedge_i)=-F(\varmedge_o)=+\vartheta(\varmedge)$.

The observable~\eqref{eq: def int observable} is given by calculating the number of left and right turns from $\varmedge_o$ to $\medge$
along $\gamma^\leftarrow$ and weighting the partition function by 
$-\vartheta(\varmedge)\, \lambda^\text{signed number of turns}$.

On the graph $\Ginner$, the loop representation configuration consists of loops and a single open curve $\tilde \gamma$
which runs from $\medge_i$ to $\medge_o$.
The observable introduced in the first paper \cite{Smirnov:2010ie}
is denoted here by $\tilde F = \tilde F_{\inneredge}:  E(\oddmdlattice) \to \C$  and defined as
\begin{equation}\label{eq: def usual observable}
\tilde F(\medge)  = 
   \vartheta(\medge_o) \, \E_{\inneredge} 
   \left( \ind_{\medge \in \tilde \gamma} \; e^{- \frac{\ii}{2} W(\medge_o,\medge) }  \right)
\end{equation}
where $W(\medge_0,\medge)$ is the winding along $\tilde{\gamma}^{\leftarrow}$ from $\medge_0$ to $\medge$,
$\tilde\gamma$ is the path from $\medge_i$ to $\medge_0$ on  $\Ginner=G$
and $\tilde{\gamma}^\leftarrow$ is the reversal of $\tilde{\gamma}$.
Notice that the difference of \eqref{eq: def int observable} 
and \eqref{eq: def usual observable} is only in the graph being used and how that affects
which part of the loop configuration is counted as the curve.
With the constant in front of the expected value, $\tilde F$ satisfies the ``correct'' complex phase 
on the lattice and in particular, $\tilde F(\medge_o) = \vartheta(\medge_o)$.

Define for each $v \in \domain_\delta \cap V(\delta\dmdlattice)$, $F(v) = F(e_E) +F(e_W)$
and $\tilde F(v) = \tilde F(e_E) +\tilde F(e_W)$ where $e_E$ and $e_W$
are the two horizontal edges whose endpoint is $v$.

\subsubsection{Discrete holomorphicity of the observables}

The following is a central definition when passing to the continuum scaling limit
\cite{Smirnov:2010ie,Chelkak:2011by,ChelkakSmirnov:2012,Kemppainen:2015vu}.
It defines a strong version of discrete holomorphicity.

\begin{definition}
For an edge $e = \{v,w\} \in E(\delta\dmdlattice)$
and a function $F$ defined on a subset of $V(\delta\dmdlattice)$,
$F$ is \emph{spin holomorphic at $e$}
if $F$ is defined on on $v$ and $w$ and 
\begin{equation}\label{eq: spin holomorphic}
\Proj[ F(v) ; \vartheta(e)\R ] = \Proj[ F(w) ; \vartheta(e)\R ] .
\end{equation}
Here $\Proj[ z ; \alpha \R ] = \alpha \real(z \bar \alpha) = (z +  \bar z \alpha^2)/2$
for any $\alpha \in \C$ with $|\alpha|=1$.
\end{definition}

In \cite{Smirnov:2010ie}, it was shown that $\tilde F$ is spin holomorphic
and it was shown that it satisfies a boundary value problem with a unique solution.
In that article, it was also demonstrated that for passing to the limit with the boundary value problem 
of an observable $F$ (which can be either one of above $\tilde F$ or $F$) arising
from the FK Ising model
is
easier for a function 
$H$ which is the discrete version of  $\frac{1}{2}\imag \int F^2 \de z$ defined as
\begin{equation}\label{eq: H def}
H(B)- H(W)= |F(e)|^2
\end{equation}
where $B$ and $W$ are neighboring black and white squares and $e$ is the common edge of the squares.\footnote{%
One can verify that $H(B_2) - H(B_1) = \frac{1}{2} \imag [ \sum_k F( \frac{1}{2}(z_k+z_{k+1}) )^2 (z_{k+1}-z_k)]$
where $z_k$ is a path of black vertices connecting $B_1$ to $B_2$.
Notice also that in the continuum, the formula $H = \frac{1}{2}\imag \int F^2 \de z$ for a holomorphic $F$
can be inverted by $F=2 \sqrt{ \ii \cd H }$.}
Then $H$ is approximately discrete harmonic, see \cite{Smirnov:2010ie} Section~3,
and satisfies Dirichlet boundary values.

\subsubsection{The martingale property of the observables}

Let $\domain$ be a Dobrushin domain  with marked edges $\medge_i$ and $\medge_o$.
Let $x \in \Vtarget \cap \domain$ and
$\therndbran(t)$, $t \in \{0,\frac{1}{2},1,\frac{3}{2},\ldots\}$ 
be the branch of the FK Ising exploration tree in $\domain$
from the root $\medge_i$ to $x$. Here the parametrization of $\therndbran$ is by lattice steps so that
at $T[k , k +\frac{1}{2}]$ is an edge of a small square and $T[k+\frac{1}{2},k+1]$
is an edge shared by two ocatagons for all non-negative integers $k$.

Denote $\domain \setminus T(0,t]$ by $\domain(t)$. Notice that $\domain(t)$
is a Dobrushin domain for integer $t$. Denote the law of the FK Ising model in $\domain(t)$
by $\P^{\domain(t)}$. Similarly as above, denote by $\gamma$ the loop going through $x$ and
by $\tilde \gamma$ the loop going through $\medge_i$.

Let the process $(\hat N(t))_\tinz$ be defined by the formulas
\begin{equation}
\hat N(t) = \P^{\domain(t)}_{\outeredge} [ (\outeredge) \in \gamma ]= \P^{\domain(t)}_{\outeredge} [ \tilde \gamma= \gamma ] .
\end{equation}
Notice that we can equivalently write $\hat N(t)= |\tilde{F}^{\domain(t)}(\varmedge)|$.
Here $\Ztime = \Z \cap [0,\infty)$.
Denote the loop being explored at time $t$ by $\tilde \gamma_t$.
It is standard to notice that $\hat N(t)$ can be written as a conditional expected value with respect to
the $\sigma$-algebra $\F^T_t$ generated by $\rndbran(u)$, $u \leq t$, as
\begin{equation}
\hat N(t) = \P^{\domain}_{\outeredge} [ \tilde \gamma_t= \gamma \,|\, \F^T_t ]
\end{equation}
When moving from $t$ to $t+1$, if there are two possible ways that $\rndbran$ can continue,
then necessarily $\tilde \gamma_{t+1} =\tilde \gamma_t$. Thus at those times it holds
that $\E^{\domain}_{\outeredge}[ \hat N(t+1) \,|\, \F^T_t ] = \hat N(t)$
by properties of the conditional expected value. We call this identity the martingale property
of $(\hat N(t))_\tinz$.

To extend the martingale property to those times that $\rndbran$ has only one choice to continue,
we introduce a set of i.i.d. $\pm 1$-valued random variables $\xi_v$, $v \in \dmdlattice \cap \Omega$, 
such that $\E[\xi_v] = 2 \P[\xi_v] -1 = \sqrt{2} - 1$.
\begin{equation}\label{eq: def discrete N}
N(t) = \left( \prod_{k} \xi_{\therndbran(k)} \right) \, \hat N(t)
\end{equation}
where  the product is over
$k=0,1,2,\ldots,t-1$ such that $\therndbran(k)$ is a branching point.

Let the process $(M(t))_\tinz$ be defined by
\begin{equation}
M(t) = F^{\domain(t)} (e) .
\end{equation}
The following result is almost immediate. For the further details of the proof, see \cite{Kemppainen:2015vu}.

\begin{proposition}
The processes $(M(t))_{t \in \{0,1,2,\ldots\}}$ and $(N(t))_{t \in \{0,1,2,\ldots\}}$
are discrete time martingales.
\end{proposition}

\subsubsection{Notation for the scaling limit of the observables}

In addition to the notation $\domain_\delta$, let us introduce $a_\delta$, $b_\delta$ and $w_\delta$
for the heads of the edges $\medge_i,\medge_o$ and $\varmedge_o$, respectively.
Then $(\domain_\delta,a_\delta,b_\delta,w_\delta)$ is a domain in the complex plane with two marked boundary points 
and a marked interior point, in that order.

For each pair $(\domain,w)$ where $\domain$ is a simply connected domain ($\neq \C$) and $w \in \domain$, 
let $\confmap_{(\domain,w)}: \domain \to \disc$ be the unique conformal and onto map satisfying 
$\confmap_{(\domain,w)}(w)=0$ and $\confmap_{(\domain,w)}'(w) > 0$.

\begin{definition}
We say that $(\domain_\delta,a_\delta,b_\delta,w_\delta)$ \emph{converges to} $(\domain,a,b,w)$, 
where $a,b$ can be prime
ends (generalized boundary points),
\emph{in the Carath\'eodory sense} if the sequence of conformal maps $\confmap_{(\domain_\delta,w_\delta)}^{-1}$ 
converges uniformly on compact subsets of $\disc$ to the conformal map
$\confmap_{(\domain,w)}^{-1}$ as $\delta$ tends to zero and in addition
$\lim_{\delta \to 0} \confmap_{(\domain_\delta,w_\delta)}(a_\delta) = \confmap_{(\domain,w)}(a)$
and $\lim_{\delta \to 0} \confmap_{(\domain_\delta,w_\delta)}(b_\delta) = \confmap_{(\domain,w)}(b)$.
In the last two equations, the values of the right-hand sides exist as boundary points of $\disc$.
\end{definition}

For a fixed sequence of domains $(\domain_\delta,a_\delta,b_\delta,w_\delta)$,
denote the observables in $(\domain_\delta,a_\delta,b_\delta,w_\delta)$
as $F_\delta$ and $\tilde F_\delta$.

\subsection{The scaling limit of the observable}

\begin{figure}[tbh]
\centering
\includegraphics[scale=.4]%
    {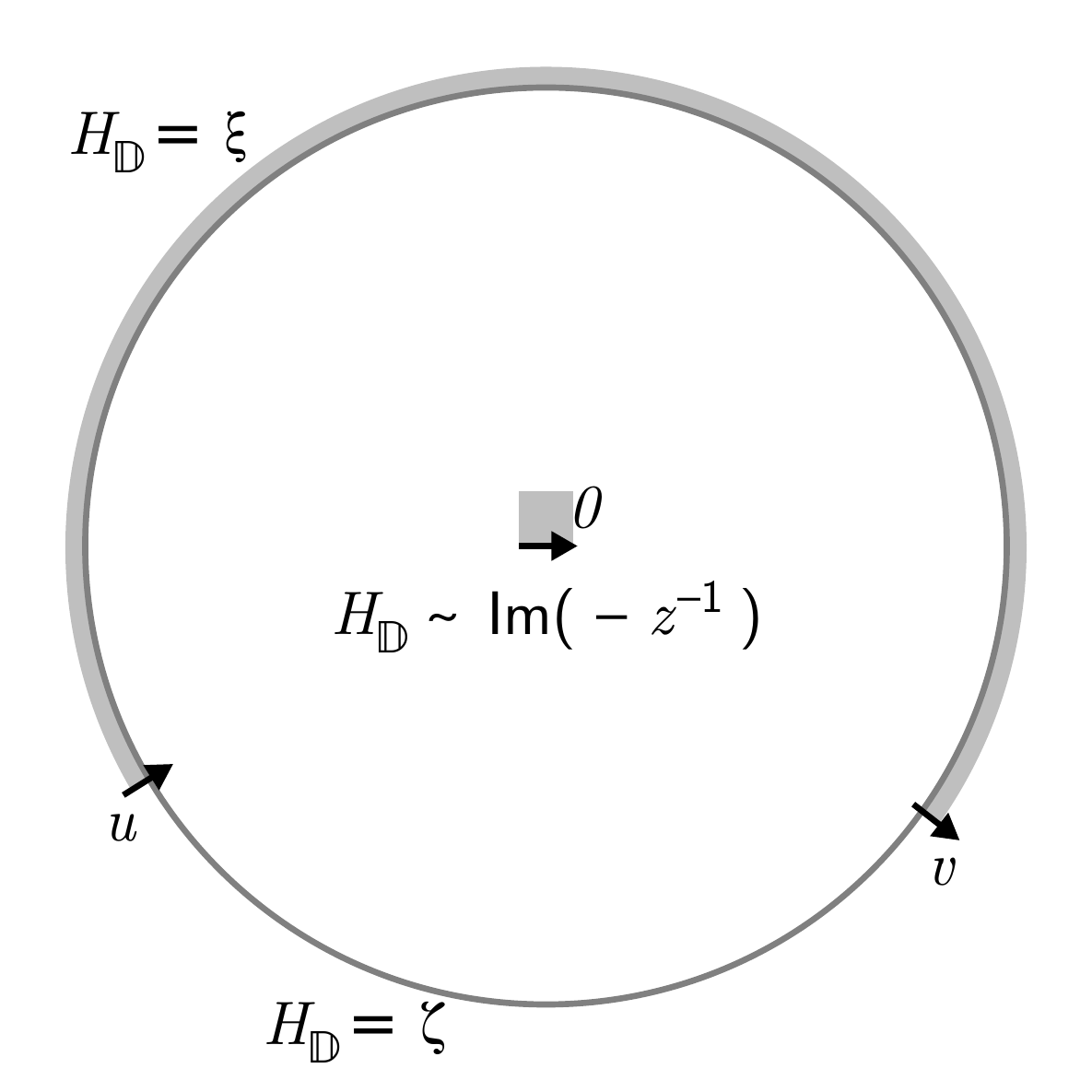}
\caption{Boundary values on arcs $uv$ and $vu$ and the pole of $H_\disc$ at $0$. The type of boundary
is illustrated by the grey shading which indicates the wired boundary.}
\label{fig: H boundary conditions}
\end{figure}

For $\coeffa,\coeffb \in \R$, $z \in \disc$ and $u,v \in \partial \disc$, define
\begin{equation} \label{eq: disc int observable}
F_\disc (z ; u,v) = \frac{2}{\pi} \,
  \sqrt{ \frac{1}{z^2} - 1 + i\,\coeffa\,\frac{1}{z} + \coeffb\, \left( \frac{1}{z-u} - \frac{1}{z-v} \right)  }
\end{equation}
where $\coeffa \in \R$ and $\coeffb>0$. Define also
\begin{equation} \label{eq: usual disc int observable}
\tilde F_\disc (z ; u,v) = \sqrt{\frac{2}{\pi}} \,  
  \sqrt{  \frac{1}{z-u} - \frac{1}{z-v}   }
\end{equation}

\begin{remark}\label{rem: scaling limit observables}
Consider a function
\begin{equation}
H_\disc=\frac{1}{2}\imag \int F_\disc^2 \,\de z .
\end{equation}
Then $H_\disc$ is constant on both arcs $uv$ and $vu$. Let those constant be equal to $\zeta$ and $\xi$,
respectively. Here $uv$ is the counterclockwise arc on $\partial \disc$ from $u$ to $v$
and $vu$ from $v$ to $u$. Then $\xi - \zeta = 2\beta /\pi > 0$.
Moreover, from \eqref{eq: disc int observable}
it follows that
\begin{equation}
H_\disc (z) = 
\frac{2}{\pi^2} \imag \left( -\frac{1}{z}\right)  + \OO\left( \log\frac{1}{|z|}\right)
\end{equation}
as $z \to 0$ . See also Figure~\ref{fig: H boundary conditions}.
\end{remark}

\begin{theorem}\label{thm: conv obs}
Suppose that $(\domain_\delta,a_\delta,b_\delta,w_\delta)$ converges to $(\domain,a,b,w)$
and let $\confmap: \domain \to \disc$ be the conformal and onto map such that $\confmap(w)=0$ and $\confmap'(w)>0$.

Suppose first that $\varmedge$ is a horizontal edge pointing to the east.
As $\delta$ tends to zero, $\delta^{-1} F_\delta$ and $\delta^{-1/2} \tilde F_\delta$ converge 
to the
scaling limits $F$ and $\tilde F$, respectively, which are uniquely determined by
\begin{align}
F(z; a,b,w) &= \sqrt{ \confmap'(w) \, \confmap'(z) } \, F_\disc( \confmap(z); \confmap(a), \confmap(b) ) \\
\tilde F(z; a,b,w) &= \sqrt{ \confmap'(z) } \, \tilde F_\disc( \confmap(z); \confmap(a), \confmap(b) )
\end{align}
with $\coeffa$ and $\coeffb$ given in terms of $\arg u <\arg v<\arg u + 2\pi$ where $u=\confmap(a)$ and $v=\confmap(b)$
as
\begin{align}
\coeffa &=  2 \cos \left( \frac{\arg v-\arg u}{2} \right) 
   \label{eq: def coeffa}  \\
\coeffb &= 2 \sin \left( \frac{\arg v-\arg u}{2} \right) \, \cos^2 \left( \frac{\arg u+\arg v + \pi}{4} \right) .
   \label{eq: def coeffb}
\end{align}
The degenerate cases $a=b$ are obtained as limit of the formulas~\eqref{eq: def coeffa} and
\eqref{eq: def coeffb} as  $\arg v - \arg u$ tends to $0$ or $2\pi$.

For general direction of $\varmedge$, the same claim holds except that
\begin{align}
F(z; a,b,w) &= \vartheta(\varmedge) 
  \sqrt{ \confmap'(w) \, \confmap'(z) } 
  \, F_\disc\left( (\hat \varmedge)^{-1}\confmap(z) ; 
         (\hat \varmedge)^{-1}\confmap(a) , (\hat \varmedge)^{-1}\confmap(b)  \right) 
\end{align}
and $u$ and $v$ are replaced by $u=(\hat \varmedge)^{-1}\confmap(a)$ and $v=(\hat \varmedge)^{-1}\confmap(b)$
in \eqref{eq: def coeffa} and \eqref{eq: def coeffb}.
\end{theorem}

The proof is given in Section~\ref{ssec: conv obs}, except the ``algebraic part'' which is 
presented next.

\subsection{Determination of the coefficients 
  \texorpdfstring{$\coeffa$}{alpha} and \texorpdfstring{$\coeffb$}{beta}}\label{ssec: coeff a and b}

In this section, we assume that the scaling limit of the observable is 
of the form \eqref{eq: disc int observable} and we focus on determining
the coefficients $\coeffa$ and $\coeffb$ under a hypothesis (called $(\ast)$ or $(\ast')$ below)
which we verify in the proof of Theorem~\ref{thm: conv obs}.

\subsubsection{Zeros of $Q$ and the coefficients $\coeffa$ and $\coeffb$ in the case $u\neq v$}

Write 
\begin{equation*}
F_\disc(z) = \frac{2}{\pi}\sqrt{ \frac{Q(z)}{ z^2\,(z-u)\,(z-v)} }
\end{equation*}
where $F_\disc$ is as in \eqref{eq: disc int observable}.

In this section we expand $Q$ as
\begin{align}\label{eq: q in terms of coeffa and coeffb}
Q(z) =& (-z^2 + i\,\coeffa\,z+1)(z-u)(z-v) + \coeffb \, (u-v) z^2 \nonumber\\
     =& -z^4+(i \, \coeffa+u+v) z^3+(1-i\, \coeffa \,(u+v)-u \, v + \coeffb \, (u-v)) z^2 \nonumber\\
     &  +(-u-v + i\,\coeffa\,u\,v) z + u \, v
\end{align}
which we compare to another expression later.

Now we claim that 
\begin{quote}
\assumpa{}
$Q$ has to have two zeros $n$ and $m$, both of multiplicity two, such that one of them lies on the arc $uv$ and the other one on $vu$. 
\end{quote}
We will verify the claim \assumparef{} 
in the proof of Theorem~\ref{thm: conv obs} in Section~\ref{ssec: conv obs}.
Basically it results from the fact that the singularities of $F$ at $a$ and $b$ are of
the same type as the singularities of $\tilde F$ at $a$ and $b$. Thus we define the coefficient
in front of that singularity, say, at $a$ by comparing $F$ to $\tilde F$ and analyze the relative
signs of those coefficients.

The observation \assumparef{}
makes it possible to determine $\coeffa$ and $\coeffb$ in terms of $u$ and $v$.
Let's write
\begin{equation*}
u=e^{i\argu}, \quad v=e^{i\argv}, \quad \argu < \argv < \argu+2\pi .
\end{equation*}
Expand $Q$ as
\begin{align}\label{eq: q in terms of m and n}
Q(z) =& -(z-m)^2 (z-n)^2 \nonumber\\
     =& -z^4 + 2(m+n) z^3 -(m^2 + n^2 + 4 m n)z^2 + 2 m n (m+n) z - m^2 n^2 
\end{align}
and compare \eqref{eq: q in terms of coeffa and coeffb} and \eqref{eq: q in terms of m and n}.
If we ignore for the time being the coefficient of $z^2$, we have to solve the equation system
\begin{align}
i \, \coeffa+u+v &= 2(m+n) \label{eq: Q coeff z3}\\
i\,\coeffa\,u\,v-u-v  &= 2 m n (m+n) \label{eq: Q coeff z1}\\
u \, v &= - m^2 n^2 \label{eq: Q coeff z0}
\end{align}
for $\coeffa, m$ and $n$. Let $\auxvara \in \C$ be such that $\auxvara^2 = -u\,v$ and let's suppose that 
\begin{equation} \label{eq: def auxvar 1}
m\,n=\auxvara . 
\end{equation}
Then \eqref{eq: Q coeff z0} is satisfied. We resolve the choice of $\auxvara$ later.

Now by \eqref{eq: Q coeff z3} and \eqref{eq: Q coeff z1}
\begin{equation*}
i\,\coeffa\,u\,v-u-v =  \auxvara (i \, \coeffa+u+v)
\end{equation*}
When $w \neq -1$, this gives
\begin{equation*}
\coeffa = i \, \auxvara^{-1} (u+v) .
\end{equation*}
Plugging this back in \eqref{eq: Q coeff z3} gives
\begin{equation} \label{eq: solution of m+n}
m+n = \frac{1}{2} (1-\auxvara^{-1})(u+v) =\dd \auxvarb
\end{equation}
Then 
\begin{equation*}
m= \frac{\auxvarb + \sqrt{\auxvarb^2 - 4\auxvara}}{2}, \qquad n= \frac{\auxvarb - \sqrt{\auxvarb^2 - 4\auxvara}}{2}.
\end{equation*}
It's not necessary to solve these equations explicitly. 
It suffices to verify later that $\auxvarb \in \R \sqrt{\auxvara}$ and $\auxvarb^2 - 4\auxvara \in \R_- \auxvara$.
Then $|m|=|n|=1$ and both arcs $uv$ and $vu$ contain one of the points $m,n$. 

Solve next $\coeffb$ from the coefficient of $Q$ and using \eqref{eq: def auxvar 1} and \eqref{eq: solution of m+n}
\begin{align*}
-(u-v)\coeffb &= m^2 + n^2 + 4 m n + 1-i\, \coeffa \,(u+v)-u v \nonumber \\
        &= \frac{1}{4} (1-\auxvara^{-1})^2 (u+v)^2 + 2 \auxvara + 1 + \auxvara^{-1}(u+v)^2 + \auxvara^2 \nonumber \\
        &= \frac{1}{4} (1 + \auxvara^{-1})^2 (u+v)^2 + (1+\auxvara)^2 \nonumber \\
        &= \frac{1}{4} \left(4 + \auxvara^{-2} (u+v)^2 \right) (1+\auxvara)^2 = -\frac{1}{4} \frac{(u-v)^2}{uv} (1+\auxvara)^2
\end{align*}
Use the explicit formula $\auxvara = \exp \left( i \frac{\argu+\argv \pm \pi}{2} \right)$
to write that
\begin{align*}
\coeffb &= \frac{1}{4} \frac{u-v}{uv} (1+\auxvara)^2 \nonumber \\
  &= \frac{1}{4} e^{-i \frac{\argu+\argv}{2}} \cdot \left[-2i \sin \left( \frac{\argv-\argu}{2} \right) \right] \cdot  
      e^{i \frac{\argu+\argv \pm \pi}{2}}  \cdot \left[2 \cos \left( \frac{\argu+\argv \pm \pi}{4} \right) \right]^2
\end{align*}
Hence $\coeffb \geq 0$ only when $\auxvara = \exp \left( i \frac{\argu+\argv + \pi}{2} \right)$.
Thus we have show the following result.

\begin{proposition}\label{prop: coefficients a and b}
A function of the form \eqref{eq: disc int observable} satisfies 
$(\ast)$
if and only if 
\begin{align*}
\coeffa &=  2 \cos \left( \frac{\argv-\argu}{2} \right) \\
\coeffb &= 2 \sin \left( \frac{\argv-\argu}{2} \right) \, \cos^2 \left( \frac{\argu+\argv + \pi}{4} \right)
\end{align*}
where $\argu = \arg u$, $\argv = \arg v$ and they satisfy $\argu < \argv < \argu + 2 \pi$.
\end{proposition}

\subsubsection{The special case \texorpdfstring{$u=v$}{u=v}}

If $u=v$, then $\coeffb=0$. 
Write 
\begin{equation*}
F_\disc(z) = \frac{2}{\pi} \sqrt{ \frac{\hat Q(z)}{ z^2} }
\end{equation*}
where 
\begin{equation}\label{eq: hat q in terms coeffa}
\tilde Q(z) = - z^2 + i \coeffa z + 1
\end{equation}

A similar claim as \assumparef{} states that
\begin{quote}
\assumpb{} $\tilde Q$ has a zero $n$ of multiplicity two. When $\arg u = \arg v$, then $\partial_\nu H_\disc > 0$ 
piecewise everywhere and when $\arg v = \arg u +2\pi$, then $\partial_\nu H_\disc < 0$ 
piecewise everywhere.
\end{quote}
We will verify the claim \assumpbref{} also
in the proof of Theorem~\ref{thm: conv obs} in Section~\ref{ssec: conv obs}.
Here $\partial_\nu$ is the derivative to the direction of the outer normal at the boundary of $\disc$.

Write
\begin{equation}\label{eq: hat q in terms n}
\tilde Q(z) = -(z-n)^2 = -z^2 + 2 n z - n^2 .
\end{equation}
Then $n=\pm i$ and $\coeffa = \pm 2$ by comparing \eqref{eq: hat q in terms coeffa}
and \eqref{eq: hat q in terms n}.

Write
\begin{equation*}
H_\disc(z) = \imag \left( -z - \frac{1}{z} + i \coeffa \log z + const.\right) .
\end{equation*}
Since $\partial_\nu = r \partial_r = x \partial_x + y \partial_y = z \partial +\overline{z} \overline{\partial}$,
\begin{equation*}
\partial_\nu H_\disc(z) = \imag \left( -z + \frac{1}{z} + i \coeffa \right) = \coeffa - 2\sin \arg z
\end{equation*}
for any $z \in \partial \disc$.
When $\coeffa=2$, then $n= i$ and $\partial_\nu H_\disc>0$ except at $z=n$, where as
when $\coeffa=-2$, then $n= -i$ and $\partial_\nu H_\disc<0$ except at $z=n$.
Notice that these are consistent with  taking
the corresponding limits of the formulas in Proposition~\ref{prop: coefficients a and b}.

\section{The scaling limit of the observable}\label{sec: observable limit}

\subsection{The operators 
  \texorpdfstring{$\Delta$}{Delta} and \texorpdfstring{$\partial$}{partial}, Green's function etc.}

On a square lattice $\varggr$ define the linear operator $\Delta_1^\varggr: \C^{V(\varggr)} \to \C^{V(\varggr)}$ 
called unnormalized discrete Laplace operator as
\begin{equation*}
\Delta_1^\varggr \, H(z) = \sum_{w \sim z} H(w)-H(z).
\end{equation*}
Similarly for define the unnormalized discrete versions of the complex derivatives
$\cd$ and $\cdbar$ as operators
$\cd_1^\varggr: \C^{V(\varggr)} \to \C^{V(\varggr^*)}$ 
and $\cdbar_1^\varggr: \C^{V(\varggr)} \to \C^{V(\varggr^*)}$
defined by
\begin{equation*}
\cd_1^\varggr \, H(z) = \frac{1}{2} \sum_{w \stackrel{\phantom{^*}\varggr\varggr^*}{\sim} z} 
   \frac{\bar{w}-\bar{z}}{|w-z|}  H(w)
\end{equation*}
and
\begin{equation*}
\cdbar_1^\varggr \, H(z) = \frac{1}{2} \sum_{w \stackrel{\phantom{^*}\varggr\varggr^*}{\sim} z} 
   \frac{w-z}{|w-z|} H(w) 
\end{equation*}
where $w\stackrel{\phantom{^*}\varggr\varggr^*}{\sim}z$ means that
$z$ and $w$ are neighbors on the square lattice with vertices $V(\Gamma) \cup V(\Gamma^*)$.
Notice that $\cd_1^{\varggr^*}$ and $\cdbar_1^{\varggr^*}$ map $\C^{V(\varggr^*)}$ to $\C^{V(\varggr)}$. Usually we drop the upper index $\varggr$
and furthermore for instance, denote both $\cd_1^{\varggr}$ and $\cd_1^{\varggr^*}$ by $\cd_1$. There is no possibility
for confusion since they operate on different spaces.

Let's apply the above operators in a slightly more concrete situation.
Suppose now that $\varggr_\bullet$ 
is a square lattice which is $\sqrt{2} \delta \Z^2$ rotated by the angle $\pi/4$
and
$\varggr_\circ$ its dual lattice.
Set $\varggr$ to be the square lattice with vertices 
$V(\varggr_\bullet) \cup V(\varggr_\circ)$.
Let $\Delta_{1,\bullet}$ be the Laplace operator of $\varggr_\bullet$.
It standard to verify the following result.

\begin{lemma}
$\cdbar_1 \cd_1 |_{V(\varggr_\bullet)}=\cd_1 \cdbar_1 |_{V(\varggr_\bullet)} = \frac{1}{4} \Delta_{1,\bullet}$
\end{lemma}

To get operators that correspond to the continuum operators, 
we define
\begin{align*}
\Delta_{\bullet} =  \frac{1}{2 \delta^2} \Delta_{1,\bullet} ,\qquad
\cd = \frac{1}{\sqrt{2} \delta} \cd_1 ,\qquad
\cdbar = \frac{1}{\sqrt{2} \delta} \cdbar_1 .
\end{align*}

The next result gives the existence of discrete Green's function. The proof can be found in
\cite{Kenyon:2002ji}.

\begin{proposition}
For each $z_0 \in V(\varggr_\bullet)$, there exists function 
a unique function $G_{z_0}: \C^{ V(\varggr_\bullet)} \to \C$ such that
$G_{z_0}(z_0)=0$, $\Delta_{1,\bullet} G_{z_0}(\,\cdot\,) = \delta_{\cdot,z_0}$ 
and $G_{z_0}$ grows sublinearly at infinity.
It satisfies the asymptotic equality
\begin{equation}\label{eq: green asymptotics}
G_{z_0}(z)= \frac{1}{2 \pi} \log \left( \frac{|z-z_0|}{\sqrt{2}\delta} \right) + C + \OO\left( \frac{\delta^2}{|z-z_0|^2} \right)
 \qquad \text{as } z\to\infty .
\end{equation}
\end{proposition}

Extend $G_{z_0}$ holomorphically to $\varggr_\circ$, that is,
suppose that it satisfies $\cdbar_1 G_{z_0} = 0$. 
The extension is defined up to an additive constant and
well-defined locally, but globally it might be multivalued.
However $\mathcal{C} \dd= \cd_1 G_{z_0}$ is well-defined and single-valued globally. 
It satisfies for $z \in \Gamma^*$,
\begin{equation}\label{eq: cauchy kernel difference}
\mathcal{C}(z) =  \sigma(z) D_\bullet G_{z_0}(z)
\end{equation}
where $D_\bullet$ is the (unnormalized) 
difference operator along the edge of $\varggr_\bullet$ going through the site $z$ (which is the midpoint of the edge)
and $\sigma(z)$ takes values $e^{i(\frac{\pi}{4} + k \frac{\pi}{2})}$, $k=0,1,2,3$, depending only on the direction and orientation
of the edge.
Remember that $\lambda = e^{-i\frac{\pi}{4}}$. 
Thus the values of $\mathcal{C}(z)$ on $\Gamma^*$ 
are restricted on the lines $\lambda \R$ and $\bar{\lambda} \R$.

Next take a square lattice $\varggr_\diamond$ so that the vertices $V(\Gamma)$ are
the midpoints of the horizontal edges of $\varggr_\diamond$.
Then the vertices $V(\varggr^*)$ are the midpoints of the vertical edges of $\varggr_\diamond$
and we can define $\mathcal{C}$ 
on the vertical edges to be the value of $\mathcal{C}$ on $\varggr^*$ constructed above.
Define $\mathcal{C}(z)$ for any $z \in V(\varggr_\diamond)$ to be the sum $\mathcal{C}$ on the two vertical edges ending to $z$.

Remember the notation $\hat e = (e_+ - e_-)/|e_+ - e_-|$ and 
define subsets 
\begin{equation}
E_\alpha = \{e \in E(\delta\dmdlattice) \,:\, \hat e =  \alpha \}
\end{equation}
for $\alpha =  \pm 1, \pm\ii$. Here $e_\to$ is $e$ with is canonical orientation
(counterclockwise around each black square).
Let $\varmedge \in E(\delta \odmdlattice)$. Suppose that $V(\Gamma_\bullet)$ are 
on the midpoints of $E_{\hat \varmedge}$. Then $V(\Gamma_\circ)$
are on the midpoints of $E_{-\hat \varmedge}$
and similarly $V(\Gamma^*)$ are midpoints of 
$E_{ \ii \hat \varmedge} \cup E_{ -\ii\hat \varmedge}$
where we use the notation that $\alpha \varmedge$ is $\varmedge$
rotated by the complex factor $\alpha$.
Define $\mathcal{C}_{\varmedge,\delta}(e)$ for any 
$e = \{v,w\} \in E_{ \ii \hat \varmedge} \cup E_{ -\ii\hat \varmedge}$
to be $\mathcal{C}( \frac{1}{2}(v+w) )$ constructed above for $\Gamma^*$.

\begin{definition}
For any $z \in V(\dmdlattice)$, define
$\mathcal{C}_{\varmedge,\delta}(z)$ to be
the sum of the values $\mathcal{C}_{\varmedge,\delta}(e)$ on the two edges in
$e \in E_{ \ii \hat \varmedge} \cup E_{ -\ii\hat \varmedge}$
that end to $z$.
\end{definition}

\begin{proposition}
The discrete Cauchy kernel $\mathcal{C}_{\varmedge,\delta}: V(\delta\dmdlattice) \to \C$ 
satisfies the following properties.
\begin{enumerate}
\item $\overline{\vartheta(\varmedge)} \mathcal{C}_{\varmedge,\delta}$ is spin preholomorphic everywhere except 
at $\varmedge$.
\item At $\varmedge$, 
$\mathcal{C}_{\varmedge,\delta}(\varmedge_+)=-\mathcal{C}_{\varmedge,\delta}(\varmedge_-)= \frac{1}{2\sqrt{2}}
(\hat \varmedge)^{-1}$. 
\item The asymptotic equality
\begin{equation}
\mathcal{C}_{\varmedge,\delta}(z)
  = \frac{\sqrt{2}}{2 \pi}  \frac{\delta}{z-z_0} +  \OO\left( \frac{\delta^2}{|z-z_0|^2} \right)
\end{equation}
holds for $z \in V(\delta \dmdlattice)$ as $z \to \infty$.
\end{enumerate}
\end{proposition}

\begin{proof}
(1) The prefactor $\overline{\vartheta(\varmedge)}$ is needed for \eqref{eq: spin holomorphic}.
With this prefactor it is easy to verify that on $E_{ \ii \hat \varmedge} \cup E_{ -\ii\hat \varmedge}$
the complex phase of $\mathcal{C}_{\varmedge,\delta}(e)$ is $\vartheta(e)$. Thus
\eqref{eq: spin holomorphic} is satisfied on $E_{ \ii \hat \varmedge} \cup E_{ -\ii\hat \varmedge}$.
It is easy to verify that
the correct complex phase on $E_{ \ii \hat \varmedge} \cup E_{ -\ii\hat \varmedge}$
and the fact, that values of $\mathcal{C}_{\varmedge,\delta}(e)$ 
on $E_{ \ii \hat \varmedge} \cup E_{ -\ii\hat \varmedge}$
satisfy the discrete Cauchy--Riemann equations, imply that
\eqref{eq: spin holomorphic} holds on $E_{  \hat \varmedge} \cup E_{ -\hat \varmedge}$.

(2) The values of Green's function on the four neighboring sites of the source are by symmetry
equal to $\frac{1}{4}$. Thus at the four edges 
$e =e_{SW},e_{SE},e_{NE},e_{NW}$ in $E_{ \ii \hat \varmedge} \cup E_{ -\ii\hat \varmedge}$ 
which are neighbors of $\varmedge$, the value of $\mathcal{C}_{\varmedge,\delta}(e)$
is equal to $-\lambda \frac{1}{4},\bar\lambda \frac{1}{4},\lambda \frac{1}{4},-\bar\lambda \frac{1}{4}$,
respectively. The claim follows easily.

(3) The claim follows from the definition of $\mathcal{C}_{\varmedge,\delta}(z)$ at a vertex as the sum
of the values at two neighboring edges. At the edges,
use \eqref{eq: green asymptotics} and \eqref{eq: cauchy kernel difference}.
\end{proof}

\subsection{Convergence of the observable}\label{ssec: conv obs}
  
\begin{proof}[Proof of Theorem~\ref{thm: conv obs}]
The convergence of $\delta^{-1/2} \tilde F_\delta$ to $\tilde F$ was shown in \cite{Smirnov:2010ie}.
Thus we need to only show the convergence of $\delta^{-1} F_\delta$ to $ F$.

The key element of the proof 
the convergence of $\delta^{-1/2} \tilde F_\delta$
in \cite{Smirnov:2010ie} was that $\tilde F_\delta$ is spin holomorhic.
The same argument goes through for $F_\delta$ showing that it is spin holomorphic
everywhere except at the edge $\varmedge$. 
At $\varmedge$, it fails to be spin holomorphic 
since $F_\delta(\varmedge_+) \neq F_\delta(\varmedge_-)$; in fact, $F_\delta(\varmedge_+)=-F_\delta(\varmedge_-)=1$.

Next we make the following assumption 
\begin{quote}
\assumpc{} On any compact subset of $\domain \setminus \{w\}$ the sequence of functions $(\frac{1}{\delta} F_\delta)_{\delta>0}$ is uniformly bounded. 
\end{quote}
We will later show that the assumption \assumpcref{} indeed holds.

Take any sequence of compact sets $K_n$ increasing to $\domain \setminus \{w\}$.
Using standard arguments of \cite{ChelkakSmirnov:2012} for spin preholomorphic functions, 
we can show that $(\frac{1}{\delta} F_\delta)_{\delta>0}$ are equicontinuous on any $K_n$ and hence
by a diagonal argument any subsequence of  $(\frac{1}{\delta} F_\delta)_{\delta>0}$ contains a subsequence which converges uniformly on
any compact subset of $\domain \setminus \{w\}$.
Let $F = \lim_{n \to \infty } \frac{1}{\delta_n} F_{\delta_n}$.

Let $\tilde{\mathcal{C}}_{\varmedge,\delta}$ be the discrete Cauchy kernel 
introduced above
with the ``singularity'' at $w$
scaled by a factor $2\sqrt{2} \; \overline{\vartheta(\varmedge)}$. 
Write $\hat F_\delta = F_\delta - \tilde{\mathcal{C}}_{\varmedge,\delta}$.
Then we easily see that 
$\hat F_\delta$ is spin holomorphic in $\domain_\delta$ including $\varmedge$.
We can assume that $\partial B(w,r) \subset K_n$ for some $n$. 
If $\frac{1}{\delta} |F_\delta| \leq M$ in $K_n$ and $\frac{1}{\delta} |\mathcal{C}_{\delta}| \leq C$ near $\partial B(w, r)$,
then $\frac{1}{\delta} |\hat F_\delta| \leq M + C $ near $\partial B(w, r)$.
By summing 
the function $\frac{1}{\delta} \hat F_\delta$ against the Cauchy kernel
(discrete Cauchy formula, see Proposition~2.22 of \cite{Chelkak:2011by}), 
we can extend this estimate
to the interior of $B(w, r)$. 
Therefore $\frac{1}{\delta} \hat F_\delta$ remains bounded on any compact subset of $\domain$
and we can extract a subsequence $\frac{1}{\delta_n} \hat F_{\delta_n}$
that converges on any compact subset of $\domain$.
Thus the subsequence $\frac{1}{\delta_n} F_{\delta_n}$ converges to a function
of the form
\begin{equation}\label{eq: general form of f}
F(z) = \overline{\vartheta(\varmedge)} \left[ \frac{2}{\pi} \frac{1}{z - w} + \hat F(z) \right]
\end{equation}
where $\hat F$ is a holomorphic function on $\domain$.
To simplify the notation, absorb the factor $\overline{\vartheta(\varmedge)}$
to the variables $z,a,b,w$ by setting $z'=z\overline{\vartheta(\varmedge)}$ etc., but let's
continue to denote the points $z,a,b,w$. Then
\begin{equation}\label{eq: general form of f simplified}
F(z) =  \frac{2}{\pi} \frac{1}{z - w} + \hat F(z)  .
\end{equation}
Notice that the change of variables 
corresponds to rotating the picture so that the edge $\varmedge$ points to the east,
which we can therefore assume without loss of generality.

By the assumption \assumpc{},
$\frac{1}{\delta} H_\delta$ is uniformly bounded on compact subsets of $\domain \setminus \{w\}$ where
$H_\delta$ is defined as in \eqref{eq: H def}.
Hence $\frac{1}{\delta} H_\delta$ is equicontinuous on compact subset by arguments
of \cite{ChelkakSmirnov:2012}
and we can extract a subsequence that converges uniformly on any of the set $K_n$. We can assume that this sequence is $\delta_n$
chosen above. 
By \eqref{eq: general form of f} and by the fact that there is no monodromy around $w$,
the limit of $\frac{1}{\delta_n}H_{\delta_n}$ has to be of the form 
\begin{equation*}
H = \lim_{n \to \infty} \frac{1}{\delta_n}H_{\delta_n}
 = 
   \frac{2}{\pi^2}\left[ \imag \left( -  \frac{1}{z - w} + i \coeffa \log(z - w)\right) + \hat{H}(z) \right]
\end{equation*}
where $\coeffa \in \R$ is a constant and $\hat{H}$ is a harmonic function on $\domain$.
Since the boundary conditions of $H_\delta$ where Dirichlet on both boundary arcs, 
say, equal to $\zeta_\delta$ and $\xi_\delta$
with $\coeffb_\delta \dd = \xi_\delta - \zeta_\delta>0$,
the quantities $\frac{1}{\delta_n} \zeta_{\delta_n}$ and $\frac{1}{\delta_n}\xi_{\delta_n}$
must remain bounded. Otherwise $\frac{1}{\delta_n}H_{\delta_n}$ wouldn't converge.
(Essentially,
near that part of the boundary, where the values go to $\pm \infty$ with the quickest rate, 
the uniform boundedness in compact subsets and the large boundary values are in contradiction.)
Hence $H$ has piecewise constant boundary values, i.e. it satisfies Dirichlet boundary conditions
which can be described in the following way: if $\confmap : \domain \to \disc$ is conformal and onto and such that $\confmap(w)=0$ and
$\confmap'(w)>0$, then $H=H_\disc \circ \confmap$ 
where
\begin{equation}\label{eq: disc int H}
H_\disc =  
  \frac{2}{\pi^2} \imag \left( -  \frac{1}{z} - z + i \coeffa \log z + \coeffb \log \frac{z-u}{z-v}\right) + const.
\end{equation}
We have to show that $\coeffa \in \R$ and $\coeffb \geq 0$ are uniquely determined. We do this by
showing that the assumptions ($\ast$) and ($\ast'$) of Section~\ref{ssec: coeff a and b} hold.

First we will observe that $F$ is single valued.
This follows directly from the properties $F_\delta$.

Next we write $\coeffb_\delta= |\tilde F_\delta(\varmedge)|^2$.
Counting the changes in the number of loops when $\inneredge$ is removed from the loop configuration
of $\Ginner$ and $\outeredge$ is added, shows that 
$\coeffb_\delta= |F_\delta(\medge_o)|^2 = |F_\delta(\medge_i)|^2$.
The convergence of $\tilde F_\delta$ shows that
$\delta^{-1} \coeffb_\delta \to \coeffb_0$ for some $\coeffb_0>0$
(notice that $\coeffb_0 = |\real \tilde F_\disc(0)| = \frac{2}{\pi} \coeffb$
in the continuum, see also Remark~\ref{rem: scaling limit observables}).

Let $\eps_\delta=F_\delta(\medge_o)/\tilde F_\delta(\medge_o)$. 
Notice that $\eps_\delta$ is real-valued. 
Then $|\eps_\delta |^2 = \beta_\delta$ by the above argument
and $\eps_\delta=- F_\delta(\medge_i)/\tilde F_\delta(\medge_i)$ 
by the fact that the two open paths in the loop configuration concatenated with
$\inneredge$ and $\outeredge$ form together a closed simple loop, 
which thus makes one full $\pm 2 \pi$ turn.

Next we will notice that by using $\tilde F_\delta$, we can define
\begin{equation*}
F^{(a)}_\delta = F_\delta + \eps_\delta \tilde F_\delta , \qquad
F^{(b)}_\delta = F_\delta - \eps_\delta \tilde F_\delta
\end{equation*}
which are discrete holomorphic in $\domain_\delta \setminus \{\varmedge\}$.
The corresponding functions 
$H^{(a)}_\delta$ and $H^{(b)}_\delta$ are approximately discrete harmonic
in the same sense as $H$. The function $H^{(a)}_\delta$
doesn't have jump at $a$ and similarly,
$H^{(b)}_\delta$ doesn't have jump at $b$.
Therefore after transforming conformally to the unit disc, it follows that
the scaling limits $F_\disc^{(a)}$ and $F_\disc^{(b)}$  can be continued holomorphically to 
neighborhoods of $a$ and $b$, respectively.

Consequently, for $\eps = \lim_{\delta \to 0} \delta^{-1/2} \eps_\delta$
\begin{align}
F_\disc(z) &= - \eps \tilde F_\disc (z) + \text{holomorphic}  &, \text{ as } z \to u
   \label{eq: f expansion at a} \\
F_\disc(z) &= + \eps \tilde F_\disc (z) + \text{holomorphic}  &, \text{ as } z \to v
   \label{eq: f expansion at b}
\end{align}
where $u = \confmap(a)$ and $v = \confmap(b)$.

Now \eqref{eq: disc int observable} follows from \eqref{eq: disc int H}.
Write $F_\disc= \sqrt{Q(z)/P(z)}$ as we did in Section~\ref{ssec: coeff a and b}.
If $F_\disc(e^{i \theta})= \phi(\theta) \tau(\theta)^{-1/2}$, then
the only way that the properties \eqref{eq: f expansion at a} and \eqref{eq: f expansion at b} 
can be satisfied is that
$\phi(\theta)$ is zero for some $\theta=\theta_1 \in (\argu,\argv)$ as well as
for some  $\theta=\theta_2 \in (\argv,\argu+2\pi)$. Here $\argu = \arg u$ and $\argv = \arg v$.

Consequently, it follows that $Q(z)$ has root of order $2 + 4 k_1$ or  $2 + 4 k_2$
for some integer $k_j$, $j=1,2$,
at $m= e^{i \theta_1}$ and at $n= e^{i \theta_2}$, respectively. Since $Q$ is of order $4$,
$k_1=k_2=0$. Thus $Q(z)= - (z-m)^2 (z-n)^2$. Thus \assumparef{} follows.
A similar argument gives \assumpbref{}.

By the calculation of Section~\ref{ssec: coeff a and b} the values of $\coeffa$ and $\coeffb$ are uniquely determined and given by 
Proposition~\ref{prop: coefficients a and b}. This shows that the limit of $\frac{1}{\delta_n} H_{\delta_n}$ is unique
along any subsequence such that $\frac{1}{\delta_n} F_{\delta_n}$ and $\frac{1}{\delta_n} H_{\delta_n}$ converge
to $F$ and $H$. Since $F=\sqrt{4\ii \cd H}$, also $F$ is uniquely determined.
Since it holds that a subsequence of any subsequence of $(\frac{1}{\delta} F_\delta)$ converges 
to this unique $F_\delta$, the whole sequence converges to $F$. 
We have arrived to the claim of the theorem.

It remains to be shown that the assumption \assumpcref{} holds. Assume on contrary that 
in a compact subset $K$ of $\domain \setminus \{w\}$
the sequence
$M_n = \sup_{K} |\frac{1}{\delta_n} F_{\delta_n}|$ goes to infinity. Since increasing $K$ only increases $M_n$ we can assume that
$A(w,r/2,2r) \subset K$ for some $r>0$.
The sequence $\frac{1}{\delta_n M_n} F_{\delta_n}$ is uniformly bounded on $K$ and hence equicontinuous.

Define $\hat F_\delta = F_\delta - \mathcal{C}_{\delta}$ as we did above.
The functions $\frac{1}{\delta_n M_n} \hat F_{\delta_n}$ extend holomorphically to $\hat{K} \dd= K \cup B(w,r)$
and 
is uniformly bounded on $\hat{K}$ and hence equicontinuous. Take a subsequence, still denoted by $\delta_n$, such that 
$\frac{1}{\delta_n M_n} F_{\delta_n}$ and $\frac{1}{\delta_n M_n} \hat F_{\delta_n}$ converge to some $F$ and $\hat F$.
The functions $F$ and $\hat F$ are holomorphic and $F$ is not identically zero.

Define functions $H_\delta $ and $\hat{H}_\delta$ 
similarly as in \eqref{eq: H def}
using 
the spin holomorphic functions $F_\delta$ and $\hat F_\delta$,
respectively.
Then $\frac{1}{\delta_n M_n^2} H_{\delta_n}$ and $\frac{1}{\delta_n M_n^2} \hat{H}_{\delta_n}$
are uniformly bounded in $K$ and $\hat{K}$ respectively. It follows from the fact 
that the boundary values of $H$ are piecewise constant and
that $\frac{1}{\delta_n M_n^2} H_{\delta_n}$ uniformly bounded on
$\domain \setminus B(w,r)$. Hence we can take a subsequence (still denoted by $\delta_n$)
such that  $\frac{1}{\delta_n M_n^2} H_{\delta_n}$ and $\frac{1}{\delta_n M_n^2} \hat{H}_{\delta_n}$
converge to some function $H$ and $\hat{H}$, respectively. 
We can assume that the former converges uniformly on any compact subset of $\overline{\domain} \setminus \{w,a,b\}$
and the latter on $\hat{K}$, which includes a neighborhood $w$.
Now 
\begin{equation*}
H = \frac{1}{2} \imag 
\int \lim_{n \to \infty}  \frac{1}{\delta_n M_n^2} \left(\mathcal{C}_{\delta} + \hat F\right)^2 \de z
= \frac{1}{2}  \imag \lim_{n \to \infty}  \int \frac{1}{\delta_n M_n^2} \hat F^2 \de z
= H_0
\end{equation*}
on $\hat{K} \setminus \{w\}$. It follows that $H$ extends harmonically to $w$ and hence
$H=H_\disc \circ \confmap$ where $\confmap:\domain \to \disc$ is conformal and
\begin{equation*}
H_\disc (z) = \coeffb \imag \log \frac{z-u}{z-v} +const.
\end{equation*}
for some $\beta \geq 0$.
To reach a contradiction, we will show that $H_\disc$ has a critical point somewhere in $\overline{\disc} \setminus \{u,v\}$.

Notice that the boundary values of $H_\delta$ are $\zeta,\xi,\eta$ and $\eta+1$
on the arcs $a_\delta b_\delta$ and $b_\delta a_\delta$ and
at the points
$w_\delta-i\delta/2$ and $w_\delta + i\delta/2$, respectively.
We know that $0< \xi-\zeta < 1$. Therefore
either $\xi<\eta+1$ or $\eta<\zeta$. Suppose that the former happens. The other case can be dealt with in a similar manner.

By applying maximum principle to $H_\delta^\bullet$ , it follows that
there exists from any point $z$ 
a path to the boundary of the domain or to $w_\delta$ such that $H_\delta^\bullet$
is strictly increasing along the path.
By the values of the normal derivative on the boundary the path can only hit $b_\delta a_\delta$ or $w_\delta$. Furthermore
the points that can be connected to $b_\delta a_\delta$ form a connected set and likewise the points that can be connected to $w_\delta$
and those sets exhaust the whole set of vertices. Define the boundary $I$ between those sets as being the set of edges
whose one end is in one of the sets and the other is in the other set.
Let $x^*$ be the vertex in $I$ that has the maximal value of $H_\delta^\bullet$ in that set.
Since $\eta+1>\xi$, the point $x^*$ can't be close to $w_\delta$ for small $\delta$. By compactness of $\overline{\domain}$ we can suppose that 
$x^*_{\delta_n}$ converges as $n \to \infty$ to some point $x^* \in \overline{\domain} \setminus \{w\}$.

Suppose that $x^* \in \domain \setminus \{w\}$. Then for any $r>0$ such that $B(x^*,r) \subset \domain \setminus \{w\}$
it holds that $H_\delta^\bullet - H_\delta^\bullet(x^*_\delta)$ changes sign at least $4$ times along
$\partial B(x^*,r)$. Furthermore the angles at which the peaks and valleys appear are bounded away from each other.
By convergence of $\frac{1}{\delta_n M_n^2} H_{\delta_n}$ to $H$ to this continuous to hold for $x^*$.
Since $H= \imag \Psi$ for some holomorphic $\Psi$ in the neighborhood, we find that
$\Psi'(x^*)=0$ and $x^*$ is a saddle point for $H$. This is a condradiction.

A similar conclusion on a contradiction can be made for $x^* \in \partial \domain$.
Thus \assumpcref{} holds.
\end{proof}

Finally, we need still the following result on the convergence of the observables
which can be extracted from the above proof. 
The proof of the conjecture is similar to the beginning of the proof of Theorem~4.3 in
\cite{ChelkakSmirnov:2012}.

\begin{corollary}[Uniform convergence of observables over a class of domains]
\label{cor: uniform conv obs}
The convergence in Theorem~\ref{thm: conv obs} 
is uniform with respect to $(\domain;a,b,w)$,
whenever $B(w,r) \subset \domain \subset B(w,R)$. More specifically,
for each $\delta>0$, there is $\eps(\delta)>0$ such that $\eps(\delta)\to 0$ as $\delta \to 0$ and
\begin{equation}
\sup_{z \in (\domain_\delta)_{w_{\delta},r}}
  |\delta^{-1} F_\delta(z; \domain_\delta,a_\delta,b_\delta,w_\delta) - F(z; \domain_\delta,a_\delta,b_\delta,w_\delta)| 
     <     \eps(\delta)
\end{equation}
for all discrete domains $(\domain_\delta,a_\delta,b_\delta,w_\delta)$ 
such that $B(w_\delta,r) \subset \domain \subset B(w_\delta,R)$.
Here $(\domain_\delta)_{w_{\delta},r}$ is the $r$-connected neighborhood of $w_{\delta}$ in $\domain_\delta$.
The same claim holds when $F_\delta$ is replaced by $\tilde F_{\delta}$.
\end{corollary}

\begin{proof}
Assume the opposite: there is $\eps>0$ and sequence $(\domain_{\delta_n},a_{\delta_n},b_{\delta_n},w_{\delta_n})$
and $z_n \in  (\domain_{\delta_n})_{w_{\delta_n},r}$
such that 
$$|\delta_n^{-1} F_{\delta_n}(z_n; \domain_{\delta_n},a_{\delta_n},b_{\delta_n},w_{\delta_n}) 
  - F(z_n; \domain_{\delta_n},a_{\delta_n},b_{\delta_n},w_{\delta_n})| \geq \eps .$$ 
The claim for $\tilde F_{\delta}$ can be proven using a similar assumption.

The sequence $(\domain_{\delta_n},a_{\delta_n},b_{\delta_n},w_{\delta_n})$ is compact in the Carath\'eodory convergence
as well as the sequence $z_n$. Thus we can choose a subsequence of $\delta_n$ such that they converge,
say, to $(\domain,a,b,w)$ and $z$, respectively.
The proof of Theorem~\ref{thm: conv obs} implies that 
$\delta_n^{-1} F_{\delta_n}(z_n; \domain_{\delta_n},a_{\delta_n},b_{\delta_n},w_{\delta_n})$
will converge to $F(z; \domain,a,b,w)$. Moreover, from the explicit formula for $F$, we see that
$F(z_n; \domain_{\delta_n},a_{\delta_n},b_{\delta_n},w_{\delta_n})$ converges to
$F(z; \domain,a,b,w)$. This leads to a contradiction.
\end{proof}

\section{A priori bounds for exploration trees and loop ensembles}\label{sec: a priori}

In this section, we present some results which we classify as
a priori results. They describe properties that ensure the regularity
of branches, trees and loops in a manner that is needed for their convergence.
We use often the framework of estimates established in \cite{Kemppainen:2012vma}
and thus, in this particular case, ultimately
rely on the estimates of FK Ising crossing probabilities of \cite{Chelkak:2016jw}
for general topological quadrilaterals.

\subsection{One-to-one correspondence of the tree and the loop-ensemble in the limit}
 \label{ssec: main result a priori}

In the discrete setting we are given a tree--loop ensemble pair. The tree and the loop ensemble
are in one-to-one correspondence as explained earlier. Recall that, 
\begin{itemize}
\item
given the loop ensemble,
the tree is constructed by the exploration process which follows the loops
and jumps to the next loop at points where the followed loop  turns away from
the target point. 
\item
the loops are recovered from the tree by noticing
that each loop corresponds to a small square where the branching to
that loop occurs (this correspondence is $1$-to-$1$, when
we also count the root as one of the branching points). 
We take the incoming edge, which is opposite to 
the other incoming edge that we used to arrive to the
small square for the first time, and select the branch corresponding to that
target edge. The loop is constructed from the branch by taking the part between the
first exit and last arrival and then 
adding to the path the 
edge of the small square that closes it to a loop.
\end{itemize}

We will show that the probability laws that we are considering form
a precompact set in the topology of weak convergence of probability measures.
Take a subsequence of the sequence of the tree -- loop ensemble pairs that converges weakly.
We can choose a probability space so that they converge almost surely.
Next theorem summarizes the convergence of the tree -- loop ensemble pair.
The second assertion basically means that
there is a way to \emph{reconstruct the loops from the tree} also in the limit.

\begin{theorem}\label{thm:  tree to loops in the limit}
Let $(\rndlpe^\disc,\rndttr^\disc)$ be the almost sure limit of $(\rndlpe_{\delta_n}^\disc,\rndttr_{\delta_n}^\disc)$
as $n \to \infty$.
Write $\rndlpe^\disc =(\rndlp_j)_{j \in J}$ and $\rndlpe_{\delta_n}^\disc = (\rndlp_{n,j})_{j \in J}$
(with possible repetitions) such that almost surely for all $j \in J$, $\rndlp_{n,j}$ converges
to $\rndlp_j$ as $n \to \infty$, and then set $x_{n,j}$ to be the target point 
of the branch of 
$\rndttr_{\delta_n}^\disc$ that corresponds to $\rndlp_{n,j}$ in the above bijection 
(described in the beginning of Section~\ref{ssec: main result a priori}).
Then 
\begin{itemize}
\item Almost surely,
the point $x_{n,j}$ converge to some point $x_{j}$ as $n\to \infty$ and
the branch $\rndbran_{x_{n,j}^+}$ converge to some branch denoted by $\rndbran_{x_{j}^+}$ 
as $n\to \infty$ for all $j$.
Furthermore, the points $x_{j}$ that correspond to non-trivial loops ($\rndlp_j$ is not a point) 
are distinct and 
they form a dense subset of $\overline{ \disc}$ 
and $\rndttr$ is the closure of $(\rndbran_{x_{j}^+})_{j \in J}$.
\item 
On the other hand, $(\rndbran_{x_{j}^+})_{j \in J}$ is characterized as being the subset of $\rndttr$
that contains all the branches of $\rndttr$ that have a triplepoint in the bulk or a doublepoint 
on the boundary.
Furthermore, that double or triple point is unique 
and it is the target point (that is, endpoint) of that branch.
Any loop $\rndlp_j$ can be reconstructed from $(\rndbran_{x_{j}^+})_{j \in J}$ 
so that the loop $\rndlp_j$ 
is the part between the second last and last visit to $x^+_j$ by $\rndbran_{x_{j}^+}$.
\end{itemize}
\end{theorem}

The proof is postponed to Section~\ref{sec: main proof}.

\subsection{Bounds related to multiple crossings of an annulus by the tree}\label{ssec: n arms tree}

\subsubsection{The bound on multiple crossings and consequences}

Multiple crossings of annuli by random curves was considered in 
\cite{Aizenman:1999dt, Kemppainen:2012vma}.
Below we state the consequences for that type of result for the FK Ising exploration tree.
Before that we make some crucial definitions.

We call a random variable $X$ \emph{tight} over a collection of probability measures $\P$
on the probability space,
if for each $\eps>0$
there exists a constant $M>0$ such that $\P(|X|<M)>1-\eps$
for all $\P$.
Remember that a crossing of an annulus $A(z_0,r,R)=\{z \in \C \,:\, r<|z-z_0|<R\}$
is a closed segment of a curve that intersects both connected components of $\C \setminus A(z_0,r,R)$
and a minimal crossing is doesn't contain any genuine subcrossings.

Recall the general setup: we are given a collection $(\varconfmap,\P) \in \Sigma$ 
where the conformal map $\varconfmap$ 
contains also the information about its domain of definition 
$(\vardomain,\vroot,z_0)=(\vardomain(\varconfmap),\vroot(\varconfmap),z_0(\varconfmap))$ 
through the requirements
\begin{equation*}
\varconfmap^{-1}(\disc)=\vardomain, \qquad \varconfmap(\vroot)=-1 \qquad \text{and} \qquad \varconfmap(z_0)=0
\end{equation*}
and $\P$ the probability law of
FK Ising model on the discrete domain and in particular gives the distribution of
the FK Ising exploration tree.
Given the collection $\Sigma$ of pairs $(\varconfmap,\P)$ 
we define the collection $\Sigma_\disc = \{\varconfmap\P \,:\, (\varconfmap,\P)\in\Sigma \}$
where $\varconfmap\P$ is the pushforward measure defined by $(\varconfmap\P) (E) = \P(\varconfmap^{-1}(E))$. 

\begin{theorem}\label{thm: aizenman--burchard bounds}
The following claim holds for the collection of the probability laws of FK Ising exploration trees
\begin{itemize}
\item 
for any $\Delta>0$, there exists $n \in \N$ and $K>0$ such that the following holds
\begin{equation*}
\P( \text{at least $n$ disjoint segments of $\rndttr$ cross } A(z_0,r,R))
  \leq K \left( \frac{r}{R} \right)^\Delta
\end{equation*}
for all $\P \in \Sigma_\disc$.
\end{itemize}
and there exists a positive number $\alpha, \alpha'>0$ such that the following claims hold
\begin{itemize}
\item 
if for each $r>0$, $M_r$ is the minimum of all $m$ such that each 
$\rndbran \in \rndttr$ can be split into $m$
segments of diameter less or equal to $r$, then
there exists a random variable $K(\rndttr)$ such that
$K$ is a tight random variable for the family $\Sigma_\disc$
and
\begin{equation*}
M_t \leq K(\rndttr) \, r^{-\frac{1}{\alpha}}
\end{equation*}
for all $r>0$.
\item All branches of $\rndttr$ can be jointly parametrized so that they are
all $\alpha'$-H\"older continuous and the H\"older norm can be bounded
by a random variable $K'(\rndttr)$ such that $K'$  
is a tight random variable for the family $\Sigma_\disc$.
\end{itemize}
\end{theorem}

The theorem highlights the importance of probability bounds on
multiple crossings of annuli. Each of the claims have their own applications
below although they are closely related, see \cite{Aizenman:1999dt}.

Call segments of branches of the exploration tree crossing an annulus $A=A(z_0,r,R)$
as simply as \emph{(tree) crossings} and open or dual-open paths crossing $A$ 
in the FK configuration as \emph{arms}.

\begin{lemma}\label{lem: crossings and arms}
It holds that the number of disjoint tree crossings is at most one 
greater than the number of disjoint arms. 
\end{lemma}

\begin{proof}
Use the notation $\underline{\rndbran}(t)$ to denote the order of exploration
of branches in the sense that
$\underline{\rndbran}(t) = \rndbran_k(t-k)$ when $t \in [k,k+1)$ where $\rndbran_k$
is the branch to a target point $x_k$.
For a fixed annulus $A=A(z_0,r,R)$, we suppose that we first explore all $x_k$'s
that are not contained in $A$ and only after that $x_k$'s that are contained in $A$.
The possible crossings are found among the former set of target points and thus we will assume below
that $x_k \notin A$.

Denote by $n(t)$ the number of open and dual-open crossings of $A$ observed by time $t$
excluding the edges which are part of the boundary conditions
and by $k(t)$ the minimal number of crossings made in the tree given the observation by time $t$.
Then $n(0)=0$ and $k(0)=1$.

Let $\tau_1$ be the time when the first crossing occurs.
Then $k(\tau_1) = 1 $ or $2$. If $k(\tau_1)=1$ then, in fact, there are no further crossings
of $A$. On the other hand, if $k(\tau_1)=2$ then at least one of the sides of the crossing
is monochromatic (open or dual-open).

We claim that
on each crossing by the tree that follows, 
the numbers $n$ and $k$ change  so that the change in $n$ is at least as great
as the change in $k$.
Thus $k=m$ crossings by the disjoint segments of the tree implies at least $m-1$ by open
or dual-open paths.

For the claim observe that, if  $\underline{\rndbran}[s,t]$ is the first minimal crossing 
(in the sense that it doesn't contain subcrossings)
after $u \geq \tau_1$ (by the time $u$ there is at least one crossing),
then $\underline{\rndbran}[s,t]$ is contained in a connected component $V$
of $\overline{A} \cap (\Omega \setminus \underline{\rndbran} [0,s))$
and $\underline{\rndbran}[s,t]$
can be disjoint from the left-hand boundary or the right-hand boundary of $V$
or it can intersect both of them.
If it is disjoint from both of them, $k$ increases by $2$,
since both of components of $V \setminus  \underline{\rndbran}[s,t]$ need to be crossed at least
once (and can be explored by crossing only once),
and $n$ increases at least by $2$, since
then both sides of $\underline{\rndbran}[s,t]$ are monochromatic (open or dual-open)
paths. 
Notice also that in this case $\underline{\rndbran}[s,t]$ is a segment of a single loop.
Similarly if one of the sides are disjoint from the boundary, say, from the left-hand boundary,
then $k$ increases by $1$ --- similarly as before --- and 
$n$ increases by at least $1$, since the left-hand side of $\underline{\rndbran}[s,t]$
is a monochromatic (open or dual-open) path. Namely, in that case, any branching point has to be on the right
and if there are any, then all the loops are clockwise loops. Finally if the path touches the boundary on both sides,
then $k$ remains unchanged. Thus the claim and the lemma follows.
\end{proof}

\begin{figure}[tb]
\newcommand*{\myscale}{0.27}
\newcommand*{\varmyscale}{0.06}
\centering
\subfigure[An annular portion of the lattice. The region inside the square is enlarged in (b) and (c).] 
{
	\includegraphics[scale=\varmyscale]
{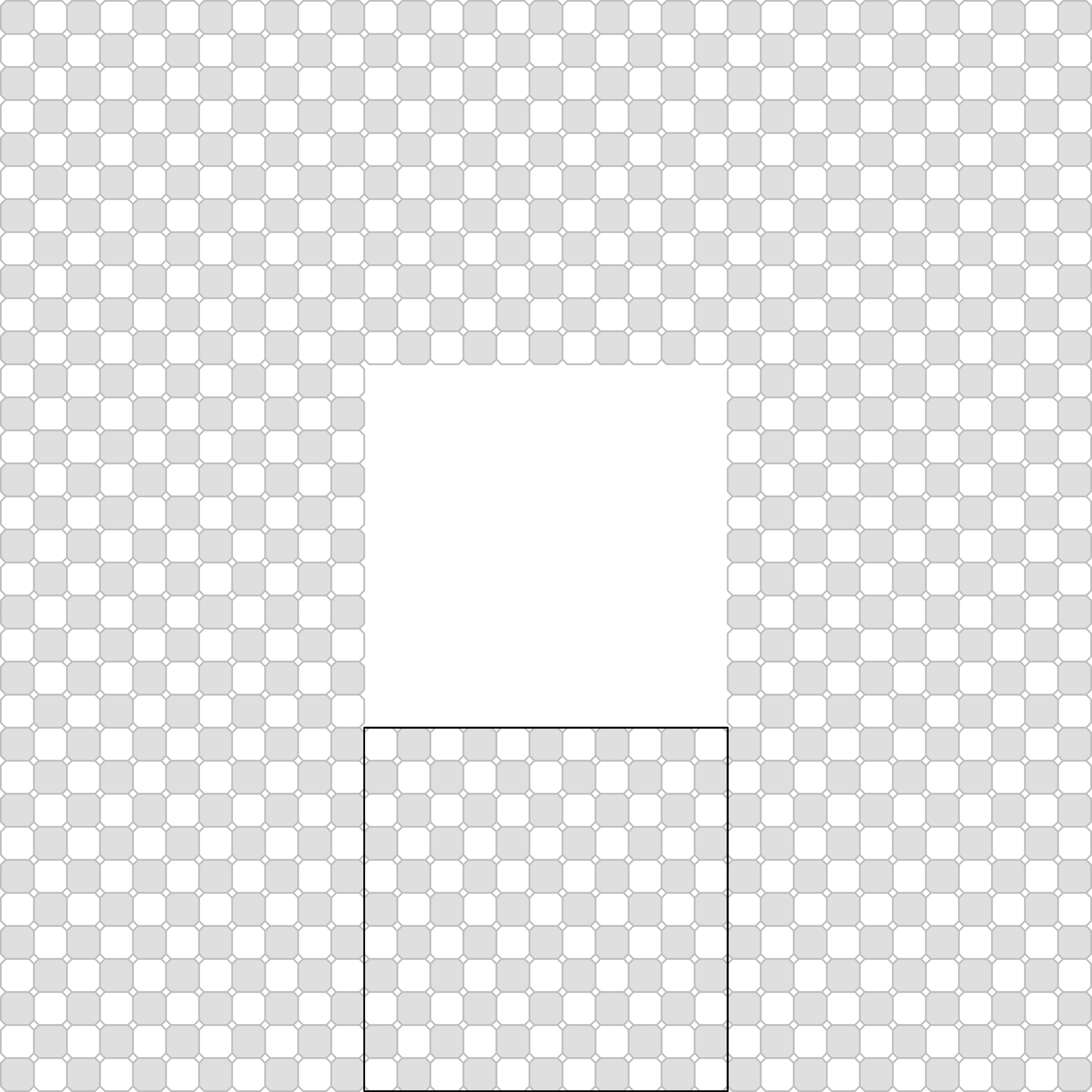}
} 
\\
\subfigure[If the crossing hits the wired arc 
with its right-hand side
(that is, the crossing visits a small square next to the wired boundary), 
then possible branching,
which occurs when there is an edge connecting the left-hand side of the branch to the wired boundary
(such as the dashed blue line in the figure), implies that there is no further crossings on the right
of this crossing.] 
{
	\includegraphics[scale=\myscale]
{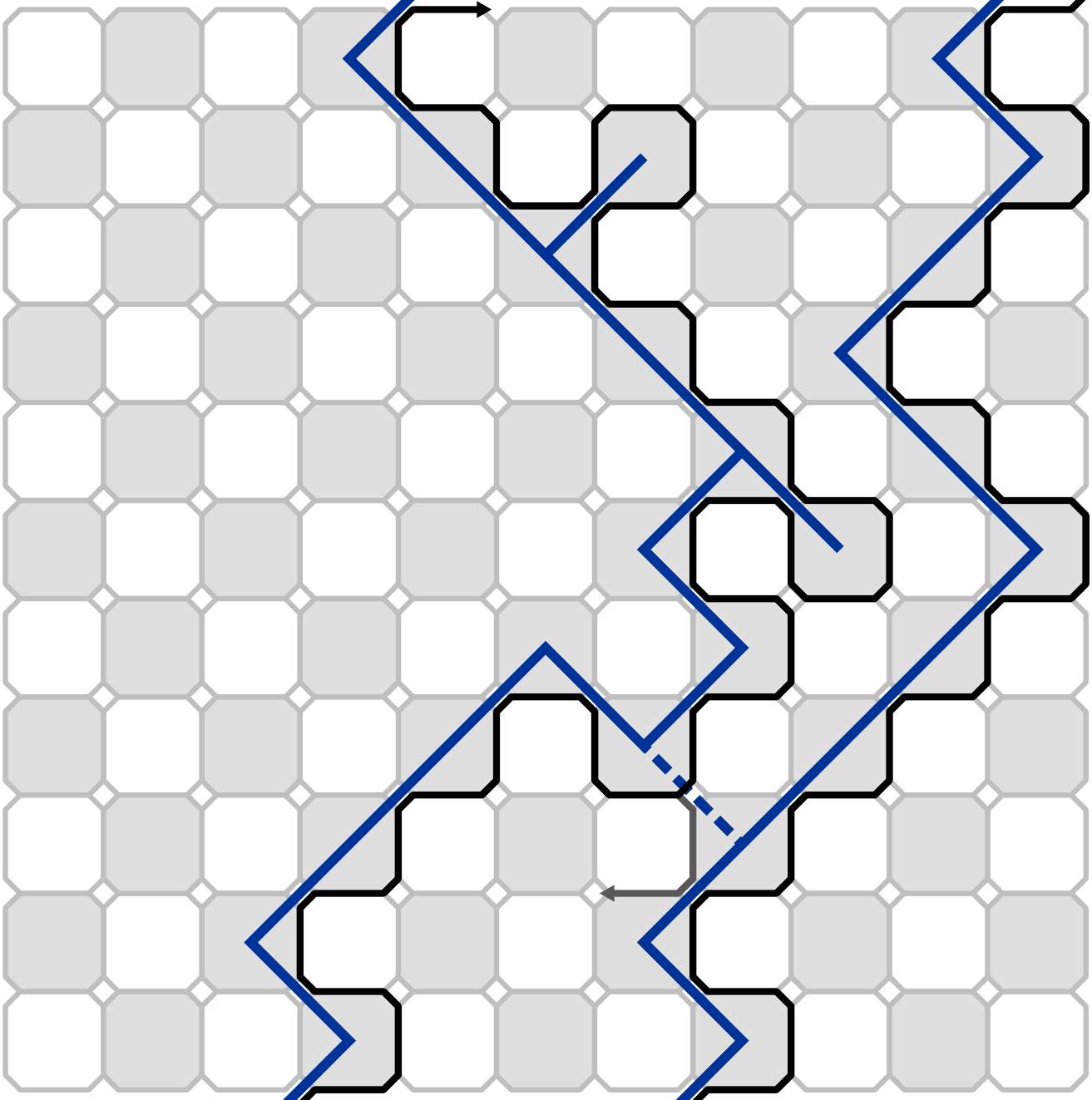}
} 
\hspace{0.5cm}
\subfigure[Similarly if the right-hand side hits the dual wired boundary, then
the tree will branch and there won't be a further crossing to the right of this crossing. ]
{
	\includegraphics[scale=\myscale]
{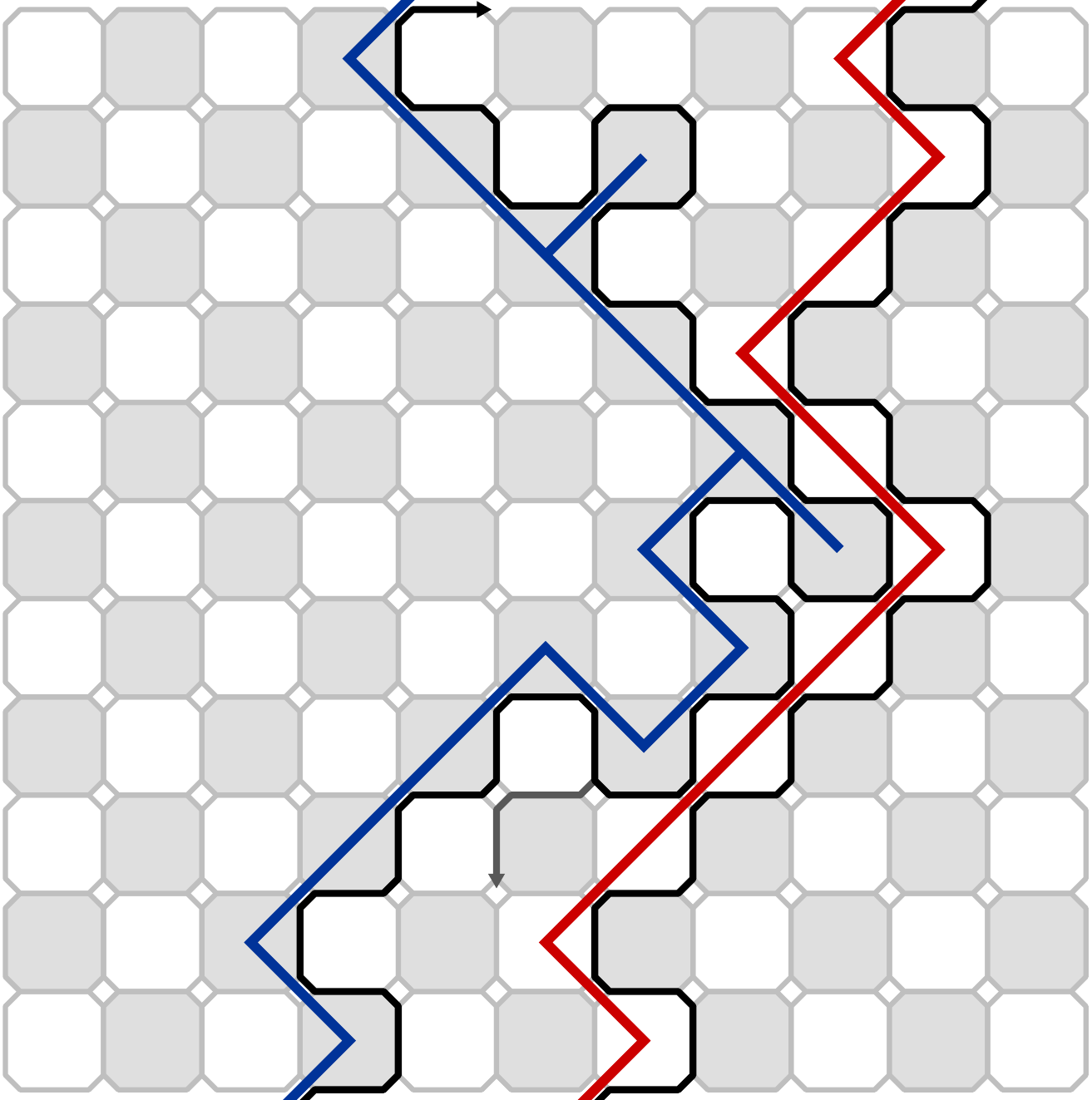}
}
\caption{Illustration of the crossing by the tree and its implied open or dual-open arms and how
touching the boundary prevents a new arm. The crossing occurs here when the black/dark gray reaches
the top of the picture}
\label{fig: f crossing}
\end{figure}

\begin{proof}[Proof of Theorem~\ref{thm: aizenman--burchard bounds}]
We need to verify the first claim
and the two other claims follow from it, by results of \cite{Aizenman:1999dt}.

Let $A=A(z_0,r,R)$ be annulus, $z_0 \in \disc$. Since
either $B(z_0,\sqrt{rR}) \subset \disc$ or $B(z_0,\sqrt{rR}) \cap \partial \disc \neq \emptyset$,
it holds that we can choose $r_1= \sqrt{rR}$, $R_1=R$ or $r_1=r$, $R_1= \sqrt{rR}$ 
such that for $A_1 = A(z_0,r_1,R_1)$, either $A_1 \subset \disc$ or 
$B(z_0,r_1) \cap \partial \disc \neq \emptyset$.
For that $A_1$, apply estimates of conformal distortion
--- such as those in Section~6.3.4 of \cite{Kemppainen:tb} ---  
to show that for some universal constant $\alpha>0$ it holds that
there exists $z_1 \in \C$ and $\tilde r>0$ such that the annulus
$\tilde A = A(z_1,\tilde r, \tilde R)$, where $\tilde R = \tilde r \, (R/r)^\alpha$,
such that the conformal image
of any crossing of $A_1$ under $\varconfmap^{-1}$ is a crossing of $A_2$.

By Lemma~\ref{lem: crossings and arms} and the
results of \cite{Chelkak:2016jw} (in particular, Lemma~5.7)
applied to crossings of $A_2$,
it follows that 
for some constant $K_n$ and $\Delta_n$ it holds 
for all annuli $A=A(z_0,r,R)$ and for all $\P \in \Sigma_\disc$ that
\begin{equation*}
\P(\text{at least $n$ disjoint segments of $\therndttr$ cross $A$} ) \leq K_n \ffrac{r}{R}^{\Delta_n} .
\end{equation*}
Here the constant $\Delta_n$ can be chosen so that they tend to $\infty$ as $n\to \infty$.
\end{proof}

\subsection{An annulus crossing property of a single branch}

A similar bound as used in \cite{Kemppainen:2012vma} will be verified next for a single branch.
It will imply a framework of estimates for the random curve which is sufficient
for the \emph{tightness in the topology of convergence of curves in the capacity parametrization}.

\subsubsection{A bound for unforced crossings}

We will make the next definition
for any annulus $A=A(z_0,r,R)$, fixed $w$ and any random time $\tau$.
We assume that $\tau$ belongs to a family of random times which at least contains
all geometric stopping times of the type ``$\tau$ is equal to the smallest $t$
such that $\therndbran_w(t)$ is in the closed set $B$.''

\begin{definition}
Define $\unf{A}{\tau}$ to be the set of points $z \in A \cap (U \setminus \therndbran_w (0,\tau])$
such that the component of $z$ in
$A \cap (U \setminus \therndbran_w (0,\tau])$ doesn't disconnect $\gamma(\tau)$
from $w$ in $A \cap (U \setminus \therndbran_w (0,\tau])$.
We say that $\unf{A}{\tau}$ is \emph{avoidable} and that any crossing of $A$
which is contained in $\unf{A}{\tau}$ is \emph{unforced}. 
\end{definition}

The next definition  is a version of Condition~G2 of \cite{Kemppainen:2012vma}.

\begin{condition*} 
The random curve $\therndbran_w$
is said to satisfy a \emph{geometric power-law bound on an unforced crossing}
with constants $K>0$ and $\Delta>0$, if
for any random time $0 \leq \tau \leq 1$ and for any annulus $A=A(z_0,r,R)$ where $0 < r \leq R$,
\begin{equation}\label{ie: cond annulus exp}
\P \big( \big.
      \therndbran_w[\tau,1] \textrm{ makes a crossing of } A
      \text{ which is contained in } \overline{A^u_\tau}
   \,\big|\, \therndbran_w[0,\tau] \big) \leq K \left( \frac{r}{R} \right)^\Delta  . 
\end{equation}
\end{condition*}

\begin{remark}
The difference in the current paper and \cite{Kemppainen:2012vma} is that here the target $w$
is an interior point whereas in \cite{Kemppainen:2012vma} the target $b$ was a boundary point.
However the conditions of \cite{Kemppainen:2012vma} are not affected by this (in fact
interior points were used in \cite{Kemppainen:2012vma} to formulate the conditions
in the case of a target on the boundary).
\end{remark}

We will then verify the condition.

\begin{theorem}
Any branch $\therndbran_w$ satisfies \refcond{}
and the constants $\Delta$ and $K$ can be chosen to be uniform for all $U,a,w$.
\end{theorem}

\begin{proof}
We will use the tools from \cite{Kemppainen:2012vma}
together with the RSW estimate for general quadrilaterals shown in \cite{Chelkak:2016jw}.

For a topological quadrilateral (quad) $Q$ denote the event, that 
the ``long sides'' are connected by both open and dual-open crossing of $Q$, by $E(Q)$.
Then for any $\eps$, there is a number $M$ such that if the extremal length of $Q$ (measured
for the curve family connecting the ``short sides'')
is greater than $M$, then the probability of $E(Q)$ is at least $1-\eps$.
Consider $A^u_\tau$ and notice that by the tools of Section~2.2 of \cite{Kemppainen:2012vma},
in particular the inequality~(27), we can find a system of separating arcs (see \cite{Kemppainen:2012vma})
and the related quads $Q_i$ such that
the sum of $1-\P(E(Q_i))$ over $i$ is at most $1/2$. Thus the probability of the complement
of $\bigcap_i E(Q_i)$ is at most $1/2$. Thus for large enough $R/r$, it follows that
the left-hand side of \eqref{ie: cond annulus exp} is at most $1/2$. As explained in \cite{Kemppainen:2012vma},
using this in disjoint concentric annuli implies the claim. 
\end{proof}

\subsubsection{Conformal properties of annulus crossing bounds} 

It turns out to be useful to consider the random curves in $\disc$ or any other
reference domain whose boundary is a smooth simple curve. The reason for
this is mainly to be able to extend any regularity considerations all the way to the endpoints
of the random curves.
Denote by $\varconfmap_w$ the conformal map from $\disc$ to $U$ such that $0$ is mapped
to $w$ and $-1$ to $a$.

What enables us to make the analysis in the domain $\disc$, is the following theorem
on the invariance (or covariance) of \refcond{} type bounds under conformal mappings.
The proof is in Section~2.2 of \cite{Kemppainen:2012vma}.

\begin{theorem}
For any $w \in U$, $\varconfmap_w^{-1}(\therndbran_w)$ satisfies \refcond{}
and the constants $\Delta$ and $K$ can be chosen to be uniform for all $U,a,w$.
\end{theorem}

\subsection{A priori bounds for a branch}

In this subsection, we study precompactness of family of laws of a single branch
The topology is given by the uniform convergence of capacity-parametrized curves
and the main tool is \refcond{}.

\subsubsection{The main lemma for a radial curve}

For similar results as presented in this subsection,
see Proposition~6.4 in \cite{Kemppainen:tb} or 
Appendix~A.2 in \cite{Kemppainen:2012vma}.

Let $\gamma: [0,\infty) \to \C$ be a simple curve such that $|\gamma(0)|=1$,
$0< |\gamma(t)|<1$ for all $t \in (0,\infty)$ and $\lim_{t \to \infty} \gamma(t)=0$. 
Assume also that $\gamma$ is parametrized by the d-capacity, which can always be done
in this case.
Denote the driving process of $\gamma$ by $U$. \emph{To emphasize the driving process},
let's use the notation $\gamma_{U}=\gamma$ as well as $f_U (t,z) = g_t^{-1}(z)$.

Define for $\eps \in (0,1)$
\begin{equation}
F_U (t,\eps) = f_U( t, (1-\eps)U_t) .
\end{equation}
When the Loewner chain is generated by a curve $\gamma_U$, 
$F_U$ extends to a continuous function on $[0,\infty) \times [0,1)$. 
Consequently, $\lim_{\eps \to 0}  \sup_{ t \in [0,T]} |F_U(t,\eps) - \gamma_U(t)| = 0$
for all $T>0$.
To get an uniform property of this type, we look at curves $\gamma_U$ that satisfy
\begin{equation}\label{ie: speed to tip}
\sup_{ t \in [0,T]} |F_U(t,\eps) - \gamma_U(t)| \leq \lambda(\eps) 
\end{equation}
where $\lambda:(0,1) \to (0,\infty)$ is a function such that 
$\lim_{\eps \to 0} \lambda (\eps) = 0$. It is natural to define for any such function $\lambda$
 and any $T>0$,
\begin{equation*}
\mathcal{E}_{\lambda,T} = \left\{ U \,:\, 
  \exists \gamma_U \text{ as above and \eqref{ie: speed to tip} 
  holds } \forall \eps \in (0,1) \right\} .
\end{equation*}
Notice that $\eps \mapsto F_U(t,\eps)$ is the so called \emph{hyperbolic geodesic} between 
$0$ and $\gamma_U(t)$ in the domain $\disc \setminus \gamma_U(0,t]$.

In next subsections, we apply the following lemma to a branch in the FK Ising exploration tree.

\begin{lemma}\label{lm: the main lemma}
Let $\lambda:(0,1) \to (0,\infty)$ be a function such that 
$\lim_{\eps \to 0} \lambda (\eps) = 0$. The map from $U \in \mathcal{E}_{\lambda,T}$
to $\gamma_U$ is uniformly continuous.
More specifically, 
for each $T>0$ and $\eps \in (0,1)$, there exists a constant $C=C(T,\eps)$
such that
if $\lambda$ is as above and
$U_k \in \mathcal{E}_{\lambda,T}$ for $k=1,2$, then 
\begin{equation}\label{ie: w to gamma continuity}
\Vert \gamma_{U_1} - \gamma_{U_2} \Vert_{\infty,[0,T]}
 \leq C(T,\eps) \Vert U_1 - U_2 \Vert_{\infty,[0,T]} + 2 \lambda (\eps) .
\end{equation}
\end{lemma}

\begin{proof}
Fix $T>0$. It is fairly straightforward to show that 
there is a constant $C(T,\eps)$ such that
for any $z_k \in \C$ such that $|z_k| \leq 1 -\eps$
and for any continuous $U_k : [0,\infty) \to \partial \disc$ it holds that
\begin{equation}\label{ie: loewner lip}
\sup_{t \in [0,T]} |f_{U_1}(t,z_1) - f_{U_2}(t,z_2)|
  \leq \frac{C(T,\eps)}{2} ( \Vert U_1 - U_2 \Vert_{\infty,[0,T]} + |z_1 - z_2|) .
\end{equation}
For instance, this can be derived using the same route as
\PropositionLoewnerLip{}. Namely, first use a version of
\LemmaReverseLoewner{} to relate $f_k(t,z)$ to a time-reversal of the (direct)
Loewner flow $g(t,z)$ for a specific driving term. Then use an argument
similar to \LemmaReverseLoewnerLip{} to estimate the difference of
the solutions of the
time-reversed Loewner equation for the two driving terms and for two initial values.
We leave the details to the reader.

Let $z_k = (1-\eps)U_1(t)$, $k=1,2$.
Then $|z_1 - z_2| \leq \Vert U_1 - U_2 \Vert_{\infty,[0,T]}$ and it follows from
\eqref{ie: loewner lip} that
\begin{equation}\label{ie: loewner tip approach lip}
|F_{U_1}(t,\eps) - F_{U_2}(t,\eps) |
  \leq C(T,\eps) \Vert U_1 - U_2 \Vert_{\infty,[0,T]} .
\end{equation}
for all $t \in [0,T]$.

Next use the assumption that $U_1,U_2 \in \mathcal{E}_{\lambda,T}$. Then
by the triangle inequality, the inequality~\eqref{ie: w to gamma continuity}
follows.

Let's finalize the proof by showing
that the uniform continuity of the map $U \mapsto \gamma_U$.
For any $\tilde \eps>0$, choose $\eps>0$ such that $\lambda(\eps) \leq \tilde \eps/3$.
Then choose $\delta>0$ such that $C(T,\eps) \delta  \leq \tilde \eps/3$.
It follows from \eqref{ie: w to gamma continuity} 
that for any $U_1,U_2 \in \mathcal{E}_{\lambda,T}$
such that $\Vert U_1 - U_2 \Vert_{\infty,[0,T]} < \delta$,
it holds that $\Vert \gamma_{U_1} - \gamma_{U_2} \Vert_{\infty,[0,T]} 
  < C(T,\eps)\delta + 2 \lambda(\eps)< \tilde \eps$.
\end{proof}

\subsubsection{Uniform probability bound on the modulus of continuity of the driving process}

As demonstrated in \cite{Kemppainen:2012vma}, a probability bound on the modulus of continuity
of the Loewner driving term
follows from the probability bound on annulus crossing
(\refcond) by considering
crossings of thin conformal rectangles along the boundary of the domain.
The argument in \cite{Kemppainen:2012vma} is written for the Loewner equation in
the upper half-plane, but the argument adapts directly to the Loewner equation
in $\disc$. For instance, from the Loewner equation in $\disc$,
we can deduce that there is a constant $c>0$ such that $|\gamma(t)| \geq 1 - c t$
and consequently, $\max_{s \in [0,t]}  |\arg U_s- \arg U_0| \geq L>0$ where $L$ and $t$ are small and $L$
is much greater than $t$,
only if $\gamma$ exits $\{z \,:\, |z|>1-ct, \arg z - \arg U_0 \geq L/2\}$
from the ``sides''. Consequently the following theorem holds.
For the original result, see \cite{Kemppainen:2012vma}, Section~3.3.

\begin{theorem}\label{thm: tight driving process}
Let $\beta \in (0,\frac{1}{2})$.
For any branch of $\rndttr$ with the target point in a compact subset of $\disc$ 
and with the d-capacity parametrization, the driving process is 
$\beta$-H\"older continuous and the H\"older norm 
is a tight random variable for the family $\Sigma_\disc$.
\end{theorem}

\subsubsection{Uniform probability bound on the modulus of continuity of the hyperbolic geodesic}

Similarly as in the previous subsection, we can adapt the theory in \cite{Kemppainen:2012vma}
to the case of the Loewner equation of $\disc$ to deduce the following result.

\begin{theorem}\label{thm: tight speed to tip}
For any $T>0$, there exists a function 
$\lambda:(0,1) \to (0,\infty)$ such that 
$\lim_{\eps \to 0} \lambda (\eps) = 0$ and the following holds.
Any branch of $\rndttr$ with the target point in a compact subset of $\disc$ 
\eqref{ie: speed to tip} holds for $T>0$ and $0<\eps < \eps_0$
where $\eps_0$
is a tight random variable for the family $\Sigma_\disc$.
\end{theorem}

Let's stress here that this result is proven using the
general $n$-arms bounds 
--- and the implied bound for the tortuosity 
(Theorem~\ref{thm: aizenman--burchard bounds} and, in particular, its second assertion) ---
and the more specific $6$-arms bound. See \cite{Kemppainen:2012vma}, Sections~3.2
and 3.4.

The next result is the main theorem among the a priori bounds for a branch.

\begin{theorem}[%
   Tightness of a single branch in the uniform convergence in the d-capacity parametrization]\label{thm: tight branch}
For each $\eps>0$ and each compact set $D_1 \subset \disc$, 
there exists a compact subset $K$ of the space $C([0,\infty))$
such that the following holds. 
Any branch of $\rndttr$ with the target point in $D_1$ 
and parametrizated with the d-capacity (seen from the target) belongs to $K$
with probability at least $1-\eps$ uniformly
for the family $\Sigma_\disc$.
\end{theorem}

\begin{proof}
Let $D_1\subset \disc$ be compact and $\eta>0$ such that $\dist(D_1,\partial \disc) \geq 2 \eta$.
For any $x \in D_1$, let $\varconfmap_x$ be the conformal and onto selfmap of the unit disc
such that $\varconfmap_x(x)=0$.
Then $\varconfmap_x$ and $\varconfmap_x^{-1}$ 
are Lipschitz continuous on $\overline{\disc}$ with a uniform Lipschitz norm over all $x$
by a direct caluculation.
Thus we can infact assume that $x=0$ and consider only the branches from the root
(which can be set to be $-1$) to the target $x=0$.

For each $n$, let $K_n$ be the set of simple curves $\gamma$ going from the root to the target
and parametrized with the d-capacity such that for some $\alpha_0>0$, $\beta \in (0,\frac{1}{2})$,
$C_n$, $\lambda_n:(0,1) \to (0,\infty)$ such that 
$\lim_{u \to 0} \lambda_n (u) = 0$ and
$v_n$, the following holds
\begin{itemize}
\item $\gamma([n,\infty)) \subset B(0,e^{-\alpha_0 n})$
\item $\gamma(t)$, $t \in [0,n]$, satisfies that its driving term is
$\beta$-H\"older continuous and the H\"older norm 
is at most $C_n$
\item the bound \eqref{ie: speed to tip}, where $\lambda$ is replaced by $\lambda_n$
and $\eps$ is replaced by $u$, 
holds for $\lambda$, $t \in [0,n]$ and $u \in (0,v_n)$
\end{itemize}
and it holds that $\P(K_n) \geq 1-2^{-n}\eps$.
Such $\alpha_0>0$, $\beta$,
$C_n$, $\lambda_n$ and
$v_n$ exist for $n \geq n_0$ for some $n_0 \in \N$  by 
using \refcond, 
Theorem~\ref{thm: tight driving process} and 
Theorem~\ref{thm: tight speed to tip} for the three claims, respectively,
and them bounding the probability of the intersection of three events
from below by $1$ minus the sum of the probabilities of their complementary events.

Let $K_\infty = \bigcap_{n \geq n_0}  K_n$. Then  $\P(K_\infty) \geq 1-\eps$.
Let $K = \overline{K_\infty}$ where the closure is on the uniform
convergence of the d-capacity-parametrized curves on compact subintervals of $[0,\infty)$.
Using Lemma~\ref{lm: the main lemma} and the Arzel\`a–Ascoli theorem,
it is straightforward to check that
$K$ is sequentially compact and thus compact.
\end{proof}

\subsection{A priori bounds for trees}\label{ssec: a priori trees}

It is straightforward to extend the convergence 
to many, but fixed number of, branches
by using the result of a single branch $n$ times and then using the union bound of probabilities. 
Consequently, we find that the probability that $n$ branches whose target points are in a compact subset of $\disc$,
each belong to a compact subset $K$ of $C([0,\infty))$ has probability greater than $1-\eps$.
It remains to be shown 
that a tree with fixed, large number
of branches is close to the full tree with uniformly high probability.

Let $\eta>0$ and $\tilde I_\eta = (\eta \Z^2) \cap \disc$. For each $x \in \tilde I_\eta$,
choose a point $z_x \in \domain_\delta \cap \delta V_\textnormal{mid}$ such that
$\varconfmap(z_x) \in B(x,\eta)$. This can be done using the following lemma.

\begin{lemma}\label{lm: aux ssec a priori trees}
For each $\eta,r,R$, there exists $\delta_0>0$ such that the following holds.
The set $\varconfmap(\domain_\delta \cap \delta V_\textnormal{mid}) \cap B(x,\eta)$
is non-empty when $\delta \in (0,\delta_0)$ and $B(0,r) \subset \Omega_\delta \subset B(0,R)$.
\end{lemma}

For the finite tree approximation of the full tree,
we use the following scheme to select the set of target points:
\begin{itemize}
\item Include all points $z_x$ where $x$
runs over  the elements of $\tilde I_\eta$.
\item If some components with diameter greater than $\eta$ still exists
in the complement of the tree in $\disc$ (notice the conformal image of the tree), 
include a target point in all those components.
More specifically, cut that component into four equal quarters 
in the direction of its diameter (for instance, if the diameter is $[\xi,\xi+\zeta]$, then cut along $\xi+\frac{k}{4}\zeta+\ii \zeta \, y$, $y \in \R$, for $k=1,2,3$). Then select the target point in one of 
the non-neighboring quarters to the quarter 
of the branching point to that component.
Repeat this second step until the maximal diameter of the domains to be explored 
is less than $\eta$.
\end{itemize}
Denote the chosen set of target points as $I_{\delta,\eta} \subset \domain_\delta$.
Let $\tree_\delta$ be the (full) exploration tree on $\domain_\delta$
and let $\tree_\delta^\disc = \varconfmap^{-1} (\tree_\delta)$.
Define $\widetilde \tree_{\delta}(I_{\delta,\eta})$
and $\widetilde \tree_{\delta}^\disc(I_{\delta,\eta})$
be the restrictions of $\tree_\delta$ and
$\tree_\delta^\disc$ to the branches whose target point is in
$I_{\delta,\eta}$. We call $\widetilde \tree_{\delta}^\disc(I_{\delta,\eta})$
the \emph{$\eta$-approximation} of $\tree_\delta^\disc$.
It follows directly that
\begin{equation*}
\dettr 
  \left(\widetilde \tree_{\delta}^\disc(I_{\delta,\eta})
  \,,\, \tree^\disc_\delta\right)\, < \eta .
\end{equation*}
The following result shows that the number of target points needed for the 
$\eta$-approximation is a tight random variable.

\begin{theorem}\label{thm: tree approximation number of branches}
For each $\eta>0$ and $\eps>0$, there exists a constant $M>0$ such that
\begin{equation*}
\sup_{ \P_\delta} \P_\delta \left[ \#(I_{\delta,\eta}) > M \right] < \eps .
\end{equation*}
\end{theorem}

Notice that the way that the sequence of points was constructed implies that
all the conclusions of we made in this section for fixed target points
hold also for these random target points.

\begin{proof}[Proof of Lemma~\ref{lm: aux ssec a priori trees}] 
Suppose first that $|x| < 1 - \eta/4$. Then by Koebe distortion theorem, there exists 
a uniform constant $\eps>0$ such that $B(\varconfmap(x),\eps) \subset \varconfmap(B(x,\eta/8))$. Thus
the claim holds for those $x$'s when $\delta_0 < \eps/2$.

Suppose then that $|x| \geq 1 - \eta/4$. 
Let $J = B(x,\eta) \cap \partial \disc$. If the length $\varconfmap^{-1}(J)$ is greater
than $\delta$, then $\varconfmap^{-1}(J)$ contains at least one lattice point (which is on the boundary)
and the claim follows. If $\varconfmap(J)$ shorter than $\delta$ and it doesn't contain any lattice
points, then the endpoints are on the same edge of the lattice. Take one of its endpoints.
It follows easily that the diameter of the image of the line segment connecting that endpoint
to the closest end of the edge under the map $\varconfmap$ is at most
$2\pi /\sqrt{\log \delta^{-1}}$ using
an estimate on the length distortion of conformal maps such as
Proposition~2.2 in \cite{Pommerenke:1992fw}.
\end{proof}

\begin{proof}[Proof of Theorem~\ref{thm: tree approximation number of branches}]
Use an argument similar to the proof of Theorem~5 of \cite{Camia:2006wv}.
Namely notice that each time we select a new target point in the above scheme
we create a segment of the exploration tree with diameter greater than $\eta/4$.
Furthermore, they are all disjoint. Since there is $n$ such that
there is no $n$-arms event of the tree 
between scales $\tilde \eps$ and $\eta/16$
by the results of Section~\ref{ssec: n arms tree}
where $\tilde \eps^{-1}$ is a tight random variable,
it follows that we need only at most $n \tilde \eps^{-2}$ points in the above construction.
This quantity is tight and thus with uniformly high probability, it is less than a given large number.
\end{proof}

\begin{theorem}
For each $\eps>0$, there exists a compact set $K$ in the space of closed subsets of $C(\Rtime)$
such that $\P_\delta( \tree^\disc_\delta \in K )\geq 1-\eps$.
\end{theorem}

\begin{proof}
By the above results, we can choose for each $k \in \N$, $M$ such that
$\P_\delta [ \#(I_{\delta,k^{-1}}) > M ] < \eps 2^{-2-k}$
by Theorem~\ref{thm: tree approximation number of branches}.
Denote the ``projection'' map of the full tree to its {$\frac{1}{k}$-approximation} by $\pi_k$.
Choose then a compact subset $K_k$ of the space of closed subsets of $C(\Rtime)$
such that $\P_\delta( \pi_k( \tree^\disc_\delta ) \in K_k ) \geq 1 -  2^{-1-k}\eps$
which exists by the previous bound and $M$-fold application of
Theorem~\ref{thm: tight branch}. Define $K= \bigcup_{k=1}^\infty \pi_k^{-1}(K_k)$.
Notice that $\P_\delta( \tree^\disc_\delta \in K ) \geq 1 - \eps$.

We show that $K$ is precompact by establishing that it is totally bounded.
Let $\hat \eps>0$ and $k$ such that $\frac{1}{k} < \frac{\hat\eps}{2}$.
Since $K_k$ is compact and thus totally bounded, we can find a finite collection of balls
$B(t_j,\frac{\tilde\eps}{2})$, $j=1,2,\ldots,m$, that cover $K_k$. Then by the fact
that $\dettr(\pi_k(t),t)< \frac{1}{k}< \frac{\tilde\eps}{2}$ for any tree $t$, it holds that 
the collection $B(t_j,\tilde\eps)$, $j=1,2,\ldots,m$, covers 
$K_k$. Thus $K_k$ is totally bounded and consequently precompact, since the 
the space of closed subsets of $C(\Rtime)$ is complete.
\end{proof}

\subsection{A priori bounds for loop ensembles}

Any loop in $\rndlpe$ can be constructed from the tree by the reverse
construction presented in Section~\ref{ssec: main result a priori}.
By construction, any loop is a subpath of the corresponding
branch of the exploration tree. Thus Theorem~\ref{thm: aizenman--burchard bounds}
implies the following result.

\begin{theorem}\label{thm: precompactness loop collections}
The family of probability laws of $\rndlpe$ is tight in the metric space of loop collections.
More specifically, there are a constant $\alpha>0$ and a tight random variable $C_\alpha$
over the family of probability laws
such that
all loops in $\rndlpe$ can be jointly parametrized such that
they all are  $\alpha$-H\"older continuous and the H\"older norm
is bounded from above by $C_\alpha$.
\end{theorem}

\section{Determining the law of a branch from the observable}
\label{sec: simple martingales}

\subsection{Simple martingales from \texorpdfstring{$F$}{F}}

Recall that the value of $F_\delta(z)$ is a discrete-time martingale as a process in time variable $n$
when $\Omega_\delta$ is replaced by $\Omega_\delta \setminus \gamma(0,n]$ where
$\gamma$ is the branch of $w$ with the lattice step parametrization.
By the uniform convergence of Corollary~\ref{cor: uniform conv obs},
it follows that $F(z)$ is a 
a continuous-time martingale as a process in time variable $t$
when $\Omega$ is replaced by $\Omega \setminus \gamma(0,t]$ where
$\gamma$ is the (subsequent) scaling limit of the branch to $w$.
For the proof of this type of statement, see the proof of Theorem~1 in \cite{Chelkak:2014gs}
and Appendix~\ref{sec: appendix martingale} below.

To benefit from the martingale property, we search for simple expressions that
we can extract from $F$ which are martingales. See also Proposition~5.1 in \cite{Kemppainen:2015vu}.

\subsubsection{Expansion of the observable of the branch}

Set $\coeffa(t)$ and $\coeffb(t)$ to be the coefficients of the observable defined as 
in the Proposition~\ref{prop: coefficients a and b}
when $\argu=\arguu_t$ and $\argv=\argvv_t$.

Let's use the expansion of the Loewner map of $\disc$
\begin{align*}
g_t(z)  &= e^t (z + c(t) \, z^2 + \ldots) \\
g_t'(z) &= e^t (1 + 2 c(t) \, z + \ldots) .
\end{align*}
Notice that there is a great simplification in the expression
\begin{equation*}
\frac{g_t'(0) g_t'(z)}{g_t(z)^2} = \frac{1}{z^2} \; \frac{1 + 2 c(t) z + \ldots}{1+ 2 c(t) z + \ldots} = \frac{1}{z^2} ( 1 + \OO(z^2) )
\end{equation*}
since it doesn't contain any $z^{-1}$ term.

The expansion of the observable around the origin is
\begin{align*}
 & \sqrt{ g_t'(0) \, g_t'(z) } \, F_\disc \left( g_t(z); e^{i\, \arguu_t}, e^{i \, \argvv_t} \right) 
 = \frac{1}{z} ( 1 + \OO(z^2) ) \cdot \nonumber \\
 & \quad \cdot \sqrt{ 1 + i \, \coeffa(t) \, g_t(z) - g_t(z)^2  
      - \coeffb(t) g_t(z)^2 \, \left( \frac{1}{g_t(z)-e^{i \, \arguu_t}} - \frac{1}{g_t(z)-e^{i \, \argvv_t}} \right) } \nonumber \\
 =& \frac{1}{z} ( 1 + e^t \coeffa(t) z +\OO(z^2)  ) .
\end{align*}
The first non-trivial coefficient is 
\begin{equation*}
M(t) \dd= e^t \cos\left( \frac{\argvv_t-\arguu_t}{2} \right) .
\end{equation*}
By the martingale property of the observable, the process $(M(t))_{t \geq 0}$ is a martingale.

\subsubsection{Value of the ``chordal'' observable at \texorpdfstring{$w$}{w}}

The leading coefficient of
\begin{align*}
 & \real \left( \sqrt{ g_t'(z) } \, \tilde F_\disc \left( g_t(z); e^{i\, \arguu_t}, e^{i \, \argvv_t} \right) \right)
 = e^{\frac{t}{2}} \sqrt{\coeffb(t)} + \OO(|z|)
\end{align*}
is 
\begin{equation*}
N(t) = \pm e^{\frac{t}{2}} \sqrt{ \sin \left( \frac{\argvv_t-\arguu_t}{2} \right) } \, \cos \left( \frac{\arguu_t+\argvv_t + \pi}{4} \right) .
\end{equation*}
Here $\pm$ are symmetrically distributed random signs independently sampled on each excursion of $\argvv_t-\arguu_t \geq 0$.
The signs are symmetrically distributed, since the macroscopic excursion can be shown to be pairwise disjoint
in time and thus the total random sign changes between excursions by a factor that 
is product of an infinite number of non-symmetric signs that originate from the discrete martingale 
as in \eqref{eq: def discrete N}, see also \cite{Kemppainen:2015vu} for similar discussion.
By the martingale property of the observable, 
the process $(N(t))_{t \geq 0}$
is a martingale.

\subsection{Solution of the martingale problem}

Next we will show that the fact that $M(t)$ and $N(t)$ are martingales 
implies that the law of the scaling limit of a single branch is \thekr.

Remember that for SLE$(\kappa,\rho)$
\begin{align}
\de \arguu_t &= \sqrt{\kappa} \, \de B_t - \frac{\rho}{2} \cot\left( \frac{\argvv_t-\arguu_t}{2} \right) \de t 
  \label{eq: slekr alpha} \\
\dot{\argvv}_t &= \cot\left( \frac{\argvv_t-\arguu_t}{2} \right) 
  \label{eq: slekr beta} .
\end{align}
The process $U_t = e^{i \, \arguu_t}$ is the driving process of Loewner equation of $\disc$
and $V_t = e^{i \, \argvv_t}$ is the other marked point.
Notice that then
\begin{align*}
\de (\argvv_t-\arguu_t) &= -\sqrt{\kappa} \, \de B_t + \frac{\rho+2}{2} \cot\left( \frac{\argvv_t-\arguu_t}{2} \right) \de t 
\end{align*}
which we call the stochastic differential equation of a (unnormalized)
radial Bessel process with parameters $\kappa$ and $\rho$. 
Notice that a linear time change could be used to eliminate of one of the parameters, which would then pop out 
as a parameter in the Loewner equation.

\subsubsection{Solution of the martingale problem}

Start from the processes 
\begin{align*}
M(t) &= e^t \cos\left( \frac{\argvv_t-\arguu_t}{2} \right) \\
N(t) &= \pm e^{\frac{t}{2}} \sqrt{ \sin \left( \frac{\argvv_t-\arguu_t}{2} \right) } \, \cos \left( \frac{\arguu_t+\argvv_t + \pi}{4} \right) .
\end{align*}
which are martingales as shown above, where $\pm$-signs are i.i.d. symmetric coin flips for each excursion
of $\argvv_t-\arguu_t$ away from $0$ or $2\pi$. 
Since the processes are continuous martingales, we can do stochastic
analysis with them, see for instance \cite{Durrett:1996wh} for background. For example, It\^o's lemma holds for these
processes.

Define auxiliary processes
\begin{align*}
X_t &= \frac{\argvv_t-\arguu_t}{2}  \\
Z_t &= e^{-t} M(t) = \cos X_t .
\end{align*}
It holds that
\begin{equation}\label{eq: itodif Z}
\de Z_t = - Z_t \de t + e^{-t} \de M_t 
\end{equation}

We can write
\begin{equation*}
N_t = \pm F(Z_t, \argvv_t) e^{\frac{t}{2}}
\end{equation*}
where
\begin{equation}\label{eq: n f explicit}
F(z,\argv)= \frac{1}{2} (1 - z^2)^{\frac{1}{4}} \left(  
  \sqrt{1 + z} \left(\cos\frac{\argv}{2} - \sin\frac{\argv}{2}\right) + 
  \sqrt{1 - z} \left(\cos\frac{\argv}{2} + \sin\frac{\argv}{2}\right) \right)
\end{equation}

Write the Loewner equation~\eqref{eq: slekr beta} as
\begin{equation*}
\dot \argvv_t = \frac{Z_t}{\sqrt{1-Z_t^2}}
\end{equation*}
By It\^o's lemma, when $Z_t \neq \pm 1$, 
\begin{align}
\de N_t = \bigg[& \frac{1}{2} F_{zz} (Z_t, \argvv_t) e^{-2t} \de \langle M \rangle_t  
    - F_{z} (Z_t, \argvv_t) Z_t \de t \nonumber \\
    &+ F_{\phi} (Z_t, \argvv_t) \frac{Z_t}{\sqrt{1-Z_t^2}} \de t  
    + \frac{1}{2} F (Z_t, \argvv_t) \de t \bigg] e^{\frac{t}{2}} + F_{z} (Z_t, \argvv_t) e^{-t} \de M_t  
    \label{eq: n ito}
\end{align}
where $F_z,F_{zz},F_\phi$ are partial derivatives of $F$ 
(first and second $z$-derivative and first $\phi$-derivative, respectively).
Since $N_t$ is a martingale, the quantity inside the brackets vanishes identically.
We will prove the following result below.
\begin{lemma}\label{lm: instantaneous reflection}
$\P[ \int_0^\infty (\ind_{\arguu_t=\argvv_t} + \ind_{\arguu_t=\argvv_t-2\pi}) \de t = 0] =  1$
\end{lemma}
By this lemma and Lemma~5.3 of \cite{Kemppainen:2015vu},
it follows that  $\de \langle M \rangle_t$ is absolutely continuous with respect to $\de t$ and that if
we write $\de \langle M \rangle_t = a_t^2 \, e^{2 t} \, \de t$,
then there exists a Brownian motion $(B_t)_{t \in [0,\infty)}$
s.t. $\de M_t = a_t \, e^{t} \, \de B_t$.
By \eqref{eq: n ito}
\begin{align}
a_t^2 &= \frac{ 2 F_{z} (Z_t, \argvv_t) Z_t -  2 F_{\phi} (Z_t, \argvv_t) \frac{Z_t}{\sqrt{1-Z_t^2}} - F (Z_t, \argvv_t) }{
    F_{zz} (Z_t, \argvv_t) } \nonumber \\
  &= \frac{4}{3} (1- Z_t^2) .
\end{align}
Notice the extremely simple expression on the right-hand side, which depends only
on the difference of $\argvv_t$ and $\arguu_t$ and not on their sum. This ``miraculous'' simplification
ultimately implies that the exploration tree branches have relatively simple law.

Notice next that by \eqref{eq: itodif Z} and It\^o's lemma
\begin{align}
\de Z_t &= - Z_t \de t + \frac{2}{\sqrt{3}} \sqrt{1- Z_t^2} \de B_t \\
\Longrightarrow \; \de X_t &= \de \arccos Z_t 
 =   \frac{2}{3} \frac{Z_t}{\sqrt{1 - Z_t^2}} \de t  - \frac{4}{\sqrt{3}} \de B_t \nonumber \\
 &= \frac{2}{3} \cot(X_t) \de t  - \frac{4}{\sqrt{3}} \de B_t
\end{align}
That is, $X_t$ follows the law of radial Bessel process of $\kappa=16/3$ and $\rho=\kappa - 6=-2/3$.

By comparing to the usual Bessel process, it follows that 
$\int_0^t \cot\left( \frac{\argvv_s-\arguu_s}{2} \right) \de s$ is finite and continuous.
Thus it follows that
\begin{align*}
\argvv_t &= \argvv_0 + \int_0^t \cot\left( \frac{\argvv_s-\arguu_s}{2} \right) \de s
  + \Lambda_t^+  - \Lambda_t^-
\end{align*}
where $\Lambda_t^+$ and $\Lambda_t^-$ 
are non-decreasing in $t$ and are constant on each excursion of $\argvv_t-\arguu_t$ away from $0$ or $2\pi$
such that $\Lambda_t^+$ and $\Lambda_t^-$ can increase only when 
$\argvv_t-\arguu_t$ hits $0$ or $2\pi$, respectively, 
see Proposition~\ref{prop: rightmost point and increasing process} in Appendix~\ref{sec: appendix loewner}.
The argument similar to the one in
Section~5.5.3 in \cite{Kemppainen:2015vu}, 
which is based on the regularity result Theorem~\ref{thm: tight driving process} above in the present case,
shows that $\Lambda_t^+$ and $\Lambda_t^-$ 
are identically zero.

\subsubsection{Instantaneous reflection at \texorpdfstring{$\argvv_t = \arguu_t$}{Phi_t=Upsilon_t}}
It remains to prove Lemma~\ref{lm: instantaneous reflection}.

\begin{lemma}\label{lm: beta in a}
Let $A \subset \R$ be a countable set. Then $\int_0^\infty \ind_{\argvv_t \in A} \de t = 0$.
\end{lemma}

\begin{proof}
The claim follows from the fact that $t \mapsto \argvv_t$ is strictly increasing.
Thus $\{t \in [0,\infty) \,:\, \argvv_t \in A \}$ is a countable set and has zero Lebesgue measure.
\end{proof}

\begin{proof}[Proof of Lemma~\ref{lm: instantaneous reflection}]
We will show that
$\P[ \int_0^\infty \ind_{\arguu_t=\argvv_t}  \de t = 0] =  1$.
The proof of $\P[ \int_0^\infty \ind_{\arguu_t=\argvv_t-2\pi} \de t = 0] =  1$
is completely symmetric.

Let $c(\argv)=\cos\frac{\argv}{2} - \sin\frac{\argv}{2}$ and
write $A = \{ \argv \in \R \,:\, c(\argv) = 0\}$.
Then by Lemma~\ref{lm: beta in a}, the Lebesgue measure
$\int_0^\infty \ind_{\argvv_t \in A} \de t = 0$.
Therefore it is sufficient to show that
almost surely
$\int_0^\infty \ind_{\arguu_t=\argvv_t}\ind_{\argvv_t \notin A} \de t = 0$.

Let $f_\argv(z) = F(z,\argv)$ where $F$ is as in \eqref{eq: n f explicit}.
Then for any $\argv \notin A$, 
$f_\argv^{-1} (n) = 1 - 2( \frac{n}{c(\argv)} )^4 + \OO( n^6 )$ near $n=0$ and
$(n,\argv) \mapsto f_\argv^{-1}(n)$ is twice differentiable function of $\argv$ and $n$ 
with similar bounds on the derivatives.

Now when $\argvv_t \notin A$ and $Z_t$ is near $1$, it holds that
\begin{equation*}
M_t = f_{\argvv_t}^{-1} (N_t e^{-\frac{t}{2}}) e^t =  e^t - \frac{N_t^4 e^{-t}}{c(\argvv_t)^4} + \OO( N_t^{6} ) .
\end{equation*}
Write
\begin{align*}
&\ind_{|N_t| \leq \eps } \, \ind_{|c(\argvv_t)|> \delta } \, \de M_t \\
  = &\ind_{|N_t| \leq \eps }\,\ind_{|c(\argvv_t)|> \delta } \,  \left( (e^t + \OO( \eps^4 )) \de t + \OO(\eps^3) \de N_t
+ \OO(\eps^2) \de \langle N \rangle_t + \OO(\eps^4) \de \argvv_t  \right) .
\end{align*}
By the fact that stochastic integrals of bounded integrands with respect to a martingale are martingales,
it follows that
\begin{equation*}
\E\left[  \int_0^T \ind_{|N_t| \leq \eps } \, \ind_{|c(\argvv_t)|> \delta } \, e^t \de t \right] 
= \E\left[ \int_0^T \ind_{|N_t| \leq \eps } \, \ind_{|c(\argvv_t)|> \delta } \, \de M_t \right] + \OO(\eps^2) 
= \OO(\eps^2).
\end{equation*}
This shows that indeed $\P[ \int_0^T \ind_{N_t=0 } \, \ind_{|c(\argvv_t)|> \delta } \, \de t = 0]=1$
for any $\delta>0$. The claim follows from this and Lemma~\ref{lm: beta in a}.
\end{proof}

\section{The proof of the main theorem on convergence of FK Ising loop ensemble to 
\texorpdfstring{CLE$(16/3)$}{CLE(16/3)}}
\label{sec: main proof}

\begin{proof}[The proof of Theorem~\ref{thm: main thm}]
We can choose convergent subsequences 
by the results of Section~\ref{sec: a priori}.
The convergence holds in the topology specified in Section~\ref{sssec: metrices}.

The sequence of observables converges by the results of Section~\ref{sec: observable limit}.
The convergence is uniform over the class of domains and the scaling limit
of the observable which is a solution of a boundary value problem depends
continuously on the initial segment of the curve.
Consequently, the martingale property extends to the scaling limit
of the observable, as we told in Section~\ref{sec: simple martingales}.

The law of the driving process of a branch
of the exploration tree is uniquely characterized
by the results of Section~\ref{sec: simple martingales}.
Finally, it is then possible to extend the convergence
to the full tree by the results of Section~\ref{sec: a priori}.
The law of the scaling limit of the loop ensemble
is uniquely determined by Theorem~\ref{thm:  tree to loops in the limit}
which is proven below.
\end{proof}

\begin{proof}[The proof of Theorem~\ref{thm:  tree to loops in the limit}]
We can assume that 
$(\rndlpe^\disc,\rndttr^\disc)$ is the almost sure limit of 
$(\rndlpe_{\delta_n}^\disc,\rndttr_{\delta_n}^\disc)$
as $n \to \infty$
by Theorem~\ref{thm: main thm} and the discussion in Section~\ref{ssec: main result a priori}.

The first claim follows from the observation that for any point $z$ in $\disc$ and for 
any ball $B$ centered at $z$ such that $B \subset \C$, the branch to $z$ will make 
a non-trivial loop inside $B$ disconnecting $z$ from the boundary. 
Any of the starting points of ``excursions'' after the disconnection,
is one of the points $x_{j}$. By an excursion, we mean a segment $\rndbran([s,t])$
such that if we denote the component of that point in $\disc \setminus \rndbran([0,s])$
by $D_s$, then $\rndbran((s,t)) \subset D_s$ and $\rndbran(u) \in \partial D_s$ for both $u=s,t$.

Let us use the following shortcut to prove the last claim of
Theorem~\ref{thm:  tree to loops in the limit}. We assume that we know that any 
finite subtree as in Section~\ref{ssec: a priori trees}
converges to a tree of coupled \thekr{} curves. It follows from the discussion
of Section~\ref{ssec: intro exploration tree} that the target points corresponding
to loops are triple points of corresponding branches in the bulk and 
double points on the boundary. If there would be other other
triple points in the bulk or double points on the boundary, there is, with high probability
(depending on the parameter $\eta$ in Section~\ref{ssec: a priori trees}),
a branch in the finite tree such that
that point is on that branch. This is in contradiction with the properties
of \thekr{} curves.
\end{proof}


\section*{Acknowledgements}

AK was supported by the Academy of Finland.
SS was supported by the ERC AG COMPASP, the NCCR SwissMAP, the Swiss NSF, and the Russian Science Foundation.


\appendix


\section{Auxiliary results on Loewner evolutions}\label{sec: appendix loewner}

For a Loewner chain $(K_t)_{t \in [0,T]}$ in $\half$ driven by $(U_t)_{t \in [0,T]}$,
let $(V_t)_{t \in [0,T]}$ be the function defined by
\begin{equation*}
V_t = g_t( \sup( K_t \cap \R ) ) .
\end{equation*}
The point $V_t$ is thus the image of the rightmost point on the hull under the conformal map $g_t$.
The following lemma can be extracted from the proof of Proposition~3.12 in \cite{Sheffield:2009}.

\begin{lemma}\label{lm: leb u equals v}
For any $(U_t)_{t \in [0,T]}$ and $(V_t)_{t \in [0,T]}$ as above,
it holds that $\int_0^T \ind_{V_t=U_t} \de t = 0$. 
\end{lemma}

\begin{proposition}\label{prop: rightmost point and increasing process}
Suppose that $\int_0^T \frac{\de t}{V_t - U_t} < \infty$.
There exists a non-decreasing function $(A_t)_{t \in [0,T]}$ with $A_0=0$ such that
\begin{equation}\label{eq: prop rightmost point and increasing process}
V_t = V_0 + \int_0^t \frac{2 \de s}{V_s - U_s} + A_t .
\end{equation}
\end{proposition}

\begin{proof}
Let $\eps>0$,
$J_t^\eps = \min \{ k \eps > \sup( K_t \cap \R ) \,:\, k \in \Z \}$ and $\tilde V_t^\eps = g_t(J_t^\eps)$.
Then it follows from monotonicity of $g_t$ that $V_t \leq V_t^\eps \leq V_t + \eps$.
By the Loewner equation we can write
\begin{equation}\label{eq: prop rightmost point and increasing process auxiliary process}
\tilde V_t^\eps = \tilde V_0^\eps + \int_0^t \frac{2 \de s}{\tilde V_s^\eps - U_s} + \sum_{s \leq t} \xi^\eps_s
\end{equation}
where $\xi^\eps_s \in [0,\eps]$ and $\xi^\eps_s \neq 0$ only when $\tilde V_u^\eps - U_u$ hits zero as $u \nearrow s$.
Thus the sum on the right-hand side is a finite sum. By Lemma~\ref{lm: leb u equals v}  
and Lebesgue's dominated convergence theorem
the middle term in \eqref{eq: prop rightmost point and increasing process auxiliary process}
converges to the integral in \eqref{eq: prop rightmost point and increasing process}.
Since also $\tilde V_t^\eps$ converges uniformly to $V_t$, it holds that
$\sum_{s \leq t} \xi^\eps_s$ converges uniformly to some continuous function $A_t$ as $\eps \to 0$.
It follows that $A_t$ is non-decreasing and $A_0=0$. The claim follows.
\end{proof}

\begin{remark}
Notice that in fact, $\int_0^T \frac{\de t}{V_t - U_t} < \infty$ always. Namely,
\begin{align*}
\int_0^T \frac{\de t}{V_t - U_t + \eps} &\leq \int_0^T \frac{\de t}{\tilde V_t^\eps - U_t} \leq
\frac{1}{2} \left(\tilde V_T^\eps - \tilde V_0^\eps  \right) \\
 &\leq \frac{1}{2} \left( V_T -  V_0 +\eps  \right) .
\end{align*}
Thus $\int_0^T \frac{\de t}{V_t - U_t} < \infty$ follows from Lebesgue's monotone convergence theorem.
\end{remark}

\section{Auxiliary result on martingales}\label{sec: appendix martingale}

We suppose that 
$\P_n$ is a sequence of probability measures on a metric space which converges weakly to $\P$ and
$(M_n(t))_{t \in \{0,1,2,\ldots\}}$ is a sequence of discrete time martingales 
with respect to $\P_n$ and
$\psi_n: [0,\infty) \to \{0,1,2,\ldots\}$ is a sequence of non-decreasing and onto random functions
such that $\psi_n(t)$ is a discrete stopping time for each $n$ and $t$.
Let $(M(t))_{t \in [0,\infty)}$ be a process. Assume also that $M_n$ and $M$ are bounded.

Suppose that for each $\eps>0$ and each $T>0$, there exists $n_0$ and a sequence of events $E_n$ such that
$\P( E_n) > 1 -\eps$ and
\begin{equation}\label{ie: martingale prop and unif conv}
\sup_{E_n} \sup_{t \in [0,T]} \left| M (t) - M_n (\psi_n(t)) \right| \leq \eps .
\end{equation}

\begin{lemma}\label{lem: martingale convergence}
Under the above assumptions, $(M(t))_{t \in [0,\infty)}$ is a martingale.
\end{lemma}

\begin{proof}
Let $s<t$ and let $f$ be any continuous, bounded random variable
which is measurable with respect to $\F_s$. 
Then $\E_n (M_n(\psi_n(t))\, f) = \E_n (M_n(\psi_n(s))\, f)$ by the martingale
property of the discrete observable. By the triangle inequality
\begin{align*}
|\E(M(t) f) - \E(M_s f)| \leq & \; |\E(M(t) f) - \E_n (M(t)\, f)| + |\E(M(s) f) - \E_n (M(s)\, f)| \\
  &+ |\E_n  \{[M(t)-M_n(\psi_n(t))] \,f\}| + |\E_n  \{[M(s)-M_n(\psi_n(s))] \,f\}| .
\end{align*}
First and second term tend to zero as $n \to \infty$ by the weak convergence of probability measures.
The third and fourth term also tend to zero by \eqref{ie: martingale prop and unif conv}, since
$|\E_n  \{[M(t)-M_n(\psi_n(t))] \,f\}| 
  \leq 2 \,\E_n ( \ind_{E^c} |f| ) 
       + \sup_E \sup_{t \in [0,T]} \; |M(t)-M_n(\psi_n(t))| \, \E_n(|f|)$.
\end{proof}

\nocite{*}

\bibliographystyle{habbrv}
\bibliography{ci_in_rcm_ii.bib}

\end{document}